\documentclass[10pt,a4paper]{article}
\usepackage[margin=1in]{geometry}
\usepackage{setspace}
\usepackage[utf8]{inputenc}
\usepackage{amsmath,amssymb,hyperref,array,xcolor,multicol,verbatim,mathpazo}
\usepackage[normalem]{ulem}
\usepackage[pdftex]{graphicx}
\usepackage{fullpage}
\usepackage{amsthm}
\newtheorem{definition}{Definition}
\newtheorem{theorem}{Theorem}
\newtheorem{condition}{Market Condition}
\newtheorem{corollary}{Corollary}

\newtheorem{proposition}{Proposition}
\newtheorem{lemma}{Lemma}
\DeclareMathOperator*{\argmin}{arg\,min}
\DeclareMathOperator*{\argmax}{arg\,max}
\DeclareMathOperator*{\E}{\mathbb{E}}
\usepackage[ruled,vlined]{algorithm2e}
\SetKwComment{Comment}{$\triangleright$\ }{}
\usepackage{graphicx}
\usepackage{subcaption}
\usepackage{tikz}
\usepackage{float}
\newcommand*\circled[1]{\tikz[baseline=(char.base)]{
            \node[shape=circle,draw,inner sep=2pt] (char) {#1};}}
\usepackage{titlesec}
\usepackage{hyperref}
\usepackage{accents}
\newcommand{\ubar}[1]{\underaccent{\bar}{#1}}

\titleclass{\subsubsubsection}{straight}[\subsection]

\newcounter{subsubsubsection}[subsubsection]
\renewcommand\thesubsubsubsection{\thesubsubsection.\arabic{subsubsubsection}}

\titleformat{\subsubsubsection}
  {\normalfont\normalsize\bfseries}{\thesubsubsubsection}{1em}{}
\titlespacing*{\subsubsubsection}
{0pt}{3.25ex plus 1ex minus .2ex}{1.5ex plus .2ex}

\makeatletter
\renewcommand\paragraph{\@startsection{paragraph}{5}{\z@}%
  {3.25ex \@plus1ex \@minus.2ex}%
  {-1em}%
  {\normalfont\normalsize\bfseries}}
\renewcommand\subparagraph{\@startsection{subparagraph}{6}{\parindent}%
  {3.25ex \@plus1ex \@minus .2ex}%
  {-1em}%
  {\normalfont\normalsize\bfseries}}
\def\toclevel@subsubsubsection{4}
\def\toclevel@paragraph{5}
\def\toclevel@paragraph{6}
\def\l@subsubsubsection{\@dottedtocline{4}{7em}{4em}}
\def\l@paragraph{\@dottedtocline{5}{10em}{5em}}
\def\l@subparagraph{\@dottedtocline{6}{14em}{6em}}
\makeatother

\makeatletter
\def\namedlabel#1#2{\begingroup
    #2%
    \def\@currentlabel{#2}%
    \phantomsection\label{#1}\endgroup
}
\makeatother

\begin{document}

\title{1-Dimensional Normal Competitive Market Equilibrium}
\author{Thanawat Sornwanee\\\small Stanford University\\
\tt\small tsornwanee@stanford.edu}
\date{Apr 15, 2024}
\maketitle

\begin{abstract}
    We introduce a new microeconomics foundation of a specific type of competitive market equilibrium that can be used to study several markets with information asymmetry such as commodity market, credit market, and insurance market.
\end{abstract}

    \section*{Introduction}

    Competitive market equilibrium is one of the basis of economic analysis and policy analysis. However, the theory behind the applications of competitive market, such as the analysis of tax effects on credit market~\cite{webb,mahoney}, is incomplete and lacks microeconomics foundation. In this paper, we will equip the competitive market equilibrium analysis with a natural and simple underlying microeconomics mechanism of the market. 

    Competitive market equilibrium analysis is normally based on the seminal work ``Existence of an Equilibrium for a Competitive Economy", Arrow \& Debreu, 1954~\cite{arrow1954} with different regularity condition and extension to the case of large number of agents or continuum of agents~\cite{scarf1962analysis,debreu1963theorem,aumann1964}. However, the main problem formulation used in the paper~\cite{arrow1954} does not have a microeconomics foundation given that it requires the equilibrium to be a Walrasian equilibrium, extrinsically forcing demand and supply to equal~\cite{wald1951some,mckenzie1959existence}. The game theoretic formulation of a market used in the paper~\cite{arrow1954} is done by adding a ``fictitious participant who chooses prices". The Nash equilibrium (with pure strategy) of such system will then entail that demand equals to supply. Moreover, this framework, which provides a positive result of efficient allocation, cannot be conveniently extended to the market with information asymmetry.

    In the case with information asymmetry, the paper ``The Market for "Lemons": Quality Uncertainty and the Market Mechanism", Akerlof, 1970~\cite{akerlof1970lemon} suggests the use of Walrasian equilibrium. Since the model focus on the case when there exists only one type of commodity (but with different unobserved quality), the demand and supply system will be characterized via only a single price, allowing the convenient use of the law of supply and demand.\footnote{However, since the demand as a function of price is no longer monotone from the adverse selection effect, there can be multiple Walrasian equilibria.~\cite{wilson1979equilibrium}} The paper ``The Nature of Equilibrium in Markets with Adverse Selection", Wilson~\cite{wilson1980nature} later suggests that other types of market equilibrium should be considered, and give the example a market when both sellers and buyers are active in the price creation and matching with transfer mechanism.\footnote{The mechanism suggested in the paper~\cite{wilson1980nature} is informally the order book auction where the sellers, who have the private information about the product, privately submits the price willing to sell to the market mechanism, and the buyer, who lacks the private information, proposes a price willing to buy. The mechanism will select the price where the transaction occurs to be the current highest price quoted by the buyer if there still exists a unit of commodity quoted by the seller at the price lower.}

    \begin{figure}
      \centering
      \includegraphics[width=0.6\linewidth]{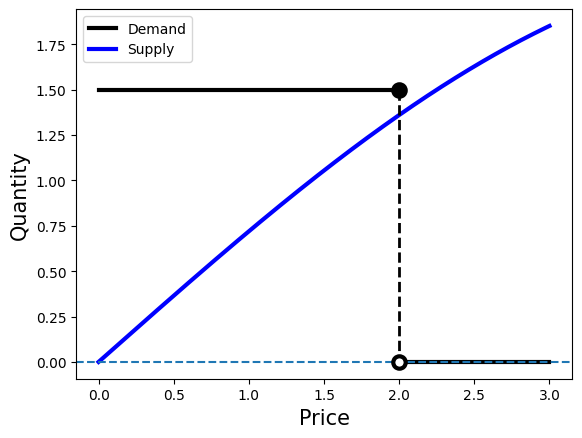}
    \caption{A simple example of double auction based market where Walrasian equilibrium does not exist, because every consumer has the same price willing to buy. Note that, under general equilibrium framework, the equilibrium will exist but it cannot be conveniently explained in the obviously dominant strategy as used in the auction based system used in the paper~\cite{wilson1979equilibrium}.}
    \end{figure}

    Several subsequent papers explored a non-commodity markets with information asymmetry such as credit market~\cite{jaffee1976, stiglitz, webb, riley, mahoney} and insurance market~\cite{rothschild1978equilibrium, mahoney}. However, these papers lack microeconomics foundation on how the price is selected, given that they merely uses the result of either price-taking (Walrasian) equilibrium or zero-profit condition~\cite{hess1984imperfect, jaffee1984imperfect}. 
    
    In this paper, we will reformulate the microeconomics foundation of competitive market equilibrium, addressing game-theoretic formulation, pricing, and asymmetric information simultaneously.

    To do so, in the part~\ref{part:marketasmechanism} (normal market: market as a mechanism), we introduce a game theoretical formulation of a finite ``normal" market with information asymmetry as a sequence of double auctions connected through the right (or the possession) of the ``scarcity", which will be the commodity in a commodity market, the loanable fund in the credit market, and insurance payment in the insurance market. Similar to the double auction approach used in~\cite{wilson1980nature,stiglitz}, in such market, the pricing would be the price quoted by a group of finitely many agents called the mediators, who do not have any private information. 

    In the part~\ref{part:competitive} (competitive market: equilibrium \& analysis), we first consider a competitive market as a natural extension of the finite market introduced in the part~\ref{part:marketasmechanism} into the case when the set of agents reach a continuum limit (section~\ref{section:competitivemarket}).

    We will also introduce ``betrayal-free-collusion-free equilibrium", which is a novel equilibrium notion for a normal form game with continuum of agents suitable for the analysis of competitive market, and show that this equilibrium notion is required even for a standard analysis of a competitive market without production and information asymmetry (section~\ref{section:equilibrium with continuum of agents}).

    We, then, define ``competitive equilibrium" for a ``well-behaved competitive market", which is a continuum limit of a normal market satisfying the 14 conditions. The necessary and sufficient conditions as well as a simple graphical/algorithmic approach to find the collection of all possible competitive equilibria are given (section~\ref{section:competitiveequilibirum}). 

    Although this formalization will be applied to a competitive market satisfying a certain type of conditions, in the part~\ref{part:application}, we will show that the framework and the graphical approach (proposed in the subsection~\ref{subsection:graph}) can be conveniently applied to analyze commodity market and credit market\footnote{with an easy generalization to insurance markets}.

\tableofcontents



\part{Normal Market: Market as a Mechanism}
\label{part:marketasmechanism}

We recall that the application of interest is from commodity market and credit market.

In a commodity market, which will be denoted as a trading market from now on, we assume that there exist finitely many ``producers", ``retails", and ``consumers". The producers possess the right to produce the ``commodity" with possibly some cost. The retails can purchase the produced (or scheduled to be produced) commodity from the producers with ``cash". After successfully acquiring some units of commodity, each retail can then offers them to the consumers in exchange with some ``cash" specified price. For notational simplicity, we will use the retails as a reference point, and denote
\begin{itemize}
    \item The cash amount per unit in the transaction between the producers and the retails as the ``buying price";
    \item The cash amount per unit in the transaction between the retails and the consumers as the ``selling price".
\end{itemize}

A simple observation of this informal model is that, if a transaction occurs, every producer prefers higher ``buying price", and every consumer prefers lower ``selling price".

Similarly, in a credit market, we assume that there exist finitely many ``depositors", ``banks", and ``entrepreneurs". The depositors has some available cash at the current time. The banks can acquire such available cash, and promise the depositors to give their principal money as well as an additional amount of money specified by the ``depositor interest rate" (and the principal amount) back to the depositors in the future time period. After successfully acquiring some deposits, each bank will have some loanable fund that can be offered as a principal amount of loan to the entrepreneurs with such loan contract specifying the maximum\footnote{Due to the limited liability, the entrepreneur does not have to repay the loan principal and the interest back to the bank if the company of such entrepreneur does not have enough valuation.} amount that the entrepreneur is legally obliged to return to the bank via the ``loan interest rate" and loan principal amount. Similarly, we will have that, if a transaction occurs, every depositor prefers higher ``depositor interest rate", and every consumer prefers lower ``loan interest rate".

We will see that the market is formulated in the way that the conventional formulation in general equilibrium literature~\cite{arrow1954, mckenzie1959existence} cannot be applied (even when there is no issue with information asymmetry). In the framework of general equilibrium, the ``buying price" and ``selling price" faced by each agent need to be the same to prevent the arbitrage. However, the arbitrage exists because every market participants can engage in the market as both buyer and seller simultaneously. However, in the economic markets of interest, each participant has a specialized role preventing such arbitrage to happen even when the prices differ. Moreover, we can see that, in a credit market, there does not even exists a buying price or a selling price in a conventional sense. Even when we consider the ``depositor interest rate" to serve as a ``buying price", and the ``loan interest rate" to serve as a ``selling price", it is obvious that the two interest rates are often different, since the bank is obliged to pay the depositor back the principal value and the interest, while the entrepreneur is protected via limited liability and obliged to pay less than the promised amount in the state of the world where the entrepreneur cash is not sufficiently high. Thus, in such market, it is not possible to use the price-taking equilibrium as a market equilibrium given that the two types of prices specifies different interactions (transaction of the scarcity and money) within the same market.

\begin{figure} [H]
    \begin{subfigure}{.5\textwidth}
      \centering
      \includegraphics[width=.8\linewidth]{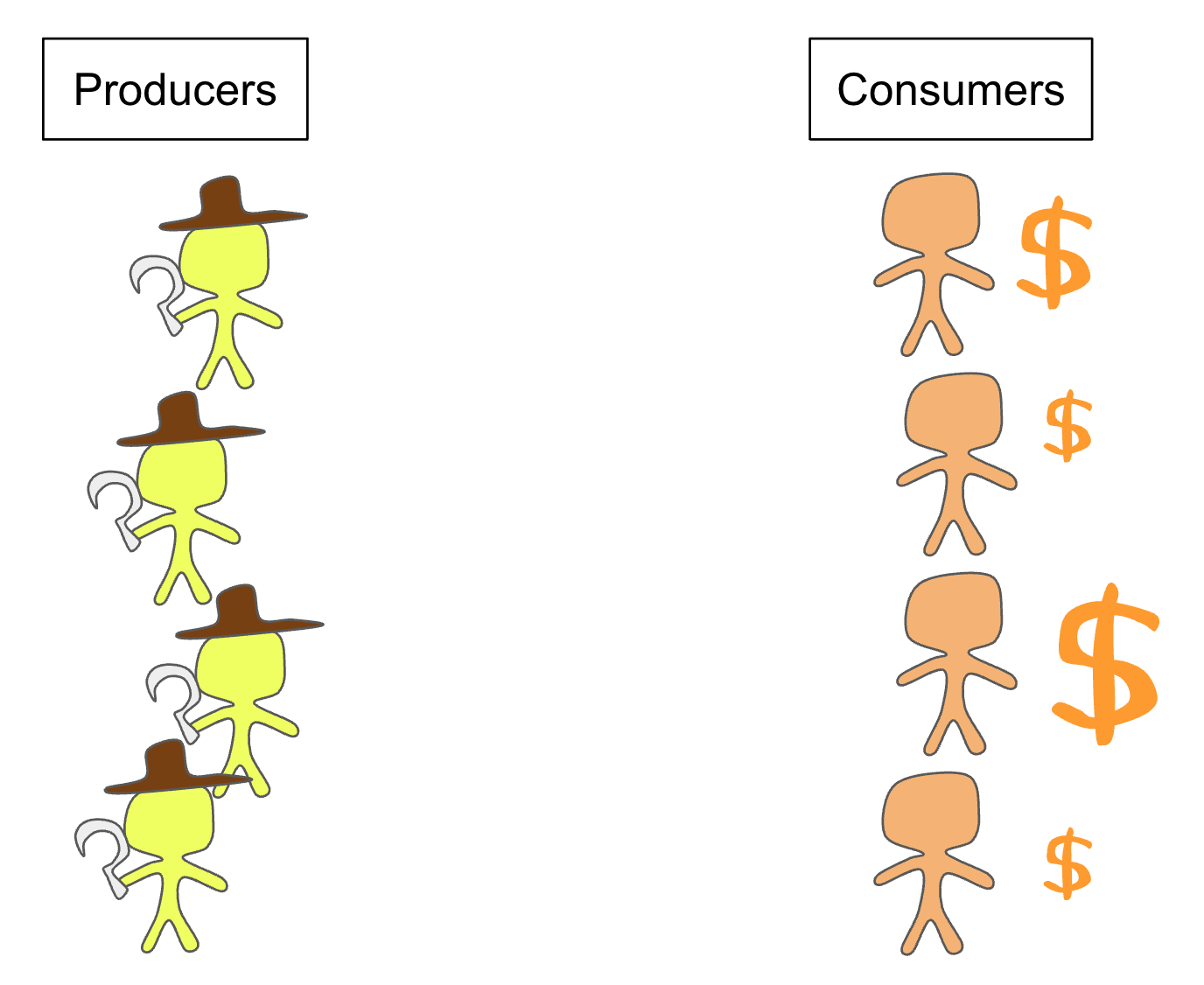}
      \caption{Trading Market as a General Equilibrium Market}
    \end{subfigure}%
    \begin{subfigure}{.5\textwidth}
      \centering
      \includegraphics[width=1\linewidth]{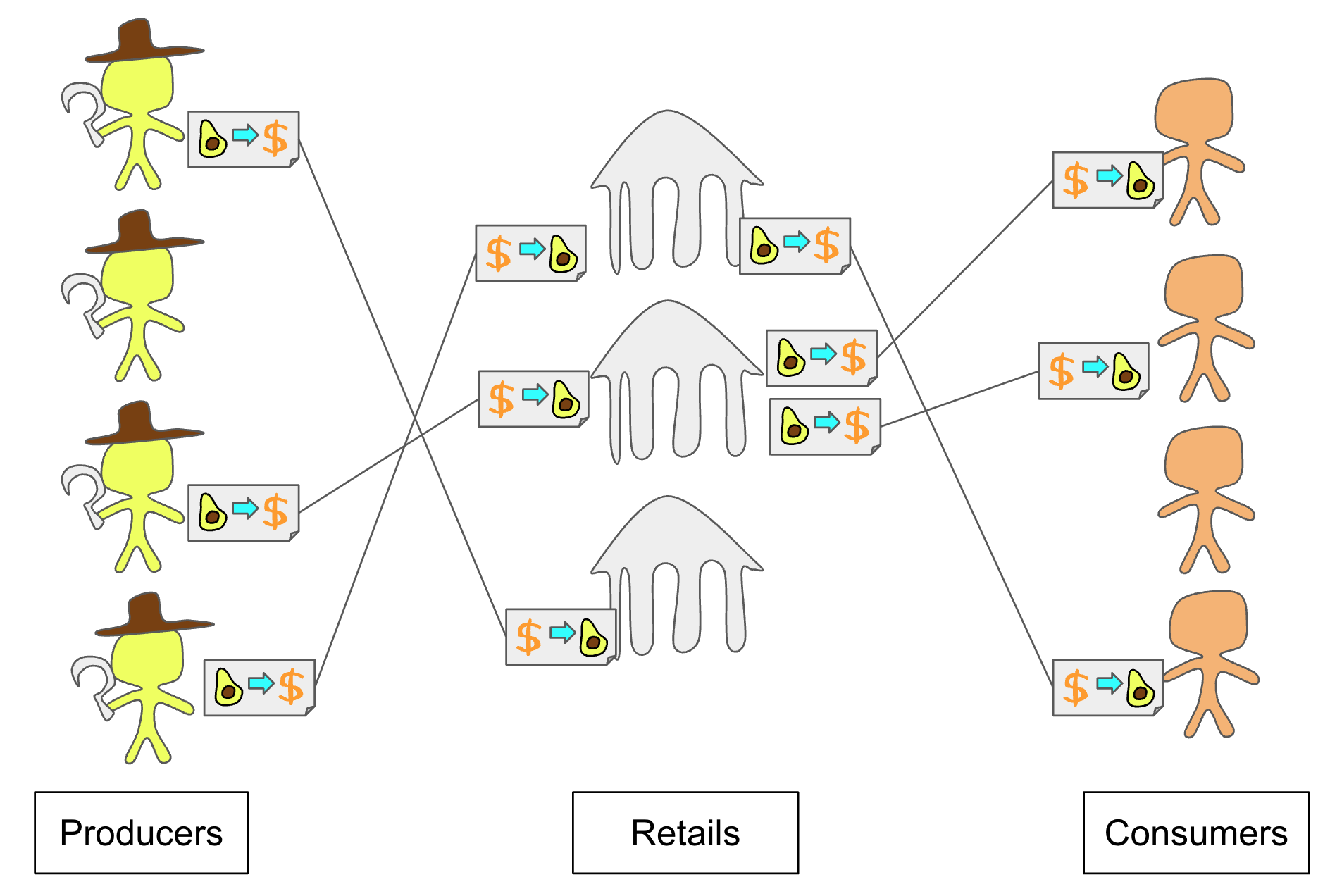}
      \caption{Trading Market as a Normal Market}
    \end{subfigure}
    \caption{If we extend the setting used in the general equilibrium in the paper~\cite{arrow1954}, we will have that we can have 2 groups of agents, ``the producers" who has the production capability, and ``the consumers" who has money. In our formulation as a ``normal market", we introduce intermediaries called ``the retails", which will manage the contract for the exchange of scarcity and money to be enforced after the market is finalized. Although this framework may make the model look more complicated, it turns out (as can be seen in the part~\ref{part:application}) that most character of the normal competitive market for a commodity market will coincide with the traditional characterization of competitive equilibrium dictated via the law of supply and demand~\cite{arrow1954}. The benefit of consider such types of markets with the existence of the intermediaries may not be apparent in the case of a trading market. However, this is crucial in the modelling of a credit market. This is because the scale of the loan and the deposit are usually highly differed. Moreover. the existence of the intermediaries makes the depositor faces an easier decision problem, given that the risk of the entrepreneurial project will be absorbed by the intermediaries (the banks) instead.}
\end{figure}

We will consider a class of markets that has similar characteristics as the formalization of trading market and credit market as a ``normal market".

    Informally, a finite normal market is a market with finitely many ``suppliers", ``mediators", and ``demanders". The suppliers are endowed with the capability to produce ``the scarcity" in the market. Each mediator is endowed with the capability to create ``supply contracts" with a group of suppliers, and are allowed to create ``demand contracts" with a group of demanders if the amount of scarcity acquired by such mediator is enough. Both types of contract will identify the physical transactions from the suppliers to the mediators and from the mediators to the demanders, respectively. After the market is finalized, the contracts generated within the market can be enforced by any market participant.

    The main characterization of a supply contract is
    \begin{enumerate}
        \item The volume, which will specify the amount of the scarcity to be physically produced and transferred to the mediator;
        \item The supply price, which will be such that, given the same amount of volume, the supplier will always prefer a higher supply price (Note that the mediator does not have to always prefers a lower supply price\footnote{Although this can be intuitive, this has been considered in the credit market model in the paper~\cite{webb}. One possible explanation is through the taxation and subsidization done by the government, which is an external entity outside of the market.});
        \item The supplier;
        \item The mediator.
    \end{enumerate}

    Similarly, the main characterization of a demand contract is
    \begin{enumerate}
        \item The volume, which will specify the amount of the scarcity to be physically transferred to the demander;
        \item The demand price, which will be such that, given the same amount of volume, the demander will always prefer a lower demand price (Note that the mediator does not have to always prefers a higher demand price under equilibrium\footnote{This can be seen in the credit market modelling in the paper~\cite{stiglitz}}.);
        \item The demander;
        \item The mediator.
    \end{enumerate}

    These suggest that each contract can only be a bipartisan contract.

    The scarcity is defined with respect to the demand of the demanders and is not available to any demanders and mediators without the transfer from the suppliers. Moreover, each supplier has a private information about one's own utility to produce a specific unit of scarcity under a specific supply contract, and each demander has a private information about one's own utility to receive a specific unit of scarcity under a specific demand contract.

    In different markets, the contracts and the scarcity will be different.

    For example, in a trading market,
    \begin{itemize}
        \item The suppliers are the producers;
        \item The mediators are the retails;
        \item The demanders are the consumers;
        \item The scarcity is the physical commodity;
        \item The supply contract will specify the amount of money that the retail has to pay to the producer in exchange of a certain amount of the physical commodity;
        \item The demand contract will specify the amount of money that the consumer has to pay to the retail in exchange of a certain amount of the physical commodity.
    \end{itemize}
    In a credit market,
    \begin{itemize}
        \item The suppliers are the depositors;
        \item The mediators are the banks;
        \item The demanders are the entrepreneurs;
        \item The scarcity is the loanable fund;
        \item The supply contract is the deposit contract, specifying the amount of money the bank has to pay to the depositor in exchange of a certain amount of deposit;
        \item The demand contract will specify the amount of money that the consumer has to pay to the retail in exchange of a certain amount of the physical commodity.
    \end{itemize}

In this part, we will introduce a formal notion of ``normal market". This includes our formal definition of market as a composition of different elements: ``market participants", ``market mechanism", and ``market context". For a holistic understanding, the ``on-market physics" and ``off-market physics" will also be introduced in order to ground the abstract market into a physical setting.
    
\section{Normal Market}
\label{section:normal}
    
    A normal market is characterized as
    \begin{align*}
        \circled{$\hat{M}$}\left(
        \mathcal{S}, \mathcal{M}, \mathcal{D}, 
        \bar{v}, \mu, \ubar{\mu}, \mathcal{X}
        \right)
        \text{.}
    \end{align*}
    The meaning of each element will be described throughout this section.

    \subsection{Normal Market Participants}

        Each normal market participant can be either a supplier, a demander, or a mediator, and is endowed with agency to perform actions both inside the market and outside of the market as well as the capability to have knowledge and belief.
    
        There exist finite suppliers, finite demanders, and finite mediators.

        Let $\mathcal{S}$ be the collection of the suppliers, $\mathcal{D}$ be the collection of the demanders, and $\mathcal{M}$ be the collection of the mediators. 
        
        The sets $\mathcal{S}$, $\mathcal{D}$, and $\mathcal{M}$ have to be disjoint, suggesting that each agent can only has one role.
        
        The normal market will finalize the physically allowed interaction between 
        \begin{enumerate}
            \item [1.)] pairs of mediator and the supplier;
            \item [2.)] pairs of mediator and demander.
        \end{enumerate}

    \subsection{Normal Market Mechanism}

        The interaction between the mediators and the other participants will be specified by a mechanism. The normal market mechanism will take some voluntary report from the market participants, and returns a physically feasible transaction, which is in the form of the collection of bipartisan ``contract", which is specific to each market mechanism.
        \subsubsection{Normal Market Contract}
            Each contract initiated via the normal market mechanism will be represented as an element of $\mathbb{R}_0^{+} \times \mathbb{R} \times \left(\mathcal{S} \cup \mathcal{D}\right) \times \mathcal{M}$. For any represented contract $c$, if $c_3 \in \mathcal{S}$, we will denote it as a ``supply contract". Otherwise, $c_3 \in \mathcal{D}$, and it will be denoted as a ``demand contract".

            A supply transaction $\mathcal{C}$ is the collection of supply contracts. A demand transaction $\mathcal{Z}$ is the collection of demand contracts.
            
            A market transaction is the tuple of a supply transaction and a demand transaction.

        \subsubsection{On-Market Physics}

            The on-market physics is specified by the production capability function $\bar{v}:\mathcal{S} \to \mathbb{R}^+$, which specify the maximum amount of scarcity that can be promised to be supplied by a supplier. This physics will ensures that the flow of scarcity from the suppliers to the mediators to the demanders is physically feasible.

            \begin{definition}[Physically Feasible Market Transaction]
                A market transaction $\left(\mathcal{C}, \mathcal{Z}\right) \in 
                2^{\left(\mathbb{R}_0^{+} \times \mathbb{R} \times \mathcal{S} \times \mathcal{M}\right)} \times 
                2^{\left(\mathbb{R}_0^{+} \times \mathbb{R} \times \mathcal{D} \times \mathcal{M}\right)}$ is physically feasible if
                \begin{itemize}
                    \item (Single Contract Physics) For any non-mediator participant $p \in \mathcal{S} \cup \mathcal{D}$,
                    \begin{align*}
                        \left\vert\left\{
                        c \in \mathcal{C} \cup \mathcal{Z}: c_3 = p \right\}\right\vert &\in \{0, 1\}
                        \text{;}
                    \end{align*}
                    \item (Supplier Physics) For any supplier $s \in \mathcal{S}$,
                    \begin{align*}
                        \sum_{c \in \mathcal{C}} 
                        \left(
                        c_1\mathbf{1}_{c_3 = s}
                        \right)
                        &\le \bar{v}(s)
                        \text{;}
                    \end{align*}
                    \item (Mediator Physics) For any mediator $m \in \mathcal{M}$,
                    \begin{align*}
                        \sum_{c \in \mathcal{C}} 
                        \left(
                        c_1 \mathbf{1}_{c_4 = m}
                        \right)
                        \ge
                        \sum_{\zeta \in \mathcal{Z}} 
                        \left(
                        \zeta_1 \mathbf{1}_{\zeta_4 = m}
                        \right)
                        \text{;}
                    \end{align*}
                \end{itemize}
            \end{definition}

            We denote the set of physically feasible market transaction
            \begin{align*}
                \mathcal{F}^{\left(\mathcal{S}, \mathcal{M}, \mathcal{D}, \bar{v}\right)}
                \text{.}
            \end{align*}
        
        \subsubsection{Normal Market Mechanism Frame}
        \label{subsection:mechanismframe}
            
            The normal market mechanism frame receives the market participants and the on-market physics as inputs, and returns the normal market mechanism.
            
            The generated normal market mechanism will then receive the voluntary report from the market participants and return the physically feasible allocation.

            Moreover, the mechanism also requires the ``capacity measure" $\mu: 2^\mathcal{M} \to \mathbb{R}^+_0$, and the threshold value $\ubar{\mu} \in \mathbb{R}^+_0$. This measure is a proxy to measure the level of monopoly/oligopoly within the market, which can be explained informally together with the mechanism frame as the following.  

            Each supplier submits a bid for volume willing to contract and lowest supply price willing to take. Each demander submits a bid for volume willing to contract and highest demand price willing to take.

            The mediators bid the supply price. The highest price will be the market supply price. Any supplier is willing to take the market price will get to contract with the volume being one's own declared volume and the supply price being the market price, which is the highest supply price offered by the mediators. 

            If there exists some producers who are indifferent between producing (and cont arcing) and not producing (and not contracting) at the market price, the group of winning mediators can select the target quantity of suppliers who are indifferent with the market price by an additional round of auction. After the target volume is set, the mechanism keeps adding another random indifferent supplier until the target volume is reached. Each supplier who is selected to get a contract will get a contract randomly with the mediators who have won both auctions (price and target volume).

            Each mediator has to place all acquired scarcity into an order book. The lowest price of the order book will be eliminated by matching the scarcity bulk at the price with any mediator whose willing to take demand price is higher \textbf{or equal} at random until the residue at such price cannot be matched or there is no demander who is willing to take such price left. Afterward, the order will be removed from the order book. The order elimination process is continue until the order book is cleared.

            The market can then measure whether the set of the winning mediators (from the supply side auction) has a large enough measure. If so, the market will end.\footnote{For example, if $\mu$ is a counting measure, and $\bar{\mu} = 1$, the market will end if there is at least two mediators with non-zero scarcity.}

            Otherwise, a supply resale will be opened for the non-winning mediators.\footnote{This can be thought as a process to avoid oligopoly or monopoly who may over-acquire the scarcity from the supply side and cannot transfer all scarcity to the demand side.} The resale is to allow the unmatched scarcity to be place on an order book again. The mediator who promises to take up the largest amount of supply and place it on the lowest price (lexicographically) will win the resale auction.

            Note that this informal description of the market mechanism is in line with that informally described in the papers~\cite{wilson1980nature,stiglitz} and is similar to the conventional understanding of competitive market (if the number of agents is increased into the continuum limit). The only point where this specification may not be intuitive is the allowance for resale, as well as the possibility that the resale period can be extended for any arbitrary number of times. 

            In defense of this choice of modelling, we will see that the resale has to be allowed to prevent a monopoly. If a single mediator bids the highest supply price bid, every produced scarcity will be acquired by such mediator, allowing one to effectively become a monopoly in the second stage of the game. To prevent this problem, there are naturally two ways of doing so. The first way is to enforce some measure of monopoly/oligopoly of the market and allow the resale to occur if the distribution of the scarcity is too wasteful. Another solution is to enforce a strict capacity limit for each mediator, ensuring that a single mediator cannot have too much scarcity to be sold to the demanders in the next market stage. This solution will make the finite market well-defined, and can be finalized within a finite number of interactions. However, the strict capacity limit will extrinsically force the number of winning mediators to be high, making it not possible for a threat of inefficient pricing (under a conventional understanding of competitive market) to remove such occurrence from an equilibrium. Both examples will be more closely and formally investigated in the definition of competitive equilibrium given in the section~\ref{section:competitiveequilibirum}.
                
            We consider the following specific mechanism frame.
            \begin{itemize}
            
                \item (Supply Introduction) At time $t = 0$, each supplier $s \in \mathcal{S}$ submits a single blind bid $\left(v_s, \rho_s \right) \in \mathbb{R}^{+} \times \mathbb{R}$ where $v_s$ and $\rho_s$ are denoted as ``volume" and ``supply price", respectively. The volume of the bid has to be such that
                \begin{align*}
                    v_s \le \bar{v}(s) \text{.}
                \end{align*}
                
                \item (Demand Introduction) At time $t=1$, each demander $d \in \mathcal{D}$ submits a single blind bid $\left(v_d, r_d\right) \in \mathbb{R}^{+} \times \mathbb{R}$ where $v_d$ and $r_d$ are denoted as ``volume" and ``demand price", respectively.
                
                \item (Supply Price Bidding) At time $t=2$, each mediator $m \in \mathcal{M}$ submits a counter-supply-price bid $\tilde{\rho}_m \in \mathbb{R} \cup \{-\infty\}$.

                \item (Supply Price Auction) At time $t = 3$, the mechanism computes
                \begin{enumerate}
                    \item [1.)] the  highest counter-supply-price bid 
                    \begin{align*}
                        \bar{\rho} = \max_{m \in \mathcal{M}} \tilde{\rho}_m
                        \text{;}
                    \end{align*}
                    \item [2.)] the intermediately winning mediators set
                    \begin{align*}
                        \tilde{\mathcal{M}} = \argmax_{m \in \mathcal{M}} \tilde{\rho}_m
                        \text{;}
                    \end{align*}
                    \item [3.)] the strictly winning suppliers set
                    \begin{align*}
                        \tilde{\mathcal{S}} = \left\{s\in\mathcal{S}:\rho_s < \bar{\rho}
                        \right\}
                        \text{;}
                    \end{align*}
                    \item [4.)] the weakly winning suppliers set
                    \begin{align*}
                        \tilde{\tilde{\mathcal{S}}} = \left\{s \in \mathcal{S}:\rho_s = \bar{\rho}\right\}
                        \text{;}
                    \end{align*}
                    \item [5.)] the strictly winning supply volume
                    \begin{align*}
                        \tilde{v}
                        =
                        \sum_{s \in \tilde{\mathcal{S}}}
                        v_s
                        \text{;}
                    \end{align*}
                    \item [6.)] the weakly winning supply volume combinations
                    \begin{align*}
                        \tilde{\tilde{V}}
                        =
                        \left\{\sum_{s \in S}
                        v_s
                        \right\}_{S \subseteq \tilde{\tilde{\mathcal{S}}}}
                        \text{.}
                    \end{align*}
                \end{enumerate}

                If the highest counter-supply-price bid  $\bar{\rho} = -\infty$, the market will be finalized and return $\left(\mathcal{C}, \mathcal{Z}\right) = \left(\emptyset, \emptyset\right)$.

                Otherwise, the mechanism discloses
                \begin{align*}
                    \left(\tilde{v}, \tilde{\tilde{V}}\right)
                \end{align*}
                to every mediator $m \in \mathcal{M}$.

                \item (Supply Residual Bidding) At time $t=4$, each intermediately winning mediator $m \in \tilde{\mathcal{M}}$ submits a residual-volume bid  $\bar{v}_m \in \tilde{\tilde{V}}$.

                \item (Supply Residual Auction) At time $t=5$, the mechanism randomizes a bijective function $f: \mathbb{N} \cap \left(-\infty, \left\vert\tilde{\tilde{S}}\right\vert\right] \to \tilde{\tilde{S}}$ uniformly at random from a collection of every bijective function with domain $\mathbb{N} \cap \left(-\infty, \left\vert\tilde{\tilde{S}}\right\vert\right]$ and range $\tilde{\tilde{S}}$, and computes
                \begin{enumerate}
                    \item [1.)] the highest feasible counter-volume bid
                    \begin{align*}
                        \bar{\bar{v}}
                        = \max_{m \in \tilde{\mathcal{M}}} \bar{v}_m
                        \text{;}
                    \end{align*}
                    \item [2.)] the winning mediators set
                    \begin{align*}
                        \tilde{\tilde{\mathcal{M}}} = 
                        \max_{m \in \tilde{\mathcal{M}}} \bar{v}_m
                        \text{;}
                    \end{align*}
                    \item [3.)] the winning suppliers set
                    \begin{align*}
                        \hat{\mathcal{S}} = \tilde{\mathcal{S}} \cup \left\{f(i)\right\}_{i=1}^{\min\left(\left\{j \in \mathbb{N}_0 \cap \left\vert\tilde{\tilde{S}}\right\vert: \sum_{k=1}^{j} v_{f(k)} \ge \bar{\bar{v}} \right\}\right)}
                        \text{.}
                    \end{align*}
                \end{enumerate}

                The mechanism generates a supply transaction
                \begin{align*}
                    \mathcal{C}^{(0)} = 
                    \left\{
                        \left(
                            v_s, 
                            \bar{\rho}^{\left(\hat{y}(s)\right)}, 
                            s, 
                            m_s
                        \right)
                    \right\}_{s \in \hat{\mathcal{S}}}
                \end{align*}
                where, for each $s \in \hat{\mathcal{S}}$,
                \begin{align*}
                    m_s 
                    \overset{\text{iid}}{\sim} \text{Unif}\left( \hat{\mathcal{M}}^{\left(\hat{y}(s)\right)}
                    \right)
                    \text{.}
                \end{align*}
                
                The mechanism discloses a scarcity quota function $s: \mathcal{M} \to \mathbb{R}_0^+$ such that, for any $m \in \mathcal{M}$,
                \begin{align*}
                    s(m) = 
                    \sum_{c \in \mathcal{C}^{(0)}}
                    \left( c_1
                    \mathbf{1}_{c_4=m} 
                    \right)
                    \text{.}
                \end{align*}

                \item (Demand Bulk Bidding) At time $t=6$, each mediator $m \in \mathcal{M}$ submits a counter-demand-price bid  $\tilde{r}_m \in \mathbb{R}$.

                \item (Demand Bulk Auction) At time $t=7$, the mechanism computes 
                \begin{enumerate}
                    \item [1.)] the initial order book
                    \begin{align*}
                        O_0^{(0)} = \left\{\left(s(m),\tilde{r}_m,m\right)\right\}_{m \in \mathcal{M}}
                        \text{;}
                    \end{align*}
                    \item [3.)] the initial contracts book
                    \begin{align*}
                        \mathcal{Z}^{(0)}_0 = \emptyset
                        \text{.}
                    \end{align*}
                \end{enumerate}

                For $i \in \mathbb{N} \cap \left(-\infty, \bar{i}^{(0)}\right]$, at time $t=7.6-2^{-i}$, 
                the mechanism randomizes the $i^{\text{th}}$-iteration winning order
                \begin{align*}
                    o^{(0)}_{i} \sim \text{Unif}
                    \left(\argmin_{o \in O^{(0)}_{i-1}} o_2\right)
                    \text{,}
                \end{align*}
                and computes the $i^{\text{th}}$-iteration feasible matchings set
                \begin{align*}
                    \tilde{D}^{(0)}_{i} = \left\{d \in \mathcal{D}
                    - 
                    \left\{\zeta_3\right\}_{\zeta \in \mathcal{Z}^{(0)}_{i-1}}
                    : v_d \le \left(o^{(0)}_{i}\right)_1, r_d \ge \left(o^{(0)}_{i}\right)_2 \right\}
                    \text{.}
                \end{align*}

                \begin{itemize}
                    \item If $\tilde{D}^{(0)}_{i} = \emptyset$, then the mechanism computes 
                    \begin{enumerate}
                        \item [1.)] the $i^{\text{th}}$-iteration order book
                        \begin{align*}
                            O^{(0)}_i = O^{(0)}_{i-1} - \left\{o^{(0)}_{i}\right\}
                            \text{;}
                        \end{align*}
                        \item [2.)] the $i^{\text{th}}$-iteration contracts book
                        \begin{align*}
                            \mathcal{Z}^{(0)}_{i} = \mathcal{Z}^{(0)}_{i-1}
                            \text{.}
                        \end{align*}
                    \end{enumerate}
                    
                    \item If $\tilde{D}^{(0)}_{i} \ne \emptyset$, then the mechanism randomizes the signal-specific $i^{\text{th}}$-iteration match
                    \begin{align*}
                        d^{(0)}_{i} \sim \text{Unif}
                        \left(
                        \tilde{D}^{(0)}_{i}
                        \right)
                        \text{,}
                    \end{align*}
                    and computes 
                    \begin{enumerate}
                        \item [1.)] the $i^{\text{th}}$-iteration order book
                        \begin{align*}
                            O^{(0)}_i = \left(O^{(0)}_{i-1} - \left\{o^{(0)}_{i}\right\}\right)
                            \cup \left\{\left(\left(o^{(0)}_{i}\right)_1 - v_{d^{(0)}_{i}}, \left(o^{(0)}_{i}\right)_2, \left(o^{(0)}_{i}\right)_3\right)\right\}
                            \text{;}
                        \end{align*}
                        \item [2.)] the $i^{\text{th}}$-iteration contracts book
                        \begin{align*}
                            \mathcal{Z}^{(0)}_{i} = 
                            \mathcal{Z}^{(0)}_{i-1}
                            \sqcup \left\{\left(
                            v_{d^{(0)}_i}, \left(o^{(0)}_i\right)_2,
                            d^{(0)}_i,
                            \left(o^{(0)}_i\right)_3
                            \right)\right\}
                            \text{.}
                        \end{align*}
                    \end{enumerate}
                \end{itemize}

                If $O^{(0)}_i = \emptyset$, then $\bar{i}^{(0)} = i$.

                At time $t=7.6$, the mechanism computes
                \begin{enumerate}
                    \item [1.)] the demand transaction
                    \begin{align*}
                        \mathcal{Z}^{(0)} = 
                        \mathcal{Z}^{(0)}_{\bar{i}^{(0)}}
                        \text{;}
                    \end{align*}
                    \item [2.)] the maximum resale supply volume 
                    \begin{align*}
                        R^{(0)}
                        =
                        \sum_{m \in \mathcal{M}} \max\left(
                        \left\{
                            \sum_{c \in C} 
                            c_1 \mathbf{1}_{c_4 = m}
                        \right\}_{C \subseteq \mathcal{C}^{0}}
                        \cap \left(-\infty, s(m) - \sum_{\zeta \in \mathcal{Z}^{(0)}}
                            \zeta_1 \mathbf{1}_{\zeta_4 = m}\right]
                        \right)
                        \text{;}
                    \end{align*}
                \end{enumerate}

                The mechanism discloses $\left(\mathcal{Z}^{(0)}, R^{(0)}\right)$ to every mediator $m \in \mathcal{M}$.

                If $\mu\left(\tilde{\tilde{\tilde{\mathcal{M}}}}\right) > \ubar{\mu}$ or $\tilde{\tilde{\tilde{\mathcal{M}}}} = \mathcal{M}$, the market is finalized and returns
                \begin{align*}
                    \left(\mathcal{C}, \mathcal{Z}\right) 
                    =
                    \left(\mathcal{C}^{(0)}, \mathcal{Z}^{(0)}\right)
                    \text{.}
                \end{align*}
                
                \item (Resale Bidding) For each $j \in \mathbb{N}\cap\left(-\infty,\bar{j}\right]$, at time $t=9-2^{1-j}$, each mediator $m \in \mathcal{M}$ submits a $(j+1)^{\text{th}}$ resale supply volume bid $\left(\tilde{v}^{(j)}_m, \tilde{r}^{(j)}_m\right) \in \left[0, R^{(0)}\right] \times \mathbb{R}$;

                \item (Demand Resale Auction) For each $j \in \mathbb{N}\cap\left(-\infty,\bar{j}\right]$, at time $t=9-(1.5)2^{-j}$, the market computes the $(j+1)^{\text{th}}$ resale supply winning volume
                \begin{align*}
                    \bar{v}^{(j)} = \max_{m \in \mathcal{M} - \tilde{\tilde{\tilde{\mathcal{M}}}}} \left\{ \tilde{v}^{(j)}_m \right\}
                    \text{.}
                \end{align*}
                
                \begin{itemize}
                    \item If $\bar{v}^{(j)} = 0$, then $\bar{j} = j$;
                    \item If $\bar{v}^{(j)} > 0$, then the mechanism randomizes the $(j+1)^{\text{th}}$ winning mediator 
                    \begin{align*}
                        \bar{m}^{(j)} \sim \text{Unif}\left(
                        \argmin_{m 
                            \in 
                            \argmax_{m' \in \mathcal{M} - \tilde{\tilde{\tilde{\mathcal{M}}}}} \left\{ \tilde{v}^{(j)}_{m'} \right\}
                            }
                            \tilde{r}^{(j)}_m
                        \right)
                        \text{,}
                    \end{align*}
                    and computes the incumbent $(j+1)^{\text{th}}$ recontracted contracts package
                    \begin{align*}
                        C^{(j)}_{0} = \mathcal{C}^{(j-1)}
                        \text{.}
                    \end{align*}

                    For $i \in \mathbb{N} \cap \left(-\infty, I_j\right]$, at time $t = 9 - (0.3)2^{-j}(4+2^{-i})$, the mechanism computes
                    \begin{align*}
                        \tilde{C}^{(j)}_i = \left\{c \in C^{(j)}_{i-1}: 
                        \sum_{\zeta \in \mathcal{Z}^{(j-1)}} \zeta_1 \mathbf{1}_{\zeta_4 = c_4} \le 
                        \sum_{c \in \mathcal{C}^{(j)}_{i-1}}
                        \left( c_1
                        \mathbf{1}_{c_4=m} 
                        \right) - c_1
                        \right\}
                        \text{.}
                    \end{align*}

                    \begin{itemize}
                        \item If $\tilde{C}^{(j)}_i \ne \emptyset$, the mechanism randomizes a recontracted contract
                        \begin{align*}
                            \tilde{c}^{(j)}_i \sim \text{Unif}\left(\tilde{C}^{(j)}_i\right)
                            \text{,}
                        \end{align*}
                        and computes
                        \begin{align*}
                            C^{(j)}_{i}
                            =
                            \left(C^{(j)}_{i-1} - \left\{ \tilde{c}^{(j)}_i \right\}\right) \cup \left\{
                            \left(c_1, c_2, c_3, \bar{m}^{(j)}\right)
                            \right\}
                            \text{.}
                        \end{align*}
                        
                        If $\sum_{c \in C^{(j)}_{i} - \mathcal{C}^{(j-1)}} c_1 \mathbf{1}_{c_4 = \bar{m}^{(j)}}$, then $I_j = i$. 
                        
                        \item If $\tilde{C}^{(j)}_i = \emptyset$, then $I_j = i$, and the mechanism computes
                        \begin{align*}
                            C^{(j)}_{i} = C^{(j)}_{i-1}
                            \text{.}
                        \end{align*}
                    \end{itemize}

                    At time $t = 9 - (0.3)2^{-j}(4+2^{-1-I_j})$, the mechanism computes 
                    \begin{enumerate}
                        \item [1.)] the $(j+1)^{\text{th}}$ recontracted contracts package
                        \begin{align*}
                            \mathcal{C}^{(j)} = C^{(j)}_{I_j }
                            \text{;}
                        \end{align*}
                        \item [2.)] the $(j+1)^{\text{th}}$ initial order book
                        \begin{align*}
                            O^{(j)}_0 = 
                            \left\{\left(
                            \sum_{c \in \mathcal{C}^{(j)} - \mathcal{C}^{(j-1)}} \left(c_1\mathbf{1}_{c_4 = \bar{m}^{(j)}}\right), \tilde{\rho}^{(j)}_{\bar{m}^{(j)}}, \bar{m}^{(j)}\right)\right\}
                            \text{;}
                        \end{align*}
                        \item [3.)] the $(j+1)^{\text{th}}$ initial contracts book
                        \begin{align*}
                            \mathcal{Z}^{(j)}_0 = \mathcal{Z}^{(j-1)}_{\bar{i}^{(j-1)}}
                            \text{;}
                        \end{align*}
                    \end{enumerate}
    
                    For each $i \in \mathbb{N} \cap \left(-\infty,\bar{i}^{(j)}\right]$, at time $t = 9 - (0.3)2^{-j}(4+2^{-i-1-I_j})$, the mechanism randomizes the $(j+1)^{\text{th}}$ $i^{\text{th}}$-iteration winning order
                    \begin{align*}
                        o^{(j)}_{i} \sim \text{Unif}\left(\argmin_{o \in O^{(j)}_{i-1}} o_2\right)
                        \text{,}
                    \end{align*}
                    and computes the $(j+1)^{\text{th}}$ $i^{\text{th}}$-iteration feasible matchings set
                    \begin{align*}
                        \tilde{D}^{(j)}_{i} &= \left\{d \in \mathcal{D}^{(j)}
                        - 
                        \left\{\zeta_3\right\}_{\zeta \in \mathcal{Z}^{(j)}_{i-1}}
                        : v_d \le \left(o^{(j)}_{i}\right)_1, r_d \ge \left(o^{(j)}_{i}\right)_2 \right\}
                        \cup 
                        \left\{\zeta_3\right\}_{\zeta \in
                            \left\{\zeta' \in \mathcal{Z}^{(j)}_{i-1}: \zeta'_1 \le  \left(o^{(0)}_{i}\right)_1, \zeta'_2 > \left(o^{(j)}_{i}\right)_2\right\}
                        }
                        \text{.}
                    \end{align*}
    
                    \begin{itemize}
                        \item If $\tilde{D}^{(j)}_{i} = \emptyset$, then the mechanism computes 
                        \begin{enumerate}
                            \item [1.)] the $(j+1)^{\text{th}}$ $i^{\text{th}}$-iteration order book
                            \begin{align*}
                                O^{(j)}_i = O^{(j)}_{i-1} - \left\{o^{(j)}_{i}\right\}
                                \text{;}
                            \end{align*}
                            \item [2.)] the $(j+1)^{\text{th}}$ $i^{\text{th}}$-iteration contracts book
                            \begin{align*}
                                \mathcal{Z}^{(j)}_{i} = \mathcal{Z}^{(j)}_{i-1}
                                \text{.}
                            \end{align*}
                        \end{enumerate}
                        
                        \item If $\tilde{D}^{(j)}_{i} \ne \emptyset$, then the mechanism randomizes the $(j+1)^{\text{th}}$ $i^{\text{th}}$-iteration match
                        \begin{align*}
                            d^{(j)}_{i} \sim \text{Unif}
                            \left(
                            \tilde{D}^{(j)}_{i}
                            \right)
                            \text{,}
                        \end{align*}
                        and computes 
                        \begin{enumerate}
                            \item [1.)] the $(j+1)^{\text{th}}$ $i^{\text{th}}$-iteration order book
                            \begin{align*}
                                O^{(j)}_i = \left(O^{(j)}_{i-1} - \left\{o^{(j)}_{i}\right\}\right)
                                \cup \left\{\left(\left(o^{(j)}_{i}\right)_1 - v_{d^{(j)}_{i}}, \left(o^{(j)}_{i}\right)_2, \left(o^{(j)}_{i}\right)_3\right)\right\}
                                \text{.}
                            \end{align*}
                            \item [2.)] the $(j+1)^{\text{th}}$ $i^{\text{th}}$-iteration contracts book
                            \begin{align*}
                                \mathcal{Z}^{(j)}_{i} = 
                                \mathcal{Z}^{(j)}_{i-1}
                                \sqcup \left\{\left(
                                v_{d^{(j)}_i}, \left(o^{(j)}_i\right)_2,
                                d^{(j)}_i,
                                \left(o^{(j)}_i\right)_3
                                \right)\right\}
                                \text{.}
                            \end{align*}
                        \end{enumerate}
                    \end{itemize}
    
                    If $O^{(j)}_i = \emptyset$, then $\bar{i}^{(j)} = i$.
    
                    At time $t=9-(1.1)2^{-j}$, the mechanism computes
                    \begin{enumerate}
                        \item [1.)] the $(j+1)^{\text{th}}$ demand transaction
                        \begin{align*}
                            \mathcal{Z}^{(j)} = \mathcal{Z}^{(j)}_{\bar{i}^{(j)}}
                            \text{.}
                        \end{align*}
                        \item [2.)] the $(j+1)^{\text{th}}$ maximum resale supply volume 
                        \begin{align*}
                            R^{(j)}
                            =
                            \sum_{m \in \mathcal{M}} \max\left(
                            \left\{
                                \sum_{c \in C} 
                                c_1 \mathbf{1}_{c_4 = m}
                            \right\}_{C \subseteq \mathcal{C}^{j}}
                            \cap \left(-\infty, \sum_{c \in \mathcal{C}^{(j)}} c_1 \mathbf{1}_{c_4 = m}- \sum_{\zeta \in \mathcal{Z}^{(j)}}
                                \zeta_1 \mathbf{1}_{\zeta_4 = m}\right]
                            \right)
                            \text{;}
                        \end{align*}
                    \end{enumerate}
    
                    The mechanism discloses $\left(\mathcal{Z}^{(j)}, R^{(j)}\right)$ to every mediator $m \in \mathcal{M}$.
                    
                    \item If $\bar{S}^{(j)} = 0$, then $\bar{j} = j$, and the mechanism returns
                    \begin{align*}
                        \left(\mathcal{C}, \mathcal{Z}\right)
                        =
                        \left(\mathcal{C}^{\left(\bar{j}-1\right)}, \mathcal{Z}^{\left(\bar{j}-1\right)}\right)
                        \text{.}
                    \end{align*}
                \end{itemize}

            \end{itemize}  

            We denote this market frame as a function $M$, such that, for any suppliers set $\mathcal{S}$, any mediators set $\mathcal{M}$, any demanders set $\mathcal{D}$, any production capability function $\bar{v}:\mathcal{S} \to \mathbb{R}_0^+$, any capacity measure $\mu: 2^{\mathcal{M}} \to \mathbb{R}^+_0$, any capacity threshold $\ubar{\mu}\in  \mathbb{R}^+_0$
            the normal market mechanism is
            \begin{align*}
                M\left(\mathcal{S}, \mathcal{M}, \mathcal{D}, \bar{v}, \mu, \ubar{\mu}\right)
                \text{.}
            \end{align*}

            The market mechanism is known to every market participant $p \in \mathcal{S} \cup \mathcal{M} \cup \mathcal{D}$.

            \begin{figure} [H]
                \begin{subfigure}{1\textwidth}
                  \centering
                  \includegraphics[width=1\linewidth]{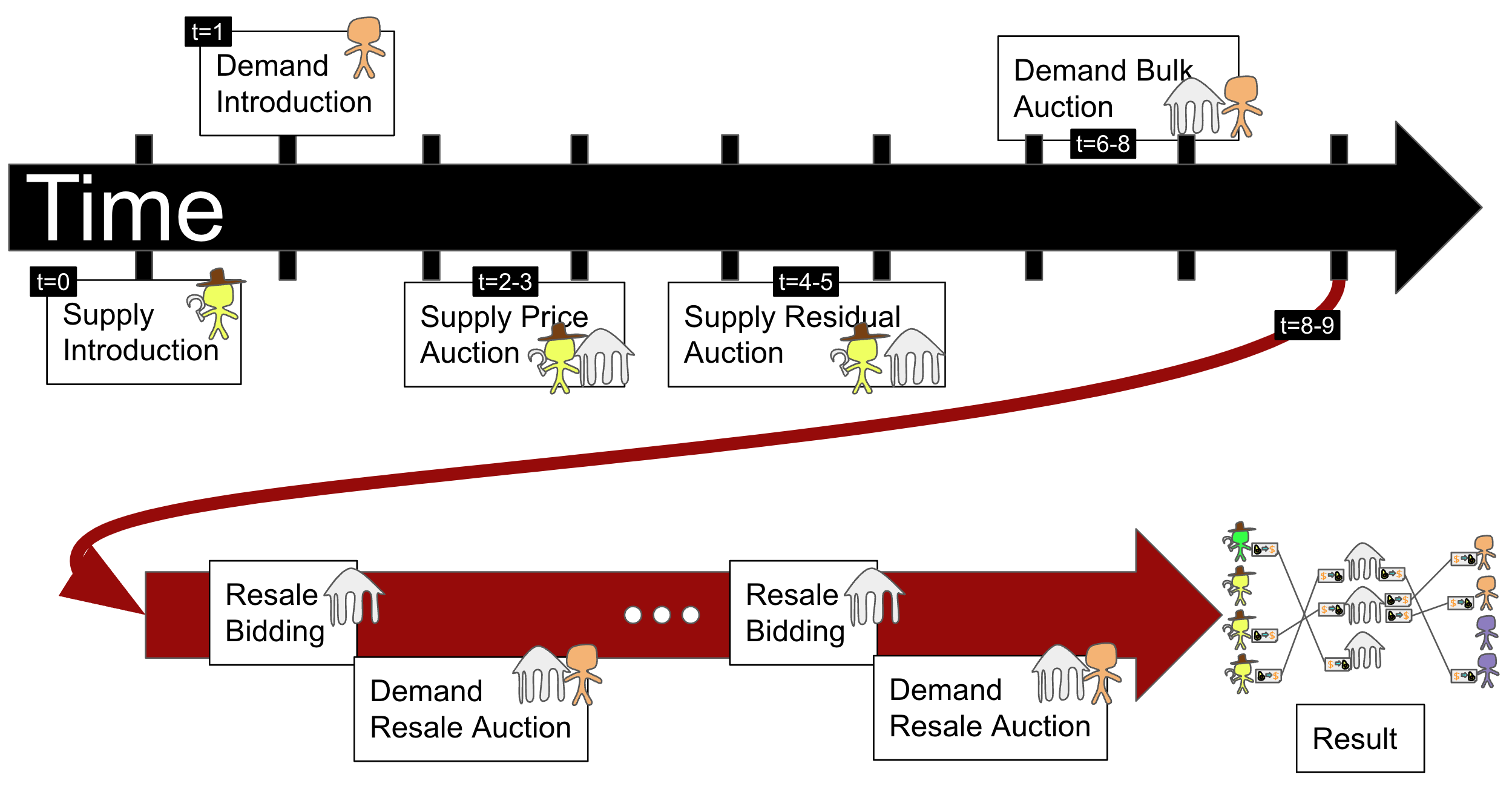}
                  \caption{Timeline of Normal Market Mechanism}
                \end{subfigure}
                \caption{At time $t=0$ and $t=2$ to $5$, the market is said to be in the first stage (the supply market). This stage will specify the market supply price and the set of suppliers who will get a supply contract under such market price. The mechanism governs the supply price auction is done in the first price auction sense, where the mediator who bids the highest supply price will win the auction.
                \\At time $t=1$ and $t=6$ to $8$, the market is said to be in the second stage (the demand market). This stage will specify the demanders who will get a demand contract. However, the demand bulk auction is done in a double auction manner (order book auction), so different demanders can get demand contracts under different demand prices.
                \\At time $t=8$ to $9$, if the market is under a monopoly or oligopoly case (under the measure and threshold provided externally), then the resale bidding and demand resale auction will be implemented iteratively. Otherwise, nothing will be done during such time. Eventually, the market will return the collection of demand contracts and the collection of supply contracts.}
            \end{figure}

    \subsection{Normal Market Context}

        The normal market context $\mathcal{X}$ is the one determined 
        \begin{enumerate}
            \item [1.)] The correct knowledge on $\mathcal{S}, \mathcal{M}, \mathcal{D}, \bar{v}$;
            \item [2.)] The correct knowledge on the physical enforcement of the market transaction $\left(\mathcal{C}, \mathcal{Z}\right)$;
            \item [3.)] The knowledge and belief structure of the off-market physics.
        \end{enumerate}
        
    \section{Off-Market Physics}

        The off-market physics is the one determined the actions and realization of utility after the market transaction is determined from the market activities. Thus, the structure of the off-market will affect the strategy used in the market of each market participant.
            
        The off-market has a degree of freedom $N \in \mathbb{N}$ and the off-market random variable 
        \begin{align*}
            Z \overset{\Delta}{=} \text{Unif}\left([0,1)^N\right)
            \text{,}
        \end{align*}
        and is independent from any randomization inside the market.

        There exists a set of the total admissible actions $\mathcal{A} \ne \emptyset$, and a admissible actions set $\hat{A}: \mathcal{S}\sqcup \mathcal{M}\sqcup \mathcal{D} \to 2^{\mathcal{A}} - \{\emptyset\}$.

        Each market participant $p \in \mathcal{S}\sqcup \mathcal{M}\sqcup \mathcal{D}$ simultaneously chooses the action 
        \begin{align*}
            a_p \in A_p = \hat{A}\left(p\right)
        \end{align*}
        with the knowledge of the realization of the physically-feasible market transaction $\left(\mathcal{C}, \mathcal{Z}\right)$.

        Since the physical enforcement of the contract is deterministic, we have that 
        There exists a measurable utility function $u:\left(\mathcal{S}\sqcup \mathcal{M}\sqcup \mathcal{D}\right) \times \left(\bigotimes_{p \in \mathcal{S}\sqcup \mathcal{M}\sqcup \mathcal{D}} A_p \right)\times \mathcal{F}^{\left(\mathcal{S}, \mathcal{M}, \mathcal{D}, \bar{v}\right)} \times [0,1)^N \to [-\infty,\infty)$ such that, for each market participant $p \in \mathcal{S}\sqcup \mathcal{M}\sqcup \mathcal{D}$, such participant will receive a realized utility
        \begin{align*}
            u_p
            =
            u
            \left(
                p, \left(a_{p'}\right)_{p' \in \mathcal{S}\sqcup \mathcal{M}\sqcup \mathcal{D}}, \left(\mathcal{C}, \mathcal{Z}\right), Z
            \right)
        \end{align*}
        conditioned on 
        \begin{enumerate}
            \item [1.)] The actions $\left(\left(a_{p'}\right)_{p' \in \mathcal{S}\sqcup \mathcal{M}\sqcup \mathcal{D}}\right)$;
            \item [2.)] The physically-feasible market transaction $\left(\mathcal{C}, \mathcal{Z}\right)$;
            \item [3.)] The realization of the off-market random variable $Z$ being $z \in [0,1)^N$.
        \end{enumerate}
        
        \subsection{Independent Off-Market Physics}
        \label{subsection:independent off-market}
    
            An independent off-market physics is an off-market physics that satisfies the following properties:
            \begin{enumerate}
            
                \item [1.)] (Individual Randomness) The degree of freedom $N = \left\vert\mathcal{S}\right\vert+\left\vert\mathcal{D}\right\vert$ and the off-market random vector $Z = \left(Z_p\right)_{p \in \mathcal{S}\sqcup \mathcal{D}}$;
                
                \item [2.)] (Identical Admissible Actions) The admissible actions set function $\hat{A}$ is such that, for each market participant $p \in \mathcal{S}\sqcup\mathcal{M}\sqcup \mathcal{D}$, the admissible actions set is
                \begin{align*}
                    A_p = \hat{A}(p)
                    =
                    \begin{cases}
                        \mathcal{A}_{\text{supply}} & \text{ if } p \in \mathcal{S}\\
                        \{0\} & \text{ if } p \in \mathcal{M}\\
                        \mathcal{A}_{\text{demand}} & \text{ if } p \in \mathcal{D}
                    \end{cases}
                    \text{,}
                \end{align*}
                with the supplier admissible actions set $\mathcal{A}_{\text{supply}} \ne \emptyset$ and the demander admissible actions set$\mathcal{A}_{\text{demand}} \ne \emptyset$;
                
                \item [3.)] (Supplier Utility)
                Each supplier $s \in \mathcal{S}$ realizes the utility 
                \begin{align*}
                    h^{(1)}(s) 
                    w\left(
                        h^{(0)}(s),
                        a_s,
                        \frac{\sum_{c \in \mathcal{C}} \left(c_1 \mathbf{1}_{c_3 = s}\right)}{h^{(1)}(s)},
                        \sum_{c \in \mathcal{C}} \left(c_2 \mathbf{1}_{c_3 = s}\right), z_s
                    \right)
                \end{align*}
                conditioned on one's own action $a_s \in \mathcal{A}_{\text{supply}}$, the market transaction $\left(\mathcal{C}, \mathcal{Z}\right)$, and the realized personal random variable $Z_s$ being $z_s$, with
                \begin{itemize}
                    \item the supplier intensive characteristics space $\Omega_{\text{supply}}$;
                    \item the supplier intensive characteristics function $h^{(0)}_{\text{}}:\mathcal{S} \to \Omega_{\text{supply}}$;
                    \item the supplier extensive characteristics function $h^{(1)}:\mathcal{S} \to \mathbb{R}^+$;
                    \item the measurable function $w:\Omega_{\text{supply}} \times \mathcal{A}_{\text{supply}}
                    \times \mathbb{R}^+_0
                    \times \mathbb{R}
                    \times [0,1)
                    \to [-\infty, \infty)$;
                \end{itemize}
                
                \item [4.)] (Demander Utility)
                Each demander $d \in \mathcal{D}$ realizes the utility 
                \begin{align*}
                    \eta^{(1)}(d) 
                    \omega\left(
                        \eta^{(0)}(d),
                        a_d,
                        \frac{\sum_{\zeta \in \mathcal{Z}} \left(\zeta_1 \mathbf{1}_{\zeta_3 = d}\right)}{\eta^{(1)}(d)},
                        \sum_{\zeta \in \mathcal{Z}} \left(\zeta_2 \mathbf{1}_{\zeta_3 = d}\right), z_d
                    \right)
                \end{align*}
                conditioned on one's own action $a_d \in \mathcal{A}_{\text{demand}}$, the market transaction $\left(\mathcal{C}, \mathcal{Z}\right)$, and the realized personal random variable $Z_d$ being $z_d$, with
                \begin{itemize}
                    \item the demander intensive characteristics space $\Omega_{\text{demand}}$;
                    \item the demander intensive characteristics function $\eta^{(0)}:\mathcal{D} \to \Omega_{\text{demand}}$;
                    \item the demander extensive characteristics function $\eta^{(1)}:\mathcal{D} \to \mathbb{R}^+$;
                    \item the function measurable $\omega:\Omega_{\text{demand}} \times \mathcal{A}_{\text{demand}}
                    \times \mathbb{R}^+_0
                    \times \mathbb{R}
                    \times [0,1)
                    \to [-\infty, \infty)$;
                \end{itemize}

                \item [5.)] (Mediator Utility)
                Each mediator $m \in \mathcal{M}$ realizes the utility 
                \begin{align*}
                    &\sum_{c\in \mathcal{C}}
                    \left(
                        h^{(1)}(c_3)
                        \tilde{w}
                        \left(
                            h^{(0)}(c_3),
                            a_{c_3},
                            \frac{c_1}{h^{(1)}(c_3)},
                            c_2, 
                            z_{c_3}
                        \right)
                        \mathbf{1}_{c_4=m}
                    \right)
                    +
                    \sum_{\zeta \in \mathcal{Z}}
                    \left(
                        \eta^{(1)}(\zeta_3)
                        \tilde{\omega}
                        \left(
                            \eta^{(0)}(\zeta_3),
                            a_{\zeta_3},
                            \frac{\zeta_1}{\eta^{(1)}(\zeta_3)},
                            \zeta_2, 
                            z_{\zeta_3}
                        \right)
                        \mathbf{1}_{\zeta_4=m}
                    \right)
                \end{align*}
                conditioned on the actions $\left(a_p\right)_{p \in \mathcal{S} \sqcup \mathcal{D}}$, the market transaction $\left(\mathcal{C}, \mathcal{Z}\right)$, and the realized random variable being $\left(z_p\right)_{p \in \mathcal{S} \sqcup \mathcal{D}}$, with
                \begin{itemize}
                    \item the function $\tilde{w}:\Omega_{\text{supply}} \times \mathcal{A}_{\text{supply}}
                    \times \mathbb{R}^+_0
                    \times \mathbb{R}
                    \times [0,1)
                    \to [-\infty, \infty)$;
                    \item the function $\tilde{\omega}:\Omega_{\text{demand}} \times \mathcal{A}_{\text{demand}}
                    \times \mathbb{R}^+_0
                    \times \mathbb{R}
                    \times [0,1)
                    \to [-\infty, \infty)$.
                \end{itemize}
            \end{enumerate}
        
            The independent off-market physics is characterized as 
            \begin{align*}
                \Phi\left(
                    \mathcal{S}, \mathcal{M}, \mathcal{D},
                    \mathcal{A}_{\text{supply}}, \mathcal{A}_{\text{demand}},
                    \Omega_{\text{supply}}, \Omega_{\text{demand}},
                    h^{(0)}, h^{(1)}, \eta^{(0)}, \eta^{(1)},
                    w,\omega,\tilde{w},\tilde{\omega}
                    \right)
                \text{.}
            \end{align*}

            In such case, we will see that we do not require the random variable $Z$ to be uniformly distributed on $[0,1)^N$ and it is sufficient to consider the case when, for any non-mediator market participant $p \in \mathcal{S} \sqcup \mathcal{D}$, 
            \begin{align*}
                Z_s \overset{\Delta}{=} \text{Unif}([0,1))\text{.}
            \end{align*}

    \section{Meta-Physics}

        Meta-physics is the rule that governs the generation of both in-market and off-market physics. For example, it can determine the correlation between the extensive supplier characteristics and the production capability.

        We will consider a specific structure that allows the tractability of the analysis yet still general enough to apply the abstract market to the markets of interest, which are trading market and credit market.
        
        \subsection{Scalable Market Generator}
        
            A scalable market generator is characterized as
            \begin{align*}
                \text{G}\left(
                    \mathcal{A}_{\text{supply}}, \mathcal{A}_{\text{demand}},
                    \bar{\Omega}_{\text{supply}}, \bar{\Omega}_{\text{demand}},
                    \hat{\Pi}_{\text{supply}},
                    \hat{\Pi}_{\text{demand}},
                    w, \omega, \tilde{w}, \tilde{\omega}
                \right)
            \end{align*}
            with
            \begin{enumerate}
                \item [1.)] The possible supplier characteristics set $\bar{\Omega}_{\text{supply}} \subseteq \Omega_{\text{supply}} \times \mathbb{R}^+ \times \mathbb{R}^+$;
                \item [2.)] The possible demander characteristics set $\bar{\Omega}_{\text{demand}} \subseteq \Omega_{\text{demand}} \times \mathbb{R}^+$
                \item [3.)] The individual supplier joint probability distribution
                \begin{align*}
                    \hat{\Pi}_{\text{supply}}
                    \in \Delta\left(
                        \bar{\Omega}_{\text{supply}}
                    \right)
                    \text{;}
                \end{align*}
                \item [4.)] The individual demander joint probability distribution
                \begin{align*}
                    \hat{\Pi}_{\text{demand}}
                    \in \Delta\left(
                        \bar{\Omega}_{\text{demand}}
                    \right)
                    \text{.}
                \end{align*}
            \end{enumerate}            
    
            By applying the generator $G\left(\mathcal{A}_{\text{supply}}, \mathcal{A}_{\text{demand}}, \bar{\Omega}_{\text{supply}}, \bar{\Omega}_{\text{demand}}, \hat{\Pi}_{\text{supply}}, \hat{\Pi}_{\text{demand}}, w, \omega, \tilde{w}, \tilde{\omega}\right)$ upon the market participants sets tuple $\left(\mathcal{S}, \mathcal{M}, \mathcal{D}\right)$, the scale tuple $\left(\lambda_{\text{supply}}, \lambda_{\text{demand}}\right) \in \left(\mathbb{R}^+\right)^2$, and the capacity threshold $\ubar{\mu} \in [0,1]$, the generator can stochastically generate a normal market in the following procedures:
            \begin{enumerate}
            
                \item [\namedlabel{condition:supplier}{A.}] (Supplier Introduction) 
                The supplier intensive characteristics function $h^{(0)}: \mathcal{S} \to \Omega_{\text{supply}}$, the supplier extensive characteristics function $h^{(1)}: \mathcal{S} \to \mathbb{R}^+$, and the production capability function $\hat{v}: \mathcal{S} \to \mathbb{R}^+_0$ are specified such that, for each supplier $s \in \mathcal{S}$,
                \begin{align*}
                    \left(h^{(0)}(s), \frac{h^{(1)}(s)}{\lambda_{\text{supply}}}, \frac{\hat{v}(s)}{\lambda_{\text{supply}}}\right)
                    \overset{\text{iid.}}{\sim}
                    \Pi_{\text{supply}}
                    \text{;}
                \end{align*}
                
                \item [\namedlabel{condition:demander}{B.}] (Demander Introduction) 
                The demander intensive characteristics function $\eta^{(0)}: \mathcal{D} \to \Omega_{\text{demand}}$ and the demander extensive characteristics function $\eta^{(1)}: \mathcal{D} \to \mathbb{R}^+$ are specified such that, for each demander $d \in \mathcal{D}$,
                \begin{align*}
                    \left(\eta^{(0)}(d), \frac{\eta^{(1)}(d)}{\lambda_{\text{demand}}}\right)
                    \overset{\text{iid.}}{\sim}
                    \Pi_{\text{demand}}
                    \text{;}
                \end{align*}

                \item [\namedlabel{condition:offmarket}{C.}] (Off-Market Generation)
                The off-market are specified to be 
                \begin{align*}
                    \Phi\left(
                        \mathcal{S}, \mathcal{M}, \mathcal{D},
                        \mathcal{A}_{\text{supply}}, \mathcal{A}_{\text{demand}},
                        \bar{\Omega}_{\text{supply}}, \bar{\Omega}_{\text{demand}},
                        h^{(0)}, h^{(1)}, \eta^{(0)}, \eta^{(1)},
                        w,\omega,\tilde{w},\tilde{\omega}
                        \right)
                    \text{;}
                \end{align*}

                \item [\namedlabel{condition:marketcontext}{D.}] (Normal Market Context)
                The normal market context $\mathcal{X}$ is specified such that the knowledge and belief structure of the off-market is such that
                \begin{enumerate}
                    \item [d.0)] There is a correct common knowledge of $\left(\mathcal{A}_{\text{supply}}, \mathcal{A}_{\text{demand}}, \Omega_{\text{supply}}, \Omega_{\text{demand}}, w,\omega,\tilde{w},\tilde{\omega}, \hat{\Pi}_{\text{supply}}, \hat{\Pi}_{\text{demand}}\right)$, $\left(\lambda_{\text{supply}}, \lambda_{\text{demand}}\right)$, $\ubar{\mu}$, and that the capacity measure $\mu: 2^{\mathcal{M}} \to [0,1]$ is defined such that $\mu(M) = \frac{\vert M\vert}{\vert \mathcal{M}\vert}$ for any $M \subseteq \mathcal{M}$;
                    
                    \item [d.1)] There is a correct common prior that 
                    \begin{itemize}
                        \item The suppliers characteristics are stochastically generated as in the process~\ref{condition:supplier};
                        \item The demanders characteristics are stochastically generated as in the process~\ref{condition:demander};
                        \item The off-market physics is deterministically generated as in the process~\ref{condition:demander} conditioned on the realization of the suppliers characteristics and the demanders characteristics;
                    \end{itemize}

                    \item [d.2)] (Supplier Private Observation) 
                    There is a correct knowledge that each supplier $s \in \mathcal{S}$ observes the realization $\left(h^{(0)}(s), h^{(1)}(s), \bar{v}(s)\right)$;

                    \item [d.3)] (Demander Private Observation) 
                    There is a correct knowledge that each demander $d \in \mathcal{D}$ observes the realization $\left(\eta^{(0)}(d), \eta^{(1)}(d)\right)$;
                \end{enumerate}

                \item [E.] The market is specified to be
                \begin{align*}
                    \circled{$\hat{M}$}\left(
                    \mathcal{S}, \mathcal{M}, \mathcal{D}, \bar{v}, \mu, \ubar{\mu}, \mathcal{X}
                    \right)
                    \text{.}
                \end{align*}

            \end{enumerate}
            
            This scalable market generator suggests that a sequence of the randomized market will be well-behave. For example, the goal number of demanded scarcity will not exploding. Although the definition of the scalable market generator only suggests that it can be applied to a sequence of finite normal markets, we can see in the next section that the normal market can be conveniently extended to the case when there are continuum of agents and the exact law of large numbers (The distribution of the agents according to the population distribution) holds.

\part{Competitive Market: Equilibrium \& Analysis}
\label{part:competitive}

\section{Competitive Market}
\label{section:competitivemarket}

    A competitive market as characterized as
    \begin{align*}
        \circled{C}^{ \left(\mathcal{A}_{\text{supply}}, \mathcal{A}_{\text{demand}},
        \bar{\Omega}_{\text{supply}}, \bar{\Omega}_{\text{demand}},
        w,\omega,\tilde{w},\tilde{\omega}\right)}
        \left(
            \hat{\Pi}_{\text{supply}}, \hat{\Pi}_{\text{demand}}
        \right)
    \end{align*}
    is the continuum limit of the market generated by a sequence of market generated by applying the scalable market generator
    \begin{align*}
        G\left(
            \mathcal{A}_{\text{supply}}, \mathcal{A}_{\text{demand}},
            \bar{\Omega}_{\text{supply}}, \bar{\Omega}_{\text{demand}},
            \hat{\Pi}_{\text{supply}},
            \hat{\Pi}_{\text{demand}}
            w, \omega, \tilde{w}, \tilde{\omega}
        \right)
    \end{align*}
    upon a sequence of $\left(\left(\mathcal{S}^{(i)}, \mathcal{M}^{(i)}, \mathcal{D}^{(i)}\right), \left(\lambda_{\text{supply}}^{(i)}, \lambda_{\text{demand}}^{(i)}\right)\right)$ with $\ubar{c} = 0$ such that
    \begin{enumerate}
        \item [1.)] 
        $\lim_{i \to \infty} \frac{\min\left(\left\{\left\vert \mathcal{S}^{(i)}\right\vert, \left\vert \mathcal{D}^{(i)}\right\vert\right\}\right)}{\left\vert \mathcal{M}^{(i)}\right\vert} = \infty$;
        \item [2.)] 
        $\lim_{i \to \infty} \left\vert \mathcal{M}^{(i)}\right\vert = \infty$;
        \item [3.)] For all $i$,
        $\lambda_{\text{supply}}^{(i)} = \left\vert \mathcal{S}^{(i)}\right\vert$ and $\lambda_{\text{demand}}^{(i)} = \left\vert \mathcal{D}^{(i)}\right\vert$.
    \end{enumerate}
    
    However, the competitive limit of the game will be the case when there exists a continuum, therefore having an uncountably infinite number, of suppliers, mediators, and suppliers.

    We assume that the continuum counter part of the capacity measure $\mu$, which is defined for a finite market case such that $\mu(M) = \frac{\vert M \vert}{\vert \mathcal{M} \vert}$ for any $M \subseteq \mathcal{M}$, is the Lebesgue measure. Thus, even when the threshold capacity $\ubar{c}$ is set to be $0$, we will have that the resale period can be enabled if (and only if) the set of winning mediators $\tilde{\tilde{\tilde{\mathcal{M}}}}$ has zero Lebesgue measure.

    \paragraph{Supplier, Demander, \& Differential Utility}

        We will model the collection of the suppliers $\mathcal{S} = \bar{\Omega}_{\text{supply}} \times [0,1) \times [0,1)$. This suggests there can be infinitely many copies of suppliers with the same supplier characteristics in the characteristics space. In this model, the ratio between the population size of a specific supplier characteristics and that of another characteristics is not well-defined. The effect of each supplier action is then has to be re-weighted by the probability distribution of the supplier characteristics $\hat{\Pi}_\text{supply}$. Specifically, the distribution over $\mathcal{S}$ will be equivalent to the distribution of $(\theta, z)$ where $\theta \sim \hat{\Pi}_\text{supply}$ and $z \sim \text{Unif}\left([0,1)^2\right)$, independently. This will also suggests that, for any supplier $s \in \mathcal{S}$, its probability of getting sampled is $0$, making any mediator cannot alter the market significantly.
    
        The supplier bid strategy is restricted to a joint deterministic strategy represented as a measurable\footnote{The measurability is an additional requirement, which is vacuously enforced in a non-limit setting.} function $f: \bar{\Omega}_{\text{supply}} \times [0,1) \to \mathbb{R}^+ \times \mathbb{R}$ such that, for any supplier $s \in \mathcal{S}$, the ``differential" volume bid and the supply price bid $(v_s, \rho_s) = f(s_1, s_2)$. The volume $v_s$ is denoted as differential, since the rescaling of the market, together with a condition on a well-behaved market to be defined in the following section, will enforce that the bid of a single supplier will become $0$. However, the total bid volume of an uncountably large number of supplier will be non-zero. Thus, we use differential volume to denote the density of the volume distribution, if it exists and is well-defined. Moreover, for each bid, there will be infinitely many other agents who use it. 
    
        The supplier off-market strategy will also be restricted to be represented by a measurable function $f:\bar{\Omega}_{\text{supply}} \times [0,1) \to \mathcal{A}_{\text{supply}}$, such that, for any supplier $s \in \mathcal{S}$, the off-market action selected is $f(s_1, s_2)$, meaning that each off-market action, if selected, will be selected by infinitely many agents.
    
        For a supplier $s \in \mathcal{S}$, if a contract is acquired via the market mechanism, it will also have a volume of $0$, but a non-zero differential volume of $v_s$. If a supply price associated with the contract is $\rho \in \mathbb{R}$, we will have that the expected utility, if it is well-defined, under an off-market action $a \in \mathcal{A}_{\text{supply}}$, for a supplier $s = \left(h_0, h_1, v, z\right)$ is,
        \begin{align*}
            \lim_{\lambda \to 0^+} \lambda h_1 \mathbb{E}\left[w\left(h_0, \frac{\lambda v_s}{\lambda h_1}, \rho, a,
            \epsilon\right)\right] = 0
        \end{align*}
        if $\mathbb{E}\left[w\left(h_0, \frac{v_s}{h_1}, \rho, a,
            \epsilon\right)\right] \in \mathbb{R}$.
    
        However, the corresponding differential utility will not be vacuous, since
        \begin{align*}
            \lim_{\lambda \to 0^+} 
            \frac{\lambda h_1}{\lambda} \mathbb{E}\left[w\left(h_0, \frac{\lambda v_s}{\lambda h_1}, \rho, a,
            \epsilon\right)\right] = h_1 \mathbb{E}\left[w\left(h_0, \frac{v_s}{h_1}, \rho, a,
            \epsilon\right)\right]
            \text{,}
        \end{align*}
        which is also the same as that of a non-limit market when the external characteristics $h_1$ and the volume bid is not rescaled.
    
        Thus, the utility evaluated for each supplier will be done with the differential utility\footnote{Note that this is in line with the analysis using utility as a non-atomic measure treated in ``Values of Non-Atomic Games. I: The Axiomatic Approach", Aumann \& Shapley, 1968,~\cite{aumann1968values}. The differential utility is equivalent to the distribution (density) function associated with the utility measure.} in order to provide a notion of rationality as (differential) utility maximization. 
        
        Another important technical issue is that, we can no longer randomize $\left(\epsilon_{s}\right)_{s \in \mathcal{S}}$ to be a collection of independent and identical random variable. As the expected utility for each supplier does not rely on the independence of the $\left(\epsilon_{s}\right)_{s \in \mathcal{S}}$ in the non-limit market, we will then only assume that $\left(\epsilon_{s}\right)_{s \in \mathcal{S}}$ such that, for any $s \in \mathcal{S}$, $\epsilon_{s} = \epsilon_{\mathcal{S}} \overset{\Delta}{=} \text{Unif}([0,1))$ 
    
        Similarly, the collection of demander $\mathcal{D}= \bar{\Omega}_{\text{supply}} \times [0,1) \times [0,1)$, the demander bid strategy is restricted to a joint deterministic strategy represented as a measurable function $f: \bar{\Omega}_{\text{supply}} \times [0,1) \to \mathbb{R}^+ \times \mathbb{R}$ such that, for any demander $d \in \mathcal{D}$, the differential volume bid and the demand price bid $(v_d, r_d) = f(d_1, d_2)$, and the demander off-market strategy will also be restricted to be represented by a measurable function $f:\bar{\Omega}_{\text{demand}} \times [0,1) \to \mathcal{A}_{\text{demand}}$, such that, for any demander $d \in \mathcal{D}$, the off-market action selected is $f(d_1, d_2)$.

        The differential utility system for the demand side will define in the same manner as that of the supply side, since every demander, under a technical assumption on the expected utility, will also have $0$ demand volume bid and has $0$ utility.

    \paragraph{Mediator Supply Bid}

        The collection of the mediators $\mathcal{M} = [0,1)$, since there is no hidden characteristics for any mediator. 
    
        The supply-side action of the mediators can be compressed into a measurable function $f: [0,1) \to [0,1] \times \left(\mathbb{R} \times \{-\infty\}\right)$ where, for any mediator $m \in [0,1)$, $\left(f(m)\right)_2$ represents the counter supply price bid made by the mediator, and $\left(f(m)\right)_1$ represents the proportion of supply residual acquired by the winning mediator groups. Thus, the winning mediator group should be restricted to be a member of a collection of subsets of $[0,1)$ such that uniform distribution is well-defined. In this case, we will use $\mathcal{G}$ (as defined in the definition~\ref{definition:goodborel}), and, for any $G \in \mathcal{G}$, the uniform distribution can be defined according to the probability measure $I_G$ (as defined in the definition~\ref{definition:goodboreluniform}). Thus, the bidding strategy $f$ is restricted to be a collection of function whose lexicographic maximizer is a set $G \in \mathcal{G}$.

        Assume that there is no supplier whose lowest supply price willing to take is exactly equal to the market supply price. Another technical issue is that a set of winning supplier will be in the form of $W \times [0,1)$, where $W \subseteq \bar{\Omega}_{\text{supply}} \times [0,1)$. Thus, if $W \ne \emptyset$, then we cannot match each supplier uniformly at random according to $I_G$, when the set of winning mediators is $G \in \mathcal{G}$ with $\vert G \vert > 1$. Moreover, we cannot relax the independence condition as done in the analysis of the off-market utility of each supplier. To circumvent this problem, if $W$ is finite, we will assume that each mediator in $G$ will get the identical measure of suppliers with the same characteristics measure and the same strategy measure. Formally, we consider a measurable surjective function $f: [0,1) \to G$ such that, for any $g \in G$, $\lambda\left(\left\{z \in [0,1): f(z) = g\right\}\right) = \frac{1}{\vert G\vert}$, and we can assume that a winning mediator $g \in G$ will be matched to a subset of mediators $W \times \left\{z \in [0,1): f(z) = g\right\}$. In the case when $G$ has non-zero Lebesgue measure, we have to consider a measurable bijective function $f:G \to [0,1)$ such that, for any $A \in \mathcal{B}([0,1))$, $I_G(A) = \lambda\left(\left\{f(a)\right\}_{a \in A\cap G}\right)$, and we can assume that a subset of winning mediators $G' \in \mathcal{B}$ with $G' \subseteq G$ will get matched a subset of mediators $W \times \left\{f(a)\right\}_{a \in G'}$.

        In the case when there may exist some supplier whose lowest supply price willing to take is exactly equal to the market supply price, we have to change the family of the winning suppliers to be $\left(W_1 \times [0,1) \right) \cup \left(W_2 \times [0,q)\right)$, where $q$ is the highest proportion of supply residual acquired by the winning mediator groups determined by the bid, which will be analog to the ratio between the target residual volume and the total residual volume. Different from that of the non-limit market, the target volume in the supply residual auction would be filled perfectly, since the willing to contract volume placed by each mediator is infinitesimal.

        An important point is that the randomness in the supply side will vanish due to the continuum limit consideration. Therefore, if the strategy used by the suppliers is known, the amount of globally acquired supply can be perfectly calculated.

    \paragraph{Mediator Demand Bid}

        The action done by the mediator in the demand side is the placement of the quote demand prices, so it will be represented by a measurable function $f: [0,1) \to \mathbb{R}$, where for any mediator $m \in [0,1)$, the counter demand price bid is $f(m)$. Since the market clearing process is done in a sequential gamer through the increasing sequence of price, the range of $f$ should be such that there exists a finite collection of non-decreasing real sequence $\left\{\left(r_i^{(j)}\right)_{i=1}^{\infty}\right\}_{j=1}^{m}$ for some $m \in \mathbb{N}$ such that $\text{Range}\left(f\right) = \bigcup_{j=1}^N \left\{r^{(j)}_i\right\}_{i=1}^{\infty}$.

        Similar to the supply side randomization, the demand side randomization will become deterministic and equally distributed among the mediators who place the same bid price.
        
        Unlike a supplier and a demander, a mediator may be capable of contracting non-zero amount of scarcity volume. This suggests that, under a technical assumption on the finiteness of expected utility per unit of volume, a mediator may receive non-zero utility. Therefore, the utility of mediator in the analysis will be represented as a $2$-vector instead, where the first element represents the utility mass and the second element represents the differential utility (utility density).   

    \paragraph{Resale}

        In a non-limit market, a resale is allowed if and only if the number of mediators who wins the auction in the supply side auction is less than half. In this case, we will replace the condition to be that the Lebesgue measure of the set of winning mediators is less than $\frac{1}{2}$.

        In each resale auction, only one mediator will win the entire bid of non-zero volume, so the randomness in the winning mediator can be taken into account by restricting the bidding strategy to be such that the set of (lexicographic) winners is finite or has non-zero Lebesgue measure.

\section{Equilibrium with Continuum of Agents}
\label{section:equilibrium with continuum of agents}

    In this part, we will consider an equilibrium concept in a limit market where the number of agencies explodes to a continuum limit. In this section, we will then introduce some relevant equilibrium notion in the model.

    \subsection{Normal Form Game with Continuum of Agents}

        The agents set is $[0,1)$. Each agent is endowed with the admissible action sets $\mathcal{A} \subseteq \mathbb{R}^d$ for some $d \in \mathbb{N}$.

        The admissible joint strategies set is a collection of functions $\mathcal{S}$, where for each $s \in \mathcal{S}$, 
        \begin{itemize}
            \item $s: [0,1) \to \mathcal{A}$ is measurable;
            \item For any $a \in \mathcal{A}$, $y \in [0,1)$, there exists $s' \in \mathcal{S}$ such that, for any $x \in [0,1)$
            \begin{align*}
                s'(x)
                =
                \begin{cases}
                    s(x) &\text{ if } x \ne y
                    \\
                    a &\text{ if } x = y
                \end{cases}
                \text{.}
            \end{align*}
        \end{itemize}

        Note that, unlike the (pure strategy) game characterization introduced in Nash's seminal papers~\cite{nash1950non, nash1950equilibrium}, we have that the actions that a group of agents can take is not completely independent of the action of other agents in the system.\footnote{This is an inherent problem in a continuous formulation for game theoretic decision analysis~\cite{aumann1964, aumann1968values}. However, we can see that this is not an unreasonable assumption, since the restriction of action based on the action of others have been used in the game-theoretic formulation of general equilibrium in the paper~\cite{arrow1954}.} 

        For each admissible joint strategy $s \in S$, a measurable conditional utility function $U^{(s)}: [0,1) \to \mathbb{R}^n$ for some $n \in \mathbb{N}$ is given, meaning that, if every agent $x \in [0,1)$ selects action $s(x) \in \mathcal{A}$, then each agent $y \in [0,1)$ will receive a lexicographical utility of $U^{(s)}(y)$.

        The game is characterized by $\left(\mathcal{A}, \mathcal{S}, n, \left(U^{(s)}\right)_{s \in \mathcal{S}}\right)$.

        \subsection{Nash Equilibrium}

            The utility comparison will be done in a lexicographic manner.

            \begin{definition}
                (Lexicographic Order Operator)
                
                For any $n \in \mathbb{N}$,
                \begin{itemize}
                    \item 
                    The weak lexicographical relationship operator $\succcurlyeq_n$ is defined such that, for any $x, y \in \mathbb{R}^n$,
                    \begin{align*}
                        [x \succcurlyeq_n y]
                        \iff
                        [
                        \min\left(\{i \in \mathbb{N} \cap (-\infty,n]: x_i > y_i\} \cup \{n+1\}\right) \le \min\left(\{i \in \mathbb{N} \cap (-\infty,n]: x_i < y_i\} \cup \{n+1\}\right)
                        ]
                        \text{;}
                    \end{align*}
                    \item 
                    The strict lexicographical relationship operator $\succ_n$ is defined such that, for any $x, y \in \mathbb{R}^n$,
                    \begin{align*}
                        [x \succ_n y]
                        \iff
                        [
                        x \succcurlyeq y \land x \ne y
                        ]
                        \text{.}
                    \end{align*}
                \end{itemize}
            \end{definition}

            A Nash equilibrium is an equilibrium, where any agent $y \in [0,1)$ will not benefit from deviating from the equilibrium strategy.

            \begin{definition} (Nash Equilibrium)
            
                For any normal form game with continuum of agents $\left(\mathcal{A}, \mathcal{S}, n, \left(U^{(s)}\right)_{s \in \mathcal{S}}\right)$, 
                the strategy $s \in \mathcal{S}$ is a Nash equilibrium strategy of the game, if, for any $y \in [0,1)$, $a \in \mathcal{A}$, by defining the altered strategy $\hat{s}^{(y,a)} \in \mathcal{S}$ such that, for any $x \in [0,1)$,
                \begin{align*}
                    \hat{s}^{(y,a)}(x)
                    =
                    \begin{cases}
                        s(x) &\text{ if } x \ne y
                        \\
                        a &\text{ if } x = y
                    \end{cases}
                    \text{,}
                \end{align*}
                we have that
                \begin{align*}
                    U^{\left(s\right)}(y)
                    \succcurlyeq
                    U^{\left(\hat{s}^{(y,a)}\right)}(y)
                    \text{.}
                \end{align*}
            \end{definition}

        \subsection{Betrayal-Free-Collusion-Free Equilibrium}

            In several applications, a collection of Nash equilibrium strategies can be too big. Since the utility of an agent in most game can often be model as a function of individual action and an aggregated macroscopic result, which cannot be altered by a single individual action.\footnote{Some examples will be provided to show how this new notion of equilibrium can produce more intuitive results even under simple circumstances.} Thus, it is beneficial to consider a notion of equilibrium when more than $1$ agents can agree to deviate.

            \begin{definition}
            \label{definition:goodborel}
                A collection of good Borel subsets of $[0,1)$ is
                \begin{align*}
                    \mathcal{G}
                    =
                    \left\{
                    B \in \mathcal{B}\left([0,1)\right)
                    :
                    \vert B \vert \in \mathbb{N} \lor \lambda(B) \in (0,1]
                    \right\}
                    \text{.}
                \end{align*}
            \end{definition}

            \begin{definition}
            \label{definition:goodboreluniform}
                For any $G \in \mathcal{G}$, a conditional uniform measure $I_G$ is a measure over $\mathcal{B}([0,1))$ such that, for any $A \in \mathcal{B}([0,1))$,
                \begin{align*}
                    I_G (A)
                    =
                    \begin{cases}
                        \frac{\vert A \cap G \vert}{\vert G \vert}
                        &\text{ if }
                        \vert G \vert \ne \infty\\
                        \frac{\lambda\left(A \cap G\right)}{\lambda\left(G\right)}
                        &\text{ if }
                        \vert G \vert = \infty
                    \end{cases}
                    \text{.}
                \end{align*}
            \end{definition}

            \begin{definition} (Collusion-Free Equilibrium)
            
                For any normal form game with continuum of agents $G = \left(\mathcal{A}, \mathcal{S}, n, \left(U^{(s)}\right)_{s \in \mathcal{S}}\right)$, 
                the strategy $s \in \mathcal{S}$ is a collusion-free equilibrium strategy of the game, if, by defining, for any $G \in \mathcal{G}$, $f \in \mathcal{S}$, an admissible collusion deviated strategies set
                \begin{align*}
                    \mathcal{S}^{(f,G)}
                    =
                    \left\{s' \in \mathcal{S}:
                    \forall x \in [0,1) - G [s'(x) = f(x)] 
                    \right\}
                    \text{,}
                \end{align*}
                we will have that, for any $G \in \mathcal{G}$, $s' \in \mathcal{S}^{(s,G)}$,
                \begin{align*}
                    I_G\left(\left\{
                        x \in G :
                        U^{(s)}(x) \succ_n U^{(s')}(x)
                    \right\}\right)
                    >
                    0\text{.}
                \end{align*}
            \end{definition}

            Since, for any $x \in [0,1)$, $\{x\} \in \mathcal{G}$, we have that the collection of collusion-free equilibrium strategies is a subset of the collection of Nash equilibrium strategies.

            However, the collusion-free equilibrium strategies can be too small in many applications, since the collusion can force every constituents of the group to commit to the planned deviation. In conventional Cournot competition, this means that all producers can commit to deviate to the choice of production as if they were a monopoly.

            Thus, we will only consider the colluded deviation when each subgroup cannot further deviate to benefit more while hurting the rest of the group.

            \begin{definition} (Betrayal-Free-Collusion-Free Equilibrium)
            
                For any normal form game with continuum of agents $\left(\mathcal{A}, \mathcal{S}, n, \left(U^{(s)}\right)_{s \in \mathcal{S}}\right)$, 
                the strategy $s \in \mathcal{S}$ is a collusion-free equilibrium strategy of the game, if, by defining, for any $G \in \mathcal{G}$, $f \in \mathcal{S}$, an admissible collusion deviated strategies set
                \begin{align*}
                    \mathcal{S}^{(f,G)}
                    =
                    \left\{s' \in \mathcal{S}:
                    \forall x \in [0,1) - G [s'(x) = f(x)] 
                    \right\}
                    \text{,}
                \end{align*}
                we will have that, for any $G \in \mathcal{G}$, $s' \in \mathcal{S}^{(s,G)}$,
                \begin{align*}
                    I_G\left(\left\{
                        x \in G :
                        U^{(s)}(x) \succ_n U^{(s')}(x)
                    \right\}\right)
                    >
                    0
                    \text{,}
                \end{align*}
                \textbf{or} there exists some $G' \subseteq G$ with $G', \left(G - G'\right) \in \mathcal{G}$ such that there  exists some $s'' \in \mathcal{S}^{\left(s', G'\right)}$ such that
                \begin{align*}
                    I_{G'}\left(\left\{
                        x \in G' :
                        U^{(s'')}(x) \succ_n U^{(s')}(x)
                    \right\}\right)
                    =
                    1
                    \text{,}
                \end{align*}
                and
                \begin{align*}
                    I_{\left(G-G'\right)}\left(\left\{
                        x \in G-G' :
                        U^{(s)}(x) 
                        \succcurlyeq_n 
                        U^{(s'')}(x)
                    \right\}\right)
                    >
                    0
                    \text{.}
                \end{align*}
            \end{definition}

            Note that, for any $x \in [0,1)$, $\{x\} \in \mathcal{G}$, there does not exists any $G \subseteq \{x\}$ such that $\{x\} - G$ and $G$ are elements of $\mathcal{G}$. Therefore, the collection of betrayal-free-collusion-free equilibrium strategies is a subset of the collection of Nash equilibrium strategies.

            \begin{figure} [H]
                \begin{subfigure}{1\textwidth}
                  \centering
                  \includegraphics[width=0.85\linewidth]{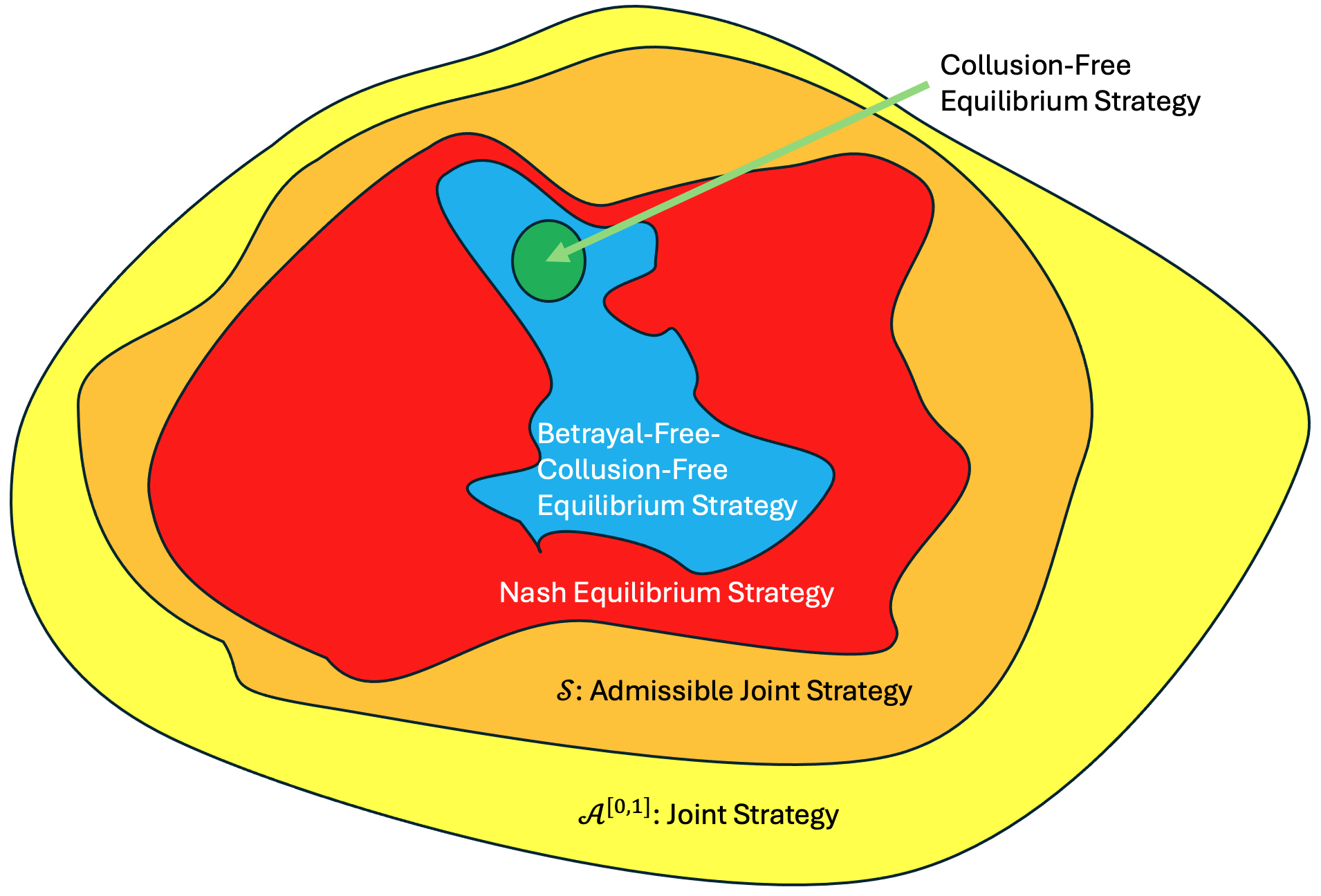}
                  \caption{Schematic Diagram comparing Different Types of Equilibria}
                \end{subfigure}
                \caption{The admissible joint strategy $\mathcal{S}$ does not contain every function $f: [0,1] \to \mathcal{A}$, since it only contains the joint strategy that the utility for each agent is well-defined. Within the $3$ notions of equilibria introduced in this section, Nash equilibria set is the biggest. The betrayal-free-collusion-free equilibria set is a subset of the Nash equilibria set, and the collusion-free equilibria set is a subset of the betrayal-free-collusion-free equilibria set.}
            \end{figure}

            \paragraph{Examples on Oligopoly Pricing}

                Below we pose a simple example when the adoption of equilibrium as the betrayal-free-collusion-free equilibrium produces a more intuitive result comparing to that provided by the Nash equilibrium, and a simple example when the use of the betrayal-free-collusion-free equilibrium provides a more intuitive result comparing to that provided by the collusion-free equilibrium. 
                
                The examples draw an intuition from the case when the scarcity is distributed uniformly to each retail $m \in [0,1]$. The problem faced by the retails is then to set a selling price in order to maximize the revenue conditioned on the prices set by other retails. This is in line with the residual demand analysis used extensively in industrial organization literature~\cite{mahoney}. However, since we have a continuum of agents, the residual demand by having one agent removed is equivalent to the residual demand as observed by a possible market entree. To make the residual demand  by having a group of agents removed, we then required to make such group (collusion) to have a non-zero Lebesgue measure. The supply shown in the graph is the aggregated amount of supply.
                \begin{figure} [H]
                    \begin{subfigure}{0.5\textwidth}
                      \centering
                      \includegraphics[width=0.85\linewidth]{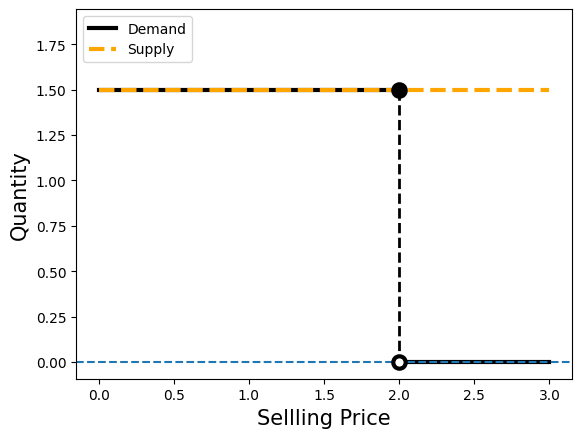}
                      \caption{Too Big Nash Equilibria}
                    \end{subfigure}%
                    \begin{subfigure}{0.5\textwidth}
                      \centering
                      \includegraphics[width=0.85\linewidth]{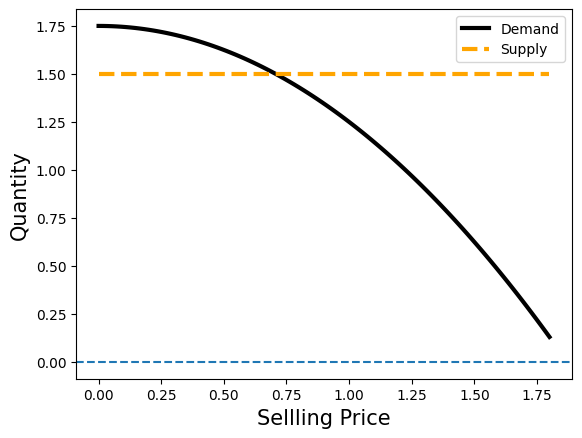}
                      \caption{Too Small (empty) Collusion-Free Equilibria}
                    \end{subfigure}
                    \caption{(a) It is easy to show that in this framework, every Walrasian equilibrium will also be a Nash equilibrium, so there exists a continuum of Nash equilibrium by having every agent set the selling price to be any value in $(0,2]$. However, the betrayal-free-collusion-free uniquely exists and is when every retail set the selling price of $2$, which is the profit maximising price.\\(b) If every market clearing price (Walrasian equilibrium) is lower than the monopoly price, we will have that, we will have that there cannot be a collusion-free equilibrium. However, in this example, the betrayal-free-collusion-free uniquely exists and coincide with the unique Walrasian equilibrium as expected.}
                \end{figure}

        \subsection{Pseudo-Equilibrium}

            The market equilibrium will be analyzed as a sub-game equilibrium. However, since the utility can be discontinuous and the action space is not necessarily closed or compact, a Nash equilibrium may not exist. However, in some cases, we still have that there exists a convergence of strategies and economic results\footnote{In a sub-game notation, the economic results can be the relevant state variables.}, and its limit of economic result, including the utility, can be used as if it is an equilibrium utility, even when the limit strategy is not an equilibrium strategy and does not support such utility profile.

            We define an augmented game to be game that includes an economic result.

            Let the set of economic results be $\mathcal{R} \subseteq \mathbb{R}^d$ for some $d \in \mathbb{N}$. 

            We are endowed with an economic results function $E: \mathcal{S} \to \mathcal{R}$.

            The augmented game is characterized by $\left(\mathcal{A}, \mathcal{S}, n, \left(U^{(s)}\right)_{s \in \mathcal{S}}, \mathcal{R}, E\right)$.

            \begin{definition}
                (Pseudo-Equilibrium) 
                
                For any augmented normal form game with continuum of agents $\left(\mathcal{A}, \mathcal{S}, n, \left(U^{(s)}\right)_{s \in \mathcal{S}}, \mathcal{R}, E\right)$, 
                a tuple $\left(\hat{u}, \hat{r} \right)$ with $\hat{u}: [0,1) \to \mathbb{R}^n$ being a measurable function and $\hat{r} \in \mathcal{R}$ is a pseudo-equilibrium of the game, if, there exists an infinite sequence $\left(s_i\right)_{i=1}^{\infty}$ with $\left\{s_i\right\}_{i=1}^{\infty} \subseteq \mathcal{S}$ such that
                \begin{itemize}
                    \item (Individually Better Response)
                    For any $x \in [0,1)$, for any $i \in \mathbb{N}$, by defining $s^{(x)}_i \in \mathcal{S}$ such that, for any $y \in [0,1)$,
                    \begin{align*}
                        s^{(x)}_i(y) = 
                        \begin{cases}
                            s_{i+1}(x) &\text{ if } y = x
                            \\
                            s_i(y) &\text{ if } y \ne x
                        \end{cases}
                        \text{,}
                    \end{align*}
                    we have that
                    \begin{align*}
                        U^{\left(s^{(x)}_i\right)}(x) \succcurlyeq_n
                        U^{\left(s_i\right)}(x)
                        \text{;}
                    \end{align*}

                    \item (Convergence in Result)
                    $U^{\left(s_i\right)}$ converges elementwise to $\hat{u}$ and $E\left(s_i\right)$ converges to $\hat{r}$;

                    \item (No Stagnation)
                    For any $a \in \mathcal{A}$, $x \in [0,1)$, $f \in \mathcal{S}$, by defining $f^{(x,a)} \in \mathcal{S}$ such that, for any $y \in [0,1)$, 
                    \begin{align*}
                        f^{(x,a)}(y) = 
                        \begin{cases}
                            a &\text{ if } y = x
                            \\
                            f(y) &\text{ if } y \ne x
                        \end{cases}
                        \text{,}
                    \end{align*}
                    we have that $\mathbf{1}_{\hat{u}(x) \succcurlyeq_n U^{\left(s_i^{(x,a)}\right)}(x)}$ converges to $1$.
                \end{itemize}
            \end{definition}

            We have that, if $s \in \mathcal{S}$ is a Nash equilibrium strategy, then $\left(U^{(s)}, E(s)\right)$ is a pseudo-equilibrium by consider a sequence $(s_i)_{i=1}^{\infty}$ with $\{s_i\}_{i=1}^\infty = \{s\}$.

            \subsubsection{Pseudo-Equilibrium for Known Common Value Auction}
            \label{subsubsection:commonlyknownvalue}

                Consider an auction system such that, there exists some value function $v: \mathbb{R} \to \mathbb{R}$ such that
                \begin{itemize}
                    \item $v$ is left-continuous;
                    \item There exists some $c \in \mathbb{R}$ such that, for any $x \le c$, $v(x) < 0$;
                    \item There exists some $b \in \mathbb{R}$ such that $v(b) > 0$.
                \end{itemize}
                
                Let $\mathcal{S}$ be a collection of measurable functions containing every function $f:[0,1) \to \mathbb{R}$ such that $\argmin_{x \in [0,1)} f(x) \in \mathcal{G}$.

                Let $n = 2$, and, for any $s \in \mathcal{S}$, $U^{(s)}: [0,1) \to \mathbb{R}^2$ is defined such that, for any $x \in [0,1)$,
                \begin{align*}
                    U^{(s)}(x)
                    =
                    \begin{cases}
                        (0,0) &\text{ if } x \not\in \argmin_{x \in [0,1)} s(x)\\
                        \begin{cases}
                            \left( \frac{v\left(\min_{y \in [0,1)} s(y)\right)}{\left\vert \argmin_{y \in [0,1)} s(y)
                            \right\vert},
                            \frac{v\left(\min_{y \in [0,1)} s(y)\right)}{\left\vert \argmin_{y \in [0,1)} s(y)
                            \right\vert}
                            \right)
                            &\text{ if }
                            \left\vert \argmin_{y \in [0,1)} s(y)
                            \right\vert \ne \infty
                            \\
                            \left(0, \frac{v\left(\min_{y \in [0,1)} s(y)\right)}{\lambda\left( \argmin_{y \in [0,1)} s(y)
                            \right)}
                            \right)
                            &\text{ if }  
                            \left\vert \argmin_{y \in [0,1)} s(y)
                            \right\vert = \infty
                        \end{cases}
                        &\text{ if } x \in \argmin_{y \in [0,1)} s(y)
                    \end{cases}
                    \text{.}
                \end{align*}

                Let $\mathcal{R} = \mathbb{R}$, and the function $E:\mathcal{S} \to \mathcal{R}$ is defined such that, for any $s \in \mathcal{S}$,
                \begin{align*}
                    E(s) = \min_{y \in [0,1)} s(y)
                    \text{.}
                \end{align*}

                We will have that a pseudo-equilibrium exists by considering a sequence of strategies $\left(s_i\right)_{i=1}^{\infty}$ such that, for any $i \in \mathbb{N}$,
                \begin{align*}
                    s_i(x) = y_i
                \end{align*}
                for all $x \in [0,1)$, where the sequence $\left(y_i\right)_{i=1}^\infty$ is such that,
                \begin{itemize}
                    \item there exists some $y \in \mathbb{R}$ with $\lim_{z \to y^+} v(z) \ge 0 \land \left\{v(z)\right\}_{z \in (-\infty,y]} \subseteq (-\infty, 0]$ such that $y_i \downarrow y$ as $i \to \infty$, \textbf{or},
                    \item there exists some $y \in \mathbb{R}$ with $v(y) \ge 0 \land \left\{v(z)\right\}_{z \in (-\infty,y)} \subseteq (-\infty, 0]$ such that $y_i = y$ for all $i \in \mathbb{N}$.
                \end{itemize}

                \begin{figure} [H]
                    \begin{subfigure}{0.5\textwidth}
                      \centering
                      \includegraphics[width=0.8\linewidth]{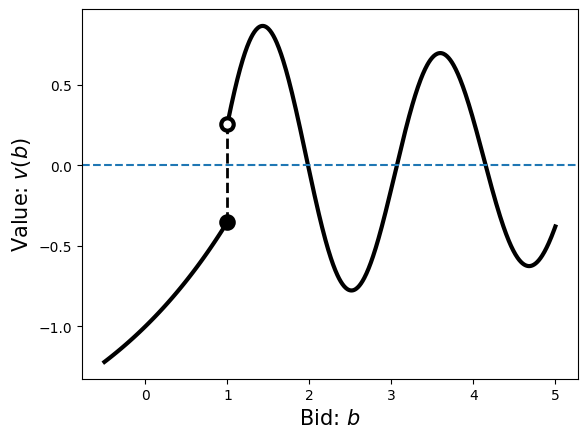}
                      \caption{Left Continuous Value Function $v$}
                    \end{subfigure}%
                    \begin{subfigure}{0.5\textwidth}
                      \centering
                      \includegraphics[width=0.8\linewidth]{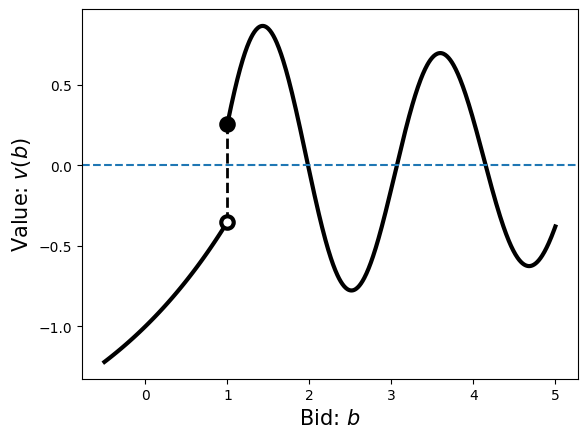}
                      \caption{Right Continuous Value Function $v$}
                    \end{subfigure}
                    \caption{In both cases, we can see that a pseudo-equilibrium exists, and correspond to the point where every agent updates the bid to be $1$ in the right-continuous function case and to be arbitrary close to $1$ from above in the left-continuous function case. For the right continuous case, it is easy to see that a Nash equilibrium uniquely exists and is when every agent bids $1$. However, a Nash equilibrium does not exist for the example above when the value function $v$ is left-continuous.}
                \end{figure}

            \subsection{Dominant Strategy}

                To simplify the analysis, we will restrict the decision making process of both suppliers and demanders to be a dominant strategy.

                \begin{definition}
                    (Dominant Strategy)
                    An action $a \in \mathcal{A}$ is a dominant strategy of an agent $x \in [0,1)$ if, for any $s\in \mathcal{S}$,
                    \begin{align*}
                        U^{\left(s^{(x,a)}\right)}(x) \succcurlyeq_n
                        U^{\left(s\right)}(x) 
                        \text{,}
                    \end{align*}
                    where $s^{(x,a)} \in \mathcal{S}$ is defined such that, for any $y \in [0,1)$
                    \begin{align*}
                        s^{(x,a)}(y)
                        =
                        \begin{cases}
                            a &\text{ if } y = x
                            \\
                            s(y) &\text{ if } y \ne x
                        \end{cases}
                        \text{.}
                    \end{align*}
                \end{definition}

\section{Competitive Market Equilibrium}
\label{section:competitiveequilibirum}

    In this section, we will first consider some regularity conditions for the competitive market equilibrium to be well-defined.

    A competitive market equilibrium, in this section\footnote{The enforcement of a dominant strategy or a conditional dominant strategy is to facilitate the analysis and can be realized to a more general notion of competitive equilibrium in the future.}, will be when
    \begin{itemize}
        \item Every supplier uses a dominant strategy in the market and off the market;
        \item Every demander uses a dominant strategy in the market and off the market;
        \item Every mediator uses a sub-game equilibrium, such that, at each state, the mediators employ a betrayal-free-collusion-free equilibrium given that the next stage would be played with a betrayal-free-collusion-free equilibrium or plays some preferable action if only one agent has the agency in the next step;
        \item If resale occurs, then the non-winning mediator then the possible results set is assumed to be the set of all pseudo-equilibrium result;
        \item Every mediator gets the same quantity of scarcity from the suppliers.
    \end{itemize}

    \subsection{Competitive Market Conditions}

        \subsubsection{Supply Side}

            We will consider conditions for the supply side, in order to
            \begin{itemize}
                \item Ensure a unique existence of a dominant strategy for each supplier;
                \item Define the collection of the amount of supply that can be acquired by mediators offering the counter supply price if every supplier uses a dominant strategy within the market.
            \end{itemize}

            \paragraph{Supplier Dominant Strategy}
            We will first consider conditions for the unique existence of a dominant strategy for each supplier. 
    
            \begin{condition}
            \label{condition:1}
                (Supplier Individual Cutoff)
                
                A competitive market $\circled{C}^{ \left(\mathcal{A}_{\text{supply}}, \mathcal{A}_{\text{demand}},
            \bar{\Omega}_{\text{supply}}, \bar{\Omega}_{\text{demand}}, w,\omega,\tilde{w},\tilde{\omega}\right)}
            \left(
                \hat{\Pi}_{\text{supply}}, \hat{\Pi}_{\text{demand}}
            \right)$ has $\left(\mathcal{A}_{\text{supply}}, \bar{\Omega}_{\text{supply}}, w\right)$ such that, 
                there exists some measurable function $\ubar{\rho}: \bar{\Omega}_{\text{supply}} \to \mathbb{R}$ such that, for any $\left(h_0, h_1, v\right) \in \bar{\Omega}_{\text{supply}}$,
                \begin{align*}
                    &\ubar{\rho}(h_0, h_1, v)
                    =\\
                    &\inf\left(\left\{
                    c_2 \in \mathbb{R}
                    :
                    \sup_{c_1 \in (0,v]} \max_{a \in \mathcal{A}_{\text{supply}}}
                    \mathbb{E}_{\epsilon \sim \text{Unif}([0,1))}\left[w\left(h_0, a, \frac{c_1}{h_1}, c_2, \epsilon\right)
                    \right] > \max_{a \in \mathcal{A}_{\text{supply}}} \mathbb{E}_{\epsilon \sim \text{Unif}([0,1))}\left[w\left(h_0, a, 0, 0, \epsilon\right)
                    \right]
                    \right\}\right)
                    \text{,}
                \end{align*}
                \begin{align*}
                    \sup_{c_1 \in (0,v]} \max_{a \in \mathcal{A}_{\text{supply}}}
                    \mathbb{E}_{\epsilon \sim \text{Unif}([0,1))}\left[w\left(h_0, a, \frac{c_1}{h_1}, \ubar{\rho}(h_0, h_1, v), \epsilon\right)
                    \right] = \max_{a \in \mathcal{A}_{\text{supply}}} \mathbb{E}_{\epsilon \sim \text{Unif}([0,1))}\left[w\left(h_0, a, 0, 0, \epsilon\right)
                    \right]
                    \text{,}
                \end{align*}
                and
                \begin{align*}
                    \left\{\max_{a \in \mathcal{A}_{\text{supply}}}
                    \mathbb{E}_{\epsilon \sim \text{Unif}([0,1))}\left[w\left(h_0, a, \frac{c_1}{h_1}, c_2, \epsilon\right)
                    \right] 
                    \right\}_{
                    \substack{c_2 \in \left(-\infty,\ubar{\rho}\left(h_0, h_1, v\right)\right)\\c_1 \in (0,v]}
                    }
                    \subseteq \left(-\infty,\max_{a \in \mathcal{A}_{\text{supply}}} \mathbb{E}_{\epsilon \sim \text{Unif}([0,1))}\left[w\left(h_0, a, 0, 0, \epsilon\right)
                    \right]\right)
                    \text{.}
                \end{align*}
            \end{condition}

            \begin{corollary} (Lowest Supply Price Willing to Take)
                For a competitive market with the condition~\ref{condition:1}, the lowest supply price willing to take function $\ubar{\rho}: \bar{\Omega}_{\text{supply}} \to \mathbb{R}$ as defined in the condition statement is unique. We will use this as the definition for $\ubar{\rho}$.
            \end{corollary}
            \begin{proof}
                This follows directly from the definition.
            \end{proof}
    
            \begin{condition}
            \label{condition:2}
                (Supplier Stationary Optimal Volume)
                
                A competitive market $\circled{C}^{ \left(\mathcal{A}_{\text{supply}}, \mathcal{A}_{\text{demand}},
                \bar{\Omega}_{\text{supply}}, \bar{\Omega}_{\text{demand}},
                w,\omega,\tilde{w},\tilde{\omega}\right)}
                \left(
                    \hat{\Pi}_{\text{supply}}, \hat{\Pi}_{\text{demand}}
                \right)$ with the condition~\ref{condition:1} has $\left(\mathcal{A}_{\text{supply}}, \bar{\Omega}_{\text{supply}}, w\right)$ such that, there exists a measurable function $\bar{v}: \bar{\Omega}_{\text{supply}} \to \mathbb{R}^+$ such that, for any $\left(h_0, h_1, v\right) \in \bar{\Omega}_{\text{supply}}$,
                \begin{align*}
                    \left\{\bar{v}\left(h_0, h_1, v\right)\right\}
                    =
                    \bigcap_{c_2 \in \left[\ubar{\rho}\left(h_0, h_1, v\right),\infty\right)}
                    \argmax_{c_1 \in (0,v]} \max_{a \in \mathcal{A}_{\text{supply}}}
                    \mathbb{E}_{\epsilon \sim \text{Unif}([0,1))}\left[w\left(h_0, a, \frac{c_1}{h_1}, c_2, \epsilon\right)
                    \right]
                    \text{.}
                \end{align*}
            \end{condition}
            \begin{corollary} (Willing to Contract Supply Volume)
                For a competitive market with the conditions~\ref{condition:1},\ref{condition:2}, the willing to contract differential volume function $\bar{v}: \bar{\Omega}_{\text{supply}} \to \mathbb{R}^+$ as defined in the condition statement is unique. We will use this as the definition for $\bar{v}$.
            \end{corollary}
            \begin{proof}
                This follows directly from the definition.
            \end{proof}
        
            These $2$ conditions suggests that, for any supplier $s$ characterized as $\left(h_0, h_1, v\right) \in \bar{\Omega}_{\text{supply}}$, a dominant strategy uniquely exists to be the strategy of
            \begin{itemize}
                \item bidding volume $v_s = \bar{v}\left(h_0, h_1, v\right)$;
                \item bidding supply price $\rho_s = \ubar{\rho}\left(h_0, h_1, v\right)$.
            \end{itemize}

            \paragraph{Real Supply}
    
            \begin{condition}
            \label{condition:3}
                (Finite Supply)
        
                A competitive market $\circled{C}^{ \left(\mathcal{A}_{\text{supply}}, \mathcal{A}_{\text{demand}},
            \bar{\Omega}_{\text{supply}}, \bar{\Omega}_{\text{demand}}, w,\omega,\tilde{w},\tilde{\omega}\right)}
            \left(
                \hat{\Pi}_{\text{supply}}, \hat{\Pi}_{\text{demand}}
            \right)$ with the conditions~\ref{condition:1},\ref{condition:2} has $\left(\mathcal{A}_{\text{supply}}, \bar{\Omega}_{\text{supply}}, w, \hat{\Pi}_{\text{supply}}\right)$ such that
                \begin{align*}
                    \mathbb{E}_{\left(h_0, h_1, v\right) \sim \hat{\Pi}_{\text{supply}}}\left[
                    \bar{v}\left(h_0, h_1, v\right)
                    \right]
                    \in \mathbb{R}^+
                    \text{.}
                \end{align*}

            \end{condition}
    
            We will define a ``real supply" to be a correspondence mapping a supply price to a supply volume.

            \begin{definition} (Real Supply)
                For a competitive market with the conditions~\ref{condition:1},\ref{condition:2},\ref{condition:3}, the real supply correspondence $\tilde{S}: \mathbb{R} \rightrightarrows \mathbb{R}^+_0$ is defined such that, for any $\rho \in \mathbb{R}$,
                \begin{align*}
                    \tilde{S}(\rho)
                    =
                    \left[
                    \mathbb{E}_{\left(h_0, h_1, v\right) \sim \hat{\Pi}_{\text{supply}}}\left[
                        \bar{v}\left(h_0, h_1, v\right)
                        \mathbf{1}_{\ubar{\rho}\left(h_0, h_1, v\right) < \rho}
                        \right],
                    \mathbb{E}_{\left(h_0, h_1, v\right) \sim \hat{\Pi}_{\text{supply}}}\left[
                        \bar{v}\left(h_0, h_1, v\right)
                        \mathbf{1}_{\ubar{\rho}\left(h_0, h_1, v\right) \le \rho}
                        \right]
                    \right]
                    \text{.}
                \end{align*}
            \end{definition}

            It can also be shown that the real supply correspondence $\tilde{S}$ is a maximal monotone operator and the limits of $\tilde{S}$ exist and are finite.

            \begin{corollary}
            \label{corollary:supplymaximalmonotone}
                For a competitive market with the conditions~\ref{condition:1},\ref{condition:2},\ref{condition:3}, the real supply correspondence $\tilde{S}: \mathbb{R} \rightrightarrows \mathbb{R}^+_0$ is maximal monotone, meaning that, for any $\rho \in \mathbb{R}$
                \begin{enumerate}
                    \item $\tilde{S}(\rho)$ is a closed compact interval on $\mathbb{R}^+_0$;
                    \item $
                            \lim_{\rho' \to \rho^+} \min\left(\tilde{S}\left(\rho'\right)\right) 
                            =
                            \max\left(\tilde{S}\left(\rho\right)\right) 
                        $;
                    \item $
                            \lim_{\rho' \to \rho^-} \max\left(\tilde{S}\left(\rho'\right)\right) 
                            =
                            \min\left(\tilde{S}\left(\rho\right)\right) 
                         $.
                \end{enumerate}
                Moreover, $\lim_{\rho \to -\infty} \max\left(\tilde{S}\left(\rho\right)\right) = 0$, and there exists a unique $S_\text{max} \in \mathbb{R}^+$ such that $\lim_{\rho \to \infty} \min\left(\tilde{S}\left(\rho\right)\right) = S_\text{max}$.
                We will define the supremum supply $S_\text{max}$ as such.
            \end{corollary}
            \begin{proof}
                For any $\rho \in \mathbb{R}$, we have that $\lim_{\rho' \to \rho^+} \mathbb{E}_{\left(h_0, h_1, v\right) \sim \hat{\Pi}_{\text{supply}}}\left[\bar{v}\left(h_0, h_1, v\right)\mathbf{1}_{\ubar{\rho}\left(h_0, h_1, v\right) < \rho'}\right]
                = \\ \mathbb{E}_{\left(h_0, h_1, v\right) \sim \hat{\Pi}_{\text{supply}}}\left[\bar{v}\left(h_0, h_1, v\right)\mathbf{1}_{\ubar{\rho}\left(h_0, h_1, v\right) \ge \rho}\right]$. Similarly, for any $\rho \in\mathbb{R}$, $\lim_{\rho' \to \rho^+} \min\left(\tilde{S}\left(\rho'\right)\right) =
                \max\left(\tilde{S}\left(\rho\right)\right)$. Thus, the correspondence $\tilde{S}$ is a maximal monotone operator.

                From the condition\ref{condition:1},
                $\lim_{\rho \to -\infty} \max\left(\tilde{S}\left(\ubar{\rho}\right)\right) = 0$, and 
                $\lim_{\rho \to -\infty} \min\left(\tilde{S}\left(\ubar{\rho}\right)\right)= \mathbb{E}_{\left(h_0, h_1, v\right) \sim \hat{\Pi}_{\text{supply}}}\left[\bar{v}\left(h_0, h_1, v\right)\right] \in \mathbb{R}^+$, because of the condition~\ref{condition:3}.
            \end{proof}

        \subsubsection{Demand}

            Similarly, we want to both ensure the unique existence of a dominant strategy for each mediator and well define the demand volume at different demand price.

            \paragraph{Demander Dominant Strategy}
            The conditions will be analogous to those of the suppliers due to the continuum limit. However, the difference will be more apparent when a non-limit market is considered. In a non-limit market, a supplier can get a dominant strategy by further impose that there exists a supply price where the corresponding supplier will be indifferent between securing and not securing a contract.\footnote{However, there is no dominant strategy in a non-limit market without imposing demanding conditions on the demander utility  or without imposing a restriction on the strategies employed by other agents in the system, such as to not allow the resale bidding.}
            
            \begin{condition}
            \label{condition:4}
                (Demander Individual Cutoff)
                
                A competitive market $\circled{C}^{ \left(\mathcal{A}_{\text{supply}}, \mathcal{A}_{\text{demand}},
            \bar{\Omega}_{\text{supply}}, \bar{\Omega}_{\text{demand}}, w,\omega,\tilde{w},\tilde{\omega}\right)}
            \left(
                \hat{\Pi}_{\text{supply}}, \hat{\Pi}_{\text{demand}}
            \right)$ has $\left(\mathcal{A}_{\text{demand}}, \bar{\Omega}_{\text{demand}}, \omega\right)$ such that, 
                there exists some $\bar{r}: \bar{\Omega}_{\text{demand}} \to \mathbb{R}$ such that, for any $\left(\eta_0, \eta_1\right) \in \bar{\Omega}_{\text{demand}}$,
                \begin{align*}
                    &\bar{r}\left(\eta_0, \eta_1\right)
                    =\\
                    &\sup\left(\left\{\zeta_2 \in \mathbb{R}: \sup_{\zeta_1 \in \mathbb{R}^+} \max_{a \in \mathcal{A}_{\text{demand}}}
                    \mathbb{E}_{\epsilon \sim \text{Unif}([0,1))}\left[\omega\left(\eta_0, a, \frac{\zeta_1}{\eta_1}, \zeta_2, \epsilon\right)
                    \right] > \max_{a \in \mathcal{A}_{\text{demand}}} \mathbb{E}_{\epsilon \sim \text{Unif}([0,1))}\left[\omega\left(\eta_0, a, 0, 0, \epsilon\right)
                    \right]
                    \right\}\right)
                    \text{,}
                \end{align*}
                \begin{align*}
                    \sup_{\zeta_1 \in \mathbb{R}^+} \max_{a \in \mathcal{A}_{\text{demand}}}
                    \mathbb{E}_{\epsilon \sim \text{Unif}([0,1))}\left[\omega\left(\eta_0, a, \frac{\zeta_1}{\eta_1}, \bar{r}\left(\eta_0, \eta_1\right), \epsilon\right)
                    \right] = \max_{a \in \mathcal{A}_{\text{demand}}} \mathbb{E}_{\epsilon \sim \text{Unif}([0,1))}\left[\omega\left(\eta_0, a, 0, 0, \epsilon\right)
                    \right]
                    \text{,}
                \end{align*}
                and
                \begin{align*}
                    \left\{\max_{a \in \mathcal{A}_{\text{demand}}}
                    \mathbb{E}_{\epsilon \sim \text{Unif}([0,1))}\left[\omega\left(\eta_0, a, \frac{\zeta_1}{\eta_1}, \zeta_2, \epsilon\right)
                    \right] 
                    \right\}_{
                    \substack{\zeta_2 \in \left( \bar{r}\left(\eta_0, \eta_1\right), \infty\right)\\ \zeta_1 \in \mathbb{R}^+}
                    }
                    \subseteq \left(-\infty,\max_{a \in \mathcal{A}_{\text{supply}}} \mathbb{E}_{\epsilon \sim \text{Unif}([0,1))}\left[w\left(h_0, a, 0, 0, \epsilon\right)
                    \right]\right)
                    \text{.}
                \end{align*}
            \end{condition}

            \begin{corollary} (Highest Demand Price Willing to Take)
                For a competitive market with the condition~\ref{condition:4}, the highest demand price willing to take function $\bar{r}: \bar{\Omega}_{\text{demand}} \to \mathbb{R}$ as defined in the condition statement is unique. We will use this as the definition for $\bar{r}$.
            \end{corollary}
            \begin{proof}
                This follows directly from the definition.
            \end{proof}
        
            \begin{condition}
            \label{condition:5}
                (Demander Stationary Optimal Volume)
                
                A competitive market $\circled{C}^{ \left(\mathcal{A}_{\text{supply}}, \mathcal{A}_{\text{demand}},
            \bar{\Omega}_{\text{supply}}, \bar{\Omega}_{\text{demand}}, w,\omega,\tilde{w},\tilde{\omega}\right)}
            \left(
                \hat{\Pi}_{\text{supply}}, \hat{\Pi}_{\text{demand}}
            \right)$ with the condition~\ref{cdomntion:4} has $\left(\mathcal{A}_{\text{demand}}, \bar{\Omega}_{\text{demand}}, \omega\right)$ such that 
                there exists a measurable function $\ubar{v}: \bar{\Omega}_{\text{demand}} \to \mathbb{R}^+$ such that, for any $\left(\eta_0, \eta_1\right) \in \bar{\Omega}_{\text{demand}}$,
                \begin{align*}
                    \left\{\ubar{v}\left(\eta_0, \eta_1\right)\right\}
                    =
                    \bigcap_{\zeta_2 \in \left(-\infty, \bar{r}\left(\eta_0, \eta_1\right)\right)}
                    \argmax_{\zeta_1 \in \mathbb{R}^+} \max_{a \in \mathcal{A}_{\text{demand}}}
                    \mathbb{E}_{\epsilon \sim \text{Unif}([0,1))}\left[\omega\left(\eta_0, a, \frac{\zeta_1}{\eta_1}, \zeta_2, \epsilon\right)
                    \right]
                    \text{.}
                \end{align*}
            \end{condition}
            \begin{corollary} (Willing to Contract Demand Volume)
                For a competitive market with the conditions~\ref{condition:4},\ref{condition:5}, the willing to contract differential volume function $\ubar{v}: \bar{\Omega}_{\text{demand}} \to \mathbb{R}^+$ as defined in the condition statement is unique. We will use this as the definition for $\ubar{v}$.
            \end{corollary}
            \begin{proof}
                This follows directly from the definition.
            \end{proof}

            These $2$ conditions suggests that, conditional on that there is no resale in the market, for any demander $d$ characterized as $\left(\eta_0, \eta_1\right) \in \bar{\Omega}_{\text{demand}}$, a dominant strategy uniquely exists to be the strategy of
            \begin{itemize}
                \item bidding volume $v_d = \ubar{v}\left(\eta_0, \eta_1\right)$;
                \item bidding demand price $r_d = \bar{r}\left(\eta_0, \eta_1\right)$.
            \end{itemize}

            \paragraph{Real Demand}

            We will impose an additional condition for the total aggregated demand to be finite in order to define a ``real demand" to be a correspondence being a convex hull of the amount of the demand volume from demanders that will weakly prefer securing a contract at demand price $r$ and the individual optimal volume bid over getting zero contract and the amount of the demand volume from demanders that will strictly prefer securing a contract at demand price $r$ and the individual optimal volume bid over getting zero contract.

            \begin{condition}
            \label{condition:6}
                (Finite Demand)
        
                A competitive market $\circled{C}^{ \left(\mathcal{A}_{\text{supply}}, \mathcal{A}_{\text{demand}},
            \bar{\Omega}_{\text{supply}}, \bar{\Omega}_{\text{demand}}, w,\omega,\tilde{w},\tilde{\omega}\right)}
            \left(
                \hat{\Pi}_{\text{supply}}, \hat{\Pi}_{\text{demand}}
            \right)$ with the conditions~\ref{condition:4},\ref{condition:5} has $\left(\mathcal{A}_{\text{demand}}, \bar{\Omega}_{\text{demand}}, \omega, \hat{\Pi}_{\text{demand}}\right)$ such that
                \begin{align*}
                    \mathbb{E}_{\left(\eta_0, \eta_1\right) \sim \hat{\Pi}_{\text{demand}}}
                    \left[\ubar{v}\left(\eta_0, \eta_1\right) 
                    \right]
                    \in \mathbb{R}^+
                    \text{.}
                \end{align*}
            \end{condition}

            \begin{definition} 
            \label{definition:realdemand}
            (Real Demand)
                For a competitive market with the conditions~\ref{condition:4},\ref{condition:5},\ref{condition:6}, the real demand correspondence $\tilde{D}: \mathbb{R} \rightrightarrows \mathbb{R}^+_0$ is defined such that, for any $r \in \mathbb{R}$,
                \begin{align*}
                    \tilde{D}(r)
                    =
                    \left[\mathbb{E}_{\left(\eta_0, \eta_1\right) \sim \hat{\Pi}_{\text{demand}}}
                    \left[\ubar{v}\left(\eta_0, \eta_1\right) 
                    \mathbf{1}_{\bar{r}\left(\eta_0, \eta_1\right) > r} \right], \mathbb{E}_{\left(\eta_0, \eta_1\right) \sim \hat{\Pi}_{\text{demand}}}
                    \left[\ubar{v}\left(\eta_0, \eta_1\right) 
                    \mathbf{1}_{\bar{r}\left(\eta_0, \eta_1\right) \ge r} \right]\right]
                    \text{.}
                \end{align*}
            \end{definition}

            Similar to the real supply, the negative of the real demand will be maximal monotone and bounded.

            \begin{corollary}
            \label{corollary:demandmaximalmonotone}
                For a competitive market with the conditions~\ref{condition:4},\ref{condition:5},\ref{condition:6}, the negative of the real demand correspondence $\tilde{D}: \mathbb{R} \rightrightarrows \mathbb{R}^+_0$ is maximal monotone, $\lim_{r \to \infty} \max\left(\tilde{D}\left(r\right)\right) = 0$, and there exists a unique $D_\text{max} \in \mathbb{R}^+$ such that $\lim_{r \to -\infty} \min\left(\tilde{D}\left(r\right)\right) = D_\text{max}$.
                We will define the supremum demand $D_\text{max}$ as such.
            \end{corollary}
            \begin{proof}
                The proof is equivalent to that of the corollary~\ref{corollary:supplymaximalmonotone}.
            \end{proof}

        \subsubsection{Relations}

            So far, we have considered the ``real supply" and ``real demand" as correspondences of the supply price and the demand price, respectively. However, in the equilibrium analysis, the mediators are utility maximizing, so it is important to compare the supply and demand in the space that is related to the mediator utility. 

            From the subsection~\ref{subsection:independent off-market}, we restrict the attention of the mediator utility realization to be a sum of the utility created by each contract. Thus, we will consider 
            \begin{itemize}
                \item the aggregated supply volume that can be secured by incurring a specific value of the ``supply cost" per unit of supply;
                \item the aggregated demand volume that can be secured by getting a specific value of the ``demand revenue" per unit of demand.
            \end{itemize}

            \paragraph{Off-Market Irrelevance}

            After the market has finalized the physically feasible transactions, each supplier and demander will perform an additional action outside of the market. Since we consider the restriction where every supplier and demander employs dominant strategy both inside the market and outside the market, we have that every supplier and demander will select an optimal action conditional on whether a contract is secured and how the secured contract is specified. In some cases, there can be more than one optimal choices, and they can affect the utility of the counter party mediator differently. Therefore, an additional is required to ensure the uniqueness of the relevant outcome of the mediators.\footnote{The additional assumption can be placed directly upon the aggregated level, such as by having that the difference between any pair of optimal off-market choices will happen with a specific type of supplier that can be drawn with probability $0$.}

            \begin{condition}
            \label{condition:7}
                (Supplier Off-Market Irrelevance)

                A competitive market $\circled{C}^{ \left(\mathcal{A}_{\text{supply}}, \mathcal{A}_{\text{demand}},
            \bar{\Omega}_{\text{supply}}, \bar{\Omega}_{\text{demand}}, w,\omega,\tilde{w},\tilde{\omega}\right)}
            \left(
                \hat{\Pi}_{\text{supply}}, \hat{\Pi}_{\text{demand}}
            \right)$ with conditions~\ref{condition:1},\ref{condition:2} has $\left(\mathcal{A}_{\text{supply}}, \bar{\Omega}_{\text{supply}}, w, \tilde{w}\right)$ such that, for any $\left(h_0, h_1, v\right) \in \bar{\Omega}_{\text{supply}}$,
                there exists some measurable function $\bar{w}^{\left(h_0, h_1, v\right)}:  \left[\ubar{\rho}\left(h_0, h_1, v\right),\infty\right) \to \mathbb{R}$ such that, for any $\rho \in \left[\ubar{\rho}\left(h_0, h_1, v\right),\infty]\right)$,
                \begin{align*}
                    \bar{w}^{\left(h_0, h_1, v\right)}(\rho)
                    =\max_{a \in
                    \argmax_{a' \in \mathcal{A}_{\text{supply}}}
                    \mathbb{E}_{\epsilon \sim \text{Unif}([0,1))}\left[w\left(h_0, a', \frac{\ubar{v}}{h_1}\left(h_0, h_1, v\right), \rho, \epsilon\right)
                    \right] 
                    }
                    \mathbb{E}_{\epsilon \sim \text{Unif}([0,1))}\left[\tilde{w}\left(h_0, a, \frac{\ubar{v}}{h_1}\left(h_0, h_1, v\right), \rho, \epsilon\right)
                    \right]
                    \text{,}
                \end{align*}
                and
                \begin{align*}
                    \bar{w}^{\left(h_0, h_1, v\right)}(\rho)
                    =
                    \min_{a \in
                    \argmax_{a' \in \mathcal{A}_{\text{supply}}}
                    \mathbb{E}_{\epsilon \sim \text{Unif}([0,1))}\left[w\left(h_0, a', \frac{\ubar{v}\left(h_0, h_1, v\right)}{h_1}, \rho, \epsilon\right)
                    \right] 
                    }
                    \mathbb{E}_{\epsilon \sim \text{Unif}([0,1))}\left[\tilde{w}\left(h_0, a, \frac{\ubar{v}\left(h_0, h_1, v\right)}{h_1}, \rho, \epsilon\right)
                    \right]
                    \text{.}
                \end{align*}
            \end{condition}
            \begin{corollary}
                For a competitive market with the conditions~\ref{condition:1},\ref{condition:2},\ref{condition:7}, for any $\left(h_0, h_1, v\right) \in \bar{\Omega}_{\text{supply}}$, the function $\bar{w}^{\left(h_0, h_1, v\right)}$ as defined in the condition statement is unique. We will use this as the definition for $\left(\bar{w}^{\left(h_0, h_1, v\right)}\right)_{\left(h_0, h_1, v\right) \in \bar{\Omega}_{\text{supply}}}$.
            \end{corollary}
            \begin{proof}
                This follows directly from the definition.
            \end{proof}
            
            \begin{condition}
            \label{condition:8}
                (Demander Off-Market Irrelevance)

                A competitive market $\circled{C}^{ \left(\mathcal{A}_{\text{supply}}, \mathcal{A}_{\text{demand}},
            \bar{\Omega}_{\text{supply}}, \bar{\Omega}_{\text{demand}}, w,\omega,\tilde{w},\tilde{\omega}\right)}
            \left(
                \hat{\Pi}_{\text{supply}}, \hat{\Pi}_{\text{demand}}
            \right)$ with conditions~\ref{condition:4},\ref{condition:5} has $\left(\mathcal{A}_{\text{demand}}, \bar{\Omega}_{\text{demand}}, \omega, \tilde{\omega}\right)$ such that, for any $\left(\eta_0, \eta_1\right) \in \bar{\Omega}_{\text{demand}}$,
                there exists some measurable function $\bar{\omega}^{\left(\eta_0, \eta_1\right)}:  \left(-\infty,\bar{r}\left(\eta_0, \eta_1\right)\right] \to \mathbb{R}$ such that, for any $r \in \left(-\infty,\bar{r}\left(\eta_0, \eta_1\right)\right]$,
                \begin{align*}
                    \bar{\omega}^{\left(\eta_0, \eta_1\right)}(r)
                    =\max_{a \in
                    \argmax_{a' \in \mathcal{A}_{\text{demand}}}
                    \mathbb{E}_{\epsilon \sim \text{Unif}([0,1))}\left[\omega\left(\eta_0, a', \frac{\bar{v}\left(\eta_0, \eta_1\right)}{\eta_1}, r, \epsilon\right)
                    \right] 
                    }
                    \mathbb{E}_{\epsilon \sim \text{Unif}([0,1))}\left[\bar{\omega}\left(\eta_0, a', \frac{\bar{v}\left(\eta_0, \eta_1\right)}{\eta_1}, r, \epsilon\right)
                    \right] 
                    \text{,}
                \end{align*}
                and
                \begin{align*}
                    \bar{\omega}^{\left(\eta_0, \eta_1\right)}(r)
                    =
                    \min_{a \in
                    \argmax_{a' \in \mathcal{A}_{\text{demand}}}
                    \mathbb{E}_{\epsilon \sim \text{Unif}([0,1))}\left[\omega\left(\eta_0, a', \frac{\bar{v}\left(\eta_0, \eta_1\right)}{\eta_1}, r, \epsilon\right)
                    \right] 
                    }
                    \mathbb{E}_{\epsilon \sim \text{Unif}([0,1))}\left[\bar{\omega}\left(\eta_0, a', \frac{\bar{v}\left(\eta_0, \eta_1\right)}{\eta_1}, r, \epsilon\right)
                    \right] 
                    \text{.}
                \end{align*}
            \end{condition}
            \begin{corollary}
                For a competitive market with the conditions~\ref{condition:4},\ref{condition:5},\ref{condition:8}, for any $\left(\eta_0, \eta_1\right) \in \bar{\Omega}_{\text{demand}}$, the function $\bar{\omega}^{\left(\eta_0, \eta_1\right)}$ as defined in the condition statement is unique. We will use this as the definition for $\left(\bar{\omega}^{\left(\eta_0, \eta_1\right)}\right)_{\left(\eta_0, \eta_1\right) \in \bar{\Omega}_{\text{demand}}}$.
            \end{corollary}
            \begin{proof}
                This follows directly from the definition.
            \end{proof}

            \paragraph{Supply Cost}

            \begin{condition}
            \label{condition:9}
                (Finite Supply Cost)
        
                A competitive market $\circled{C}^{ \left(\mathcal{A}_{\text{supply}}, \mathcal{A}_{\text{demand}},
            \bar{\Omega}_{\text{supply}}, \bar{\Omega}_{\text{demand}}, w,\omega,\tilde{w},\tilde{\omega}\right)}
            \left(
                \hat{\Pi}_{\text{supply}}, \hat{\Pi}_{\text{demand}}
            \right)$ with conditions~\ref{condition:1},\ref{condition:2},\ref{condition:3},\ref{condition:7} has $\left(\mathcal{A}_{\text{supply}}, \bar{\Omega}_{\text{supply}}, w, \tilde{w}, \hat{\Pi}_{\text{supply}}\right)$ such that, for any $\rho \in \mathbb{R}$,
                \begin{align*}
                    \lim_{x \to \left(
                    \mathbb{E}_{\left(h_0, h_1, v\right) \sim \hat{\Pi}_{\text{supply}}}
                    \left[\bar{v}\left(h_0, h_1, v\right)
                    \mathbf{1}_{\ubar{\rho}\left(h_0, h_1, v\right) \le \rho}
                    \right]
                    \right)^+}
                    \frac{
                    \mathbb{E}_{\left(h_0, h_1, v\right) \sim \hat{\Pi}_{\text{supply}}}
                    \left[h_1\bar{w}^{\left(h_0, h_1, v\right)}(\rho)
                    \mathbf{1}_{\ubar{\rho}\left(h_0, h_1, v\right) \le \rho}
                    \right]
                    }{x}
                    \in \mathbb{R}
                    \text{,}
                \end{align*}
                and
                \begin{align*}
                    \lim_{x \to \left(
                    \mathbb{E}_{\left(h_0, h_1, v\right) \sim \hat{\Pi}_{\text{supply}}}
                    \left[\bar{v}\left(h_0, h_1, v\right)
                    \mathbf{1}_{\ubar{\rho}\left(h_0, h_1, v\right) < \rho}
                    \right]
                    \right)^+}
                    \frac{
                    \mathbb{E}_{\left(h_0, h_1, v\right) \sim \hat{\Pi}_{\text{supply}}}
                    \left[h_1\bar{w}^{\left(h_0, h_1, v\right)}(\rho)
                    \mathbf{1}_{\ubar{\rho}\left(h_0, h_1, v\right) < \rho}
                    \right]
                    }{x}
                    \in \mathbb{R}
                    \text{.}
                \end{align*}
            \end{condition}

            From this condition, we can define $2$ cost functions, $\bar{p}_{\text{supply}}: \mathbb{R} \to \mathbb{R}$ and $\ubar{p}_{\text{supply}}: \mathbb{R} \to \mathbb{R}$,  upon the supply price space. 

            \begin{definition} (Supply Cost)
                For a competitive market with the conditions~\ref{condition:1},\ref{condition:2},\ref{condition:3},\ref{condition:7},\ref{condition:9}, the supply cost functions $\bar{p}_{\text{supply}}: \mathbb{R} \to \mathbb{R}$ and $\ubar{p}_{\text{supply}}: \mathbb{R} \to \mathbb{R}$ are defined such that, for any $\rho \in \mathbb{R}$,
                \begin{align*}
                    \bar{p}_{\text{supply}}(\rho)
                    =
                    -\lim_{x \to \left(
                        \mathbb{E}_{\left(h_0, h_1, v\right) \sim \hat{\Pi}_{\text{supply}}}
                        \left[\bar{v}\left(h_0, h_1, v\right)
                        \mathbf{1}_{\ubar{\rho}\left(h_0, h_1, v\right) \le \rho}
                        \right]
                        \right)^+}
                        \frac{
                        \mathbb{E}_{\left(h_0, h_1, v\right) \sim \hat{\Pi}_{\text{supply}}}
                        \left[h_1\bar{w}^{\left(h_0, h_1, v\right)}(\rho)
                        \mathbf{1}_{\ubar{\rho}\left(h_0, h_1, v\right) \le \rho}
                        \right]
                        }{x}
                    \text{,}
                \end{align*}
                and
                \begin{align*}
                    \ubar{p}_{\text{supply}}(\rho)
                    =
                    -\lim_{x \to \left(
                        \mathbb{E}_{\left(h_0, h_1, v\right) \sim \hat{\Pi}_{\text{supply}}}
                        \left[\bar{v}\left(h_0, h_1, v\right)
                        \mathbf{1}_{\ubar{\rho}\left(h_0, h_1, v\right) < \rho}
                        \right]
                        \right)^+}
                        \frac{
                        \mathbb{E}_{\left(h_0, h_1, v\right) \sim \hat{\Pi}_{\text{supply}}}
                        \left[h_1\bar{w}^{\left(h_0, h_1, v\right)}(\rho)
                        \mathbf{1}_{\ubar{\rho}\left(h_0, h_1, v\right) < \rho}
                        \right]
                        }{x}
                    \text{.}
                \end{align*}
            \end{definition}

            Although the supply cost will be specified by two functions, the two functions are the same at every but countable $\rho \in \mathbb{R}$.

            \begin{corollary}
            \label{corollary:equalsupplyae}
            There exists countably many $\rho \in \mathbb{R}$ where $\bar{p}_{\text{supply}}(\rho) \ne \ubar{p}_{\text{supply}}(\rho)$, and
            \begin{align*}
                \left\{\rho \in \mathbb{R}:\bar{p}_{\text{supply}}(\rho) \ne \ubar{p}_{\text{supply}}(\rho)\right\}
                \subseteq
                \left\{\rho \in \mathbb{R}: 
                \max\left(\tilde{S}(\rho)\right) \ne \min\left(\tilde{S}(\rho)\right)
                \right\}
            \end{align*}
            \end{corollary}
            \begin{proof}
                For any $\rho \in \mathbb{R}$, if $\max\left(\tilde{S}(\rho)\right) = \min\left(\tilde{S}(\rho)\right)$, then $\mathbb{E}_{\left(h_0, h_1, v\right) \sim \hat{\Pi}_{\text{supply}}} \left[h_1\bar{w}^{\left(h_0, h_1, v\right)}(\rho) \mathbf{1}_{\ubar{\rho}\left(h_0, h_1, v\right) < \rho} \right] = \mathbb{E}_{\left(h_0, h_1, v\right) \sim \hat{\Pi}_{\text{supply}}} \left[h_1\bar{w}^{\left(h_0, h_1, v\right)}(\rho) \mathbf{1}_{\ubar{\rho}\left(h_0, h_1, v\right) \le \rho} \right]$, so $\bar{p}_{\text{supply}}(\rho) = \ubar{p}_{\text{supply}}(\rho)$. The set $\left\{\rho \in \mathbb{R}: 
                \max\left(\tilde{S}(\rho)\right) \ne \min\left(\tilde{S}(\rho)\right)
            \right\}$ is countable since the correspondence $\tilde{S}$ is bounded and maximal monotone (as shown in the corollary~\ref{corollary:supplymaximalmonotone}).
            \end{proof}

            Note that, 
            \begin{align*}
                \left\{\rho \in \mathbb{R}:\bar{p}_{\text{supply}}(\rho) \ne \ubar{p}_{\text{supply}}(\rho)\right\}
                \subseteq
                \left\{\rho \in \mathbb{R}: 
                \max\left(\tilde{S}(\rho)\right) \ne \min\left(\tilde{S}(\rho)\right)
                \right\}
                \text{,}
            \end{align*}
            which is countable as a result of the condition~\ref{condition:3}.

            We will then consider conditions on the aggregated level.

            One important condition is that there is a positive cost to acquire positive quantity of scarcity.
            
            \begin{condition}
            \label{condition:10}
                (No Free Supply)
        
                A competitive market $\circled{C}^{ \left(\mathcal{A}_{\text{supply}}, \mathcal{A}_{\text{demand}},
            \bar{\Omega}_{\text{supply}}, \bar{\Omega}_{\text{demand}}, w,\omega,\tilde{w},\tilde{\omega}\right)}
            \left(
                \hat{\Pi}_{\text{supply}}, \hat{\Pi}_{\text{demand}}
            \right)$ with conditions~\ref{condition:1},\ref{condition:2},\ref{condition:3},\ref{condition:7},\ref{condition:9} has $\left(\mathcal{A}_{\text{supply}}, \bar{\Omega}_{\text{supply}}, w, \tilde{w}, \hat{\Pi}_{\text{supply}}\right)$ such that, for any $\rho \in \mathbb{R}$,
                \begin{itemize}
                    \item If $\max\left(\tilde{S}(\rho)\right) > 0$, then $\bar{p}_{\text{supply}}(\rho) >0$;
                    \item If $\min\left(\tilde{S}(\rho)\right) > 0$, then $\ubar{p}_{\text{supply}}(\rho) >0$.
                \end{itemize}
            \end{condition}

            This condition is important for the definition of the competitive equilibrium to be simple.

            \begin{condition}
            \label{condition:11}
                (Supply Monotone Trend)
        
                A competitive market $\circled{C}^{ \left(\mathcal{A}_{\text{supply}}, \mathcal{A}_{\text{demand}},
            \bar{\Omega}_{\text{supply}}, \bar{\Omega}_{\text{demand}}, w,\omega,\tilde{w},\tilde{\omega}\right)}
            \left(
                \hat{\Pi}_{\text{supply}}, \hat{\Pi}_{\text{demand}}
            \right)$  with conditions~\ref{condition:1},\ref{condition:2},\ref{condition:3},\ref{condition:7},\ref{condition:9} has $\left(\mathcal{A}_{\text{supply}}, \bar{\Omega}_{\text{supply}}, w, \tilde{w}, \hat{\Pi}_{\text{supply}}\right)$ such that, for any $\rho_1, \rho_2 \in \mathbb{R}$ with $\rho_1 > \rho_2$,
                \begin{itemize}
                    \item If $\max\left(\tilde{S}(\rho_1)\right) = \max\left(\tilde{S}(\rho_2)\right) > 0$, then $\bar{p}_{\text{supply}}(\rho_1) > \bar{p}_{\text{supply}}(\rho_2)$;
                    \item If $\min\left(\tilde{S}(\rho_1)\right) = \min\left(\tilde{S}(\rho_2)\right) > 0$, then $\ubar{p}_{\text{supply}}(\rho_1) > \ubar{p}_{\text{supply}}(\rho_2)$.
                \end{itemize}
            \end{condition}

            This condition does not make the definition of the competitive equilibrium easier. However, it is an important economic condition for the competitive equilibrium analysis.

            Although these two conditions may not seem natural, it can be conveniently satisfied in most market of interest. Next, we will show that, by enforcing some individual structure regularity, the condition~\ref{condition:10} and the condition~\ref{condition:11} can be satisfied even without the further knowledge of the distribution $\hat{\Pi}_{\text{supply}}$.

            \begin{proposition}
            \label{proposition:randomfunction}
                If a competitive market $\circled{C}^{ \left(\mathcal{A}_{\text{supply}}, \mathcal{A}_{\text{demand}},
                \bar{\Omega}_{\text{supply}}, \bar{\Omega}_{\text{demand}}, w,\omega,\tilde{w},\tilde{\omega}\right)}
                \left(
                    \hat{\Pi}_{\text{supply}}, \hat{\Pi}_{\text{demand}}
                \right)$ with conditions~\ref{condition:1},\ref{condition:2},\ref{condition:3},\ref{condition:7},\ref{condition:9} has $\left(\bar{w}^{\left(h_0, h_1, v\right)}\right)_{\left(h_0, h_1, v\right) \in \bar{\Omega}_{\text{supply}}}$ such that, for any $\left(h_0, h_1, v\right) \in \bar{\Omega}_{\text{supply}}$, the function $\bar{w}^{\left(h_0, h_1, v\right)}$ is positive and increasing, then the condition~\ref{condition:10} and the condition~\ref{condition:11} are satisfied.
            \end{proposition}
            \begin{proof}
                By defining $P = \left\{\rho\in\mathbb{R}: \max\left(\tilde{S}\left(\rho\right)\right) > 0\right\}$, it suffices to show that $\mathbb{E}_{\left(h_0, h_1, v\right) \sim \hat{\Pi}_{\text{supply}}} \left[h_1\bar{w}^{\left(h_0, h_1, v\right)}(\rho)
                \mathbf{1}_{\ubar{\rho}\left(h_0, h_1, v\right) \le \rho}\right]$ is positive and increasing on $\rho \in P$, which would be an open interval or a closed interval, and would be a half-space of $\mathbb{R}$.

                For any $\rho_0 \in P$, we define
                $\bar{\Omega}_{\text{supply}}^{(0)} = \left\{
                x \in \bar{\Omega}_{\text{supply}}: \ubar{\rho}(x) \le \rho_0 \right\}$. 
                
                Therefore, the value $\mathbb{E}_{\left(h_0, h_1, v\right) \sim \hat{\Pi}_{\text{supply}}} \left[h_1\bar{w}^{\left(h_0, h_1, v\right)}(\rho_0)
                \mathbf{1}_{\ubar{\rho}\left(h_0, h_1, v\right) \le \rho_0}\right]$ is an expectation of a positive random variable over the space $\bar{\Omega}_{\text{supply}}^{(0)}$. Thus, the function is positive.

                For any $\rho_1, \rho_2 \in P$ with $\rho_1 > \rho_2$ and $\max\left(\tilde{S}(\rho_1)\right)=\max\left(\tilde{S}(\rho_2)\right)$, we define $\bar{\Omega}_{\text{supply}}^{(1)} = \left\{
                x \in \bar{\Omega}_{\text{supply}}: \ubar{\rho}(x) \le \rho_1 \right\}$, and $\bar{\Omega}_{\text{supply}}^{(2)} = \left\{
                x \in \bar{\Omega}_{\text{supply}}: \ubar{\rho}(x) \le \rho_2 \right\}$. Therefore, $\bar{\Omega}_{\text{supply}}^{(2)} \subseteq \bar{\Omega}_{\text{supply}}^{(1)}$ and $\hat{\Pi}_{\text{supply}}\left(\bar{\Omega}_{\text{supply}}^{(1)} - \bar{\Omega}_{\text{supply}}^{(2)}\right) = 0$.
                
                Therefore, the value
                $\mathbb{E}_{\left(h_0, h_1, v\right) \sim \hat{\Pi}_{\text{supply}}} \left[h_1\bar{w}^{\left(h_0, h_1, v\right)}(\rho_1)
                \mathbf{1}_{\ubar{\rho}\left(h_0, h_1, v\right) \le \rho_1}\right] - \mathbb{E}_{\left(h_0, h_1, v\right) \sim \hat{\Pi}_{\text{supply}}} \left[h_1\bar{w}^{\left(h_0, h_1, v\right)}(\rho_2)
                \mathbf{1}_{\ubar{\rho}\left(h_0, h_1, v\right) \le \rho_2}\right]$ is equal to
                $\mathbb{E}_{\left(h_0, h_1, v\right) \sim \hat{\Pi}_{\text{supply}}} \left[
                \left(
                h_1\bar{w}^{\left(h_0, h_1, v\right)}(\rho_1)
                -
                h_1\bar{w}^{\left(h_0, h_1, v\right)}(\rho_2)
                \right)
                \mathbf{1}_{\ubar{\rho}\left(h_0, h_1, v\right) \le \rho_2}\right]$, which is the expectation of a positive random variable over the space $\bar{\Omega}_{\text{supply}}^{(2)}$, so it is positive.

                Thus, the function is positive and increasing.
            \end{proof}

            This suggests that the positivity and the monotonicity, including weak monotonicity according to the proof, in the individual level can lead to the same property in the aggregated level in the region where the supply level remains constant. However, it is important to note that continuity and many interesting properties cannot be transferred in the same manner without an additional conditions.

            \paragraph{Demand Revenue}

            \begin{condition}
            \label{condition:13}
                (Finite Demand Revenue)
        
                A competitive market $\circled{C}^{ \left(\mathcal{A}_{\text{supply}}, \mathcal{A}_{\text{demand}},
            \bar{\Omega}_{\text{supply}}, \bar{\Omega}_{\text{demand}}, w,\omega,\tilde{w},\tilde{\omega}\right)}
            \left(
                \hat{\Pi}_{\text{supply}}, \hat{\Pi}_{\text{demand}}
            \right)$ with conditions~\ref{condition:4},\ref{condition:5},\ref{condition:6},\ref{condition:8} has $\left(\mathcal{A}_{\text{demand}}, \bar{\Omega}_{\text{demand}}, \omega, \tilde{\omega}, \hat{\Pi}_{\text{demand}}\right)$ such that, for any $r \in \mathbb{R}$,
                \begin{align*}
                    \lim_{x \to \left(
                    \mathbb{E}_{\left(\eta_0, \eta_1\right) \sim \hat{\Pi}_{\text{demand}}}
                    \left[\ubar{v}\left(\eta_0, \eta_1\right)
                    \mathbf{1}_{\bar{r}\left(\eta_0, \eta_1\right) \ge r}
                    \right]
                    \right)^+}
                    \frac{
                    \mathbb{E}_{\left(\eta_0, \eta_1\right) \sim \hat{\Pi}_{\text{demand}}}
                    \left[\eta_1\bar{\omega}
                    ^{\left(\eta_0, \eta_1\right)}(r)
                    \mathbf{1}_{\bar{r}\left(\eta_0, \eta_1\right) \ge r}
                    \right]
                    }{x}
                    \in \mathbb{R}
                    \text{.}
                \end{align*}
            \end{condition}

            We can then define a ``demand revenue" function $\hat{p}_{\text{demand}}: \mathbb{R} \to \mathbb{R}$.
            
            \begin{definition}\label{definition:demandrevenue} (Demand Revenue)
                For a competitive market with the conditions~\ref{condition:4},\ref{condition:5},\ref{condition:6},\ref{condition:8},\ref{condition:13}, the demand revenue function $\hat{p}_{\text{demand}}: \mathbb{R} \to \mathbb{R}$ is defined such that, for any $r \in \mathbb{R}$,
                \begin{align*}
                    \hat{p}_{\text{demand}}(r) = 
                    \lim_{x \to \left(
                        \mathbb{E}_{\left(\eta_0, \eta_1\right) \sim \hat{\Pi}_{\text{demand}}}
                        \left[\ubar{v}\left(\eta_0, \eta_1\right)
                        \mathbf{1}_{\bar{r}\left(\eta_0, \eta_1\right) \ge r}
                        \right]
                        \right)^+}
                        \frac{
                        \mathbb{E}_{\left(\eta_0, \eta_1\right) \sim \hat{\Pi}_{\text{demand}}}
                        \left[\eta_1\bar{\omega}
                        ^{\left(\eta_0, \eta_1\right)}(r)
                        \mathbf{1}_{\bar{r}\left(\eta_0, \eta_1\right) \ge r}
                        \right]
                        }{x}
                    \text{.}
                \end{align*}
            \end{definition}

            Note that the mechanism does not allow the selection of quantity for the contracted made with the demanders who are indifferent between getting a contract and not getting a contract. This makes the demand revenue characterized by a function $\hat{p}_{\text{demand}}$, while the supply cost has to be characterized by $2$ functions.

            \begin{condition}
            \label{condition:14}
                (Demand Monotone Trend)
        
                A competitive market $\circled{C}^{ \left(\mathcal{A}_{\text{supply}}, \mathcal{A}_{\text{demand}},
            \bar{\Omega}_{\text{supply}}, \bar{\Omega}_{\text{demand}}, w,\omega,\tilde{w},\tilde{\omega}\right)}
            \left(
                \hat{\Pi}_{\text{supply}}, \hat{\Pi}_{\text{demand}}
            \right)$ with the conditions~\ref{condition:4},\ref{condition:5},\ref{condition:6},\ref{condition:8},\ref{condition:13} has $\left(\mathcal{A}_{\text{demand}}, \bar{\Omega}_{\text{demand}}, \omega, \tilde{\omega}, \hat{\Pi}_{\text{demand}}\right)$ such that, for any $r_1, r_2 \in \mathbb{R}$ with $r_1 > r_2$, if $\max\left(\tilde{D}(r_1)\right) = \max\left(\tilde{D}(r_2)\right) > 0$, then $\hat{p}_{\text{demand}}(r_1) > \hat{p}_{\text{demand}}(r_2)$.
            \end{condition}

            Similar to the sufficient condition for the condition~\ref{condition:11} (supply monotone trend), we can have a sufficient condition for the condition above.

            \begin{proposition}
            \label{proposition:randomfunction2}
                If a competitive market $\circled{C}^{ \left(\mathcal{A}_{\text{supply}}, \mathcal{A}_{\text{demand}},
                \bar{\Omega}_{\text{supply}}, \bar{\Omega}_{\text{demand}}, w,\omega,\tilde{w},\tilde{\omega}\right)}
                \left(
                    \hat{\Pi}_{\text{supply}}, \hat{\Pi}_{\text{demand}}
                \right)$ with conditions~\ref{condition:4},\ref{condition:5},\ref{condition:6},\ref{condition:8},\ref{condition:13} has $\left(\bar{\omega}^{\left(\eta_0, \eta_1\right)}\right)_{\left(\eta_0, \eta_1\right) \in \bar{\Omega}_{\text{demand}}}$ such that, for any $\left(\eta_0, \eta_1\right)\in \bar{\Omega}_{\text{demand}}$, the function $\bar{\omega}^{\left(\eta_0, \eta_1\right)}$ is increasing, then the condition~\ref{condition:14} is satisfied.
            \end{proposition}
            \begin{proof}
                This follows directly from the proof for the proposition~\ref{proposition:randomfunction}.
            \end{proof}

            \begin{condition}
            \label{condition:15}
                (Demand Left-Continuity)
        
                A competitive market $\circled{C}^{ \left(\mathcal{A}_{\text{supply}}, \mathcal{A}_{\text{demand}},
            \bar{\Omega}_{\text{supply}}, \bar{\Omega}_{\text{demand}}, w,\omega,\tilde{w},\tilde{\omega}\right)}
            \left(
                \hat{\Pi}_{\text{supply}}, \hat{\Pi}_{\text{demand}}
            \right)$ with the conditions~\ref{condition:4},\ref{condition:5},\ref{condition:6},\ref{condition:8},\ref{condition:13} has $\left(\mathcal{A}_{\text{demand}}, \bar{\Omega}_{\text{demand}}, \omega, \tilde{\omega}, \hat{\Pi}_{\text{demand}}\right)$ such that the function $\hat{p}_{\text{demand}}$ is left-continuous.
            \end{condition}

            This condition are important for the analysis. Without it, the definition of the competitive equilibrium will be altered and includes quantitative inequalities of the demand revenue function and a real demand.

            Next, we show that a simple condition upon an individual level function can lead to the satisfaction of the condition~\ref{condition:15}.
            
            \begin{proposition}
                If a competitive market $\circled{C}^{ \left(\mathcal{A}_{\text{supply}}, \mathcal{A}_{\text{demand}},
                \bar{\Omega}_{\text{supply}}, \bar{\Omega}_{\text{demand}}, w,\omega,\tilde{w},\tilde{\omega}\right)}
                \left(
                    \hat{\Pi}_{\text{supply}}, \hat{\Pi}_{\text{demand}}
                \right)$ with conditions~\ref{condition:4},\ref{condition:5},\ref{condition:6},\ref{condition:8},\ref{condition:13} has $\left(\bar{\omega}^{\left(\eta_0, \eta_1\right)}\right)_{\left(\eta_0, \eta_1\right) \in \bar{\Omega}_{\text{demand}}}$ such that, for any $\left(\eta_0, \eta_1\right)\in \bar{\Omega}_{\text{demand}}$, the function $\bar{\omega}^{\left(\eta_0, \eta_1\right)}$ is left-continuous and non-decreasing, then the condition~\ref{condition:15} is satisfied.
            \end{proposition}
            \begin{proof}
                Note that the denominator of the demand revenue in the definition~\ref{definition:demandrevenue} is $\mathbb{E}_{\left(\eta_0, \eta_1\right) \sim \hat{\Pi}_{\text{demand}}}
                    \left[\ubar{v}\left(\eta_0, \eta_1\right)
                    \mathbf{1}_{\bar{r}\left(\eta_0, \eta_1\right) \ge r}
                    \right]$, which is left-continuous, it then suffices to show that the nominator $\mathbb{E}_{\left(\eta_0, \eta_1\right) \sim \hat{\Pi}_{\text{demand}}}
                    \left[\eta_1\bar{\omega}
                    ^{\left(\eta_0, \eta_1\right)}(r)
                    \mathbf{1}_{\bar{r}\left(\eta_0, \eta_1\right) \ge r}
                    \right]$, which is equal to $\mathbb{E}_{\left(\eta_0, \eta_1\right) \sim \hat{\Pi}_{\text{demand}}}
                    \left[\eta_1 \max\left(\left\{\bar{\omega}^{\left(\eta_0, \eta_1\right)}(r),\lim_{r' \to \left(\bar{r}\left(\eta_0, \eta_1\right)\right)^-}\bar{\omega}^{\left(\eta_0, \eta_1\right)}(r')\right\}\right) 
                    \mathbf{1}_{\bar{r}\left(\eta_0, \eta_1\right) \ge r}
                    \right]$, is left-continuous. Thus, it suffices to show that the function $\mathbb{E}_{\left(\eta_0, \eta_1\right) \sim \hat{\Pi}_{\text{demand}}}
                    \left[\eta_1 \bar{\omega}^{\left(\eta_0, \eta_1\right)}(r)
                    \right]$ is left-continuous. Consider any $r \in \mathbb{R}$, and an arbitrary increasing sequence $\left(r_i\right)_{i=1}^\infty$ where $r_i \uparrow r$ as $i \to \infty$. From the individual non-decreasing premise, for any $\left(\eta_0, \eta_1\right)\in \bar{\Omega}_{\text{demand}}$, the sequence $\left(\eta_1 \bar{\omega}^{\left(\eta_0, \eta_1\right)}(r_i)\right)_{i=1}^\infty$ is a non-decreasing sequence. From monotone convergence theorem, we will then have that, the limit and expectation can be interchanged, yeilding $\lim_{i \to \infty}\mathbb{E}_{\left(\eta_0, \eta_1\right) \sim \hat{\Pi}_{\text{demand}}}
                    \left[\eta_1 \bar{\omega}^{\left(\eta_0, \eta_1\right)}(r_i)
                    \right] = \mathbb{E}_{\left(\eta_0, \eta_1\right) \sim \hat{\Pi}_{\text{demand}}}
                    \left[\eta_1 \lim_{i \to \infty}\bar{\omega}^{\left(\eta_0, \eta_1\right)}(r_i)
                    \right]$. From the individual left-continuous premise, we will have that $\mathbb{E}_{\left(\eta_0, \eta_1\right) \sim \hat{\Pi}_{\text{demand}}}
                    \left[\eta_1 \lim_{i \to \infty}\bar{\omega}^{\left(\eta_0, \eta_1\right)}(r_i)
                    \right] = \mathbb{E}_{\left(\eta_0, \eta_1\right) \sim \hat{\Pi}_{\text{demand}}}
                    \left[\eta_1 \bar{\omega}^{\left(\eta_0, \eta_1\right)}(r)
                    \right]$. Since this holds, for any increasing seqeunce, we will then have that $\lim_{r' \to r^-}\mathbb{E}_{\left(\eta_0, \eta_1\right) \sim \hat{\Pi}_{\text{demand}}}
                    \left[\eta_1 \lim_{i \to \infty}\bar{\omega}^{\left(\eta_0, \eta_1\right)}(r')
                    \right] = \mathbb{E}_{\left(\eta_0, \eta_1\right) \sim \hat{\Pi}_{\text{demand}}}
                    \left[\eta_1 \bar{\omega}^{\left(\eta_0, \eta_1\right)}(r)
                    \right]$, implying the left-continuity.\footnote{The author thanks Rajat Dwaraknath for the suggestion of the use of monotone convergence theorem.}
            \end{proof}

        From the 14 conditions we have posited, we can define a class of competitive market that will be considered for the ``competitive market equilibrium".

        \begin{definition}
            (Well-Behaved Competitive Market)
            For any 
            \begin{itemize}
                \item finite supplier admissible actions set $\mathcal{A}_{\text{supply}} \ne \emptyset$,
                
                \item finite demander admissible actions set $\mathcal{A}_{\text{demand}} \ne \emptyset$,
                
                \item supplier intensive characteristics space $\Omega_{\text{supply}} \ne \emptyset$,
    
                \item demander intensive characteristics space $\Omega_{\text{demand}}  \ne \emptyset$,
            
                \item measurable function $w:\Omega_{\text{supply}}
                \times \mathcal{A}_{\text{supply}}
                        \times \mathbb{R}^+_0
                        \times \mathbb{R}
                        \times [0,1)
                        \to [-\infty, \infty)$,

                \item measurable function $\omega:\Omega_{\text{demand}}
                \times \mathcal{A}_{\text{demand}}
                        \times \mathbb{R}^+_0
                        \times \mathbb{R}
                        \times [0,1)
                        \to [-\infty, \infty)$,
    
                \item measurable function $\tilde{w}:\Omega_{\text{supply}}              \times \mathcal{A}_{\text{supply}}
                        \times \mathbb{R}^+_0
                        \times \mathbb{R}
                        \times [0,1)
                        \to \mathbb{R}$,
    
                \item measurable function $\tilde{\omega}:\Omega_{\text{demand}}         \times \mathcal{A}_{\text{demand}}
                        \times \mathbb{R}^+_0
                        \times \mathbb{R}
                        \times [0,1)
                        \to \mathbb{R}$,     
    
                \item individual supplier joint probability distribution $\hat{\Pi}_{\text{supply}} \in \Delta\left(\Omega_{\text{supply}} \times \mathbb{R}^+ \times \mathbb{R}^+_0\right)$,
    
                \item individual demander joint probability distribution $\hat{\Pi}_{\text{demand}} \in \Delta\left(\Omega_{\text{demand}} \times \mathbb{R}^+\right)$,   
            \end{itemize}
            a competitive market 
            \begin{align*}
                \circled{C}^{ \left(\mathcal{A}_{\text{supply}}, \mathcal{A}_{\text{demand}},
            \bar{\Omega}_{\text{supply}}, \bar{\Omega}_{\text{demand}}, w,\omega,\tilde{w},\tilde{\omega}\right)}
            \left(
                \hat{\Pi}_{\text{supply}}, \hat{\Pi}_{\text{demand}}
            \right)
            \end{align*}
            is a well-behaved competitive market if the conditions~\ref{condition:1},~\ref{condition:2},~\ref{condition:3},~\ref{condition:4},~\ref{condition:5},~\ref{condition:6},~\ref{condition:7},~\ref{condition:8},~\ref{condition:9},~\ref{condition:10},~\ref{condition:11},~\ref{condition:13},~\ref{condition:14}, and~\ref{condition:15} are satisfied.
        \end{definition}
        
    \subsection{Competitive Market Equilibrium Definition}
        
        Since every mediator has an option to not participate in the market, the equilibrium will require every mediator to achieve a non-negative utility. This suggests that the any winning mediator will place a demand price $r \in \mathbb{R}$ such that the demand revenue scaled by the probability of getting contracted is higher or equal to the supply cost, which is positive from the condition~\ref{condition:10} (no free supply). Therefore, if a resale opportunity is provided and if the highest demand price $r \in \mathbb{R}$ yields the corresponding demand revenue higher than the supply price, which is globally shared, and has some matched scarcity and some unmatched scarcity quoted at the price $r$, then the subsequent resales bidding in pseudo-equilibrium will occur with non-zero volume and a promised demand price $r' < r$, because of the condition~\ref{condition:15} (demand left-continuity). Note that this result is similar to the example of commonly known value auction analyzed in the subsubsection~\ref{subsubsection:commonlyknownvalue}. After the first resale, the demand matched at the original demand price $r$ will be reduced, since a non-zero mass of demanders will recontract with the newly offered price $r'$. Therefore, a (possibly infinite) sequence of resales will lead to the matched demand at $r$ converging to some non-negative value. If such value is $0$, then the utility of the original winning mediator who quotes price $r$ non-positive. Otherwise, the profit per unit will be positive. From the maximal monotonicity of the negative of the real demand correspondence~\ref{corollary:demandmaximalmonotone}, we have that, \textbf{if} there exists a deviation where the matched demand at the highest matched price $r$ converging to some positive value, \textbf{then}, by having every winning mediator quoting the price $r$ and having every resale bid price to be the same, the profit per unit of every mediator will be non-zero. 
        
        Note that a monopoly can replicate this bidding strategy, since the demand-side bid only consists of $1$ value. Thus, we will have that a monopoly will get a positive utility, which is better than getting any value of differential utility that one will otherwise get in the market equilibrium, which will be when every mediator has only secured differential volume of scarcity. 
        
        This suggests that the analysis can be done with the abstraction of the resale bidding strategy, which can be an infinite-time strategy, into only the limit matched quantity at the single originally quoted price $r$.

        \begin{definition}
            (Quasi-Well-Ordered)
            A set $R \subseteq \mathbb{R}$ is a quasi-well-ordered set if there exists a finite collection of non-decreasing real sequence $\left\{\left(r_i^{(j)}\right)_{i=1}^{\infty}\right\}_{j=1}^{m}$ for some $m \in \mathbb{N}$ such that
            \begin{align*}
                R = \bigcup_{j=1}^m \left\{r^{(j)}_i\right\}_{i=1}^{\infty}
                \text{.}
            \end{align*}
        \end{definition}

        \begin{corollary}
            For any quasi-well-ordered set $R$, for any $R' \subseteq R$, if $R' \ne \emptyset$, then $R'$ is quasi-well-ordered.
        \end{corollary}
        \begin{proof}
            Since $R \in W$, there exists a finite collection of non-decreasing real sequence $\left\{\left(r_i^{(j)}\right)_{i=1}^{\infty}\right\}_{j=1}^{m}$ for some $m \in \mathbb{N}$ such that
            $R = \bigcup_{j=1}^m \left\{r^{(j)}_i\right\}_{i=1}^{\infty}$. Let $J = \left\{j\in \mathbb{N} \cap (-\infty,m]:\left\{r^{(j)}_i\right\}_{i=1}^{\infty} \cap R' \ne \emptyset\right\}$, so $\vert J\vert \in \mathbb{N}$. For each $j \in J$, we can find a non-empty subsequence of $\left(r^{(j)}_i\right)$ whose every element is in $R'$. If such subsequence is infinite, we denote its as $\tilde{r}^{(j)}$. Otherwise, we can append the biggest element into it infinitely, and denote the result sequence as $\tilde{r}^{(j)}$. Therefore, $R' = \bigcup_{j \in J}\left\{\tilde{r}^{(j)}_i\right\}_{i=1}^{\infty}$.
        \end{proof}

        \begin{definition}
            (Quasi-Well-Ordered Collection)
            For any $A \subseteq \mathbb{R}$ with $A \ne \emptyset$, the collection $W^{(A)} \in \mathcal{B}\left(\mathbb{R}\right)$ is defined to be a collection containing every quasi-well-ordered subset of $A$.
        \end{definition}

        We can then define a special case of equilibrium called ``competitive market equilibrium" by restricting the attention to the equilibrium with the set of winning mediators (from the supply market) is $[0,1]$.
        
        \begin{definition}
        \label{definition:competitiveequilibrium}
            (Competitive Market Equilibrium)

            A tuple
            \begin{align*}
                \left(
                    \mu_{\text{supply}}, 
                    \mu_{\text{demand}}
                \right)
            \end{align*}
            is a competitive market equilibrium of a well-behaved competitive market
            \begin{align*}
                \circled{C}^{ \left(\mathcal{A}_{\text{supply}}, \mathcal{A}_{\text{demand}},
                \bar{\Omega}_{\text{supply}}, \bar{\Omega}_{\text{demand}}, w,\omega,\tilde{w},\tilde{\omega}\right)}
                \left(
                    \hat{\Pi}_{\text{supply}}, \hat{\Pi}_{\text{demand}}
                \right)
            \end{align*}
            if, by defining
            \begin{enumerate}

                \item [a.] The admissible joint mediator supply bid strategies collection $\mathcal{B}_{\text{supply}}$ to be a collection containing every measurable function $f: [0,1) \to [0,1] \times \left(\mathbb{R} \cup \{-\infty\}\right)$ such that 
                \begin{itemize}
                    \item $\argmax_{z \in [0,1)} \left(f(z)\right)_2$ is well-defined;
                    \item $\argmax_{z \in \argmax_{z' \in [0,1)} \left(f(z')\right)_2} \left(f(z)\right)_1 \in \mathcal{B}\left([0,1)\right)$, and it either has finitely many elements or has a non-zero Lebesgue measure;
                \end{itemize}

                \item [b.] For any $G \in \mathcal{G}$, the conditional admissible joint mediator demand bid strategies collection $\mathcal{B}_{\text{demand}}^{(G)}$ to be a collection containing every measurable function $f:
                G \to \mathbb{R}$ such that $\text{Range}\left(f\right) \in W^{\left(\mathbb{R}\right)}$,
            \end{enumerate}
            
            there exists a tuple of functions 
            \begin{align*}
                \left(B_{\text{supply}}, B_{\text{demand}}\right)
                \in \mathcal{B}_{\text{supply}} \times \mathcal{B}_{\text{demand}}^{([0,1))}
            \end{align*}
            such that, by defining
            \begin{enumerate}

                \item [0.)] The conditional supply cost function $\hat{c}: [0,1] \times \mathbb{R} \to \mathbb{R}^+_0$ such that, for any $q \in [0,1]$, $\rho \in \mathbb{R}$,
                \begin{align*}
                    \hat{c}(q,\rho)
                    =
                    \lim_{x \to \left((1-q)\min\left(\tilde{S}(\rho)\right)+q\max\left(\tilde{S}(\rho)\right)\right)^+} \frac{(1-q)\min\left(\tilde{S}(\rho)\right)\ubar{p}_{\text{supply}}(\rho) +q\max\left(\tilde{S}(\rho)\right)\bar{p}_{\text{supply}}(\rho)}{x}
                    \text{.}
                \end{align*}

                \item [1.)] The wining supply price function $\bar{\rho}: \mathcal{B}_{\text{supply}} \to \mathbb{R} \cup \{-\infty\}$ such that, for any $\tilde{B}_{\text{supply}} \in \mathcal{B}_{\text{supply}}$,
                \begin{align*}
                    \bar{\rho}\left(\tilde{B}_{\text{supply}}\right) = \max_{z \in [0,1)} \left(\tilde{B}_{\text{supply}}(z)\right)_2
                    \text{;}
                \end{align*}

                \item [2.)] The wining residual supply proportion function $\bar{q}: \mathcal{B}_{\text{supply}} \to [0,1]$ such that, for any $\tilde{B}_{\text{supply}} \in \mathcal{B}_{\text{supply}}$,
                \begin{align*}
                    \bar{q}\left(\tilde{B}_{\text{supply}}\right) = \max_{z \in \argmax_{z' \in [0,1)} \left({\tilde{B}}_{\text{supply}}(z')\right)_2} \left({\tilde{B}}_{\text{supply}}(z)\right)_1
                    \text{;}
                \end{align*}

                \item [3.)] The wining supply mediators set function $\bar{M}: \mathcal{B}_{\text{supply}} \to \mathcal{G}$ such that, for any $\tilde{B}_{\text{supply}} \in \mathcal{B}_{\text{supply}}$,
                \begin{align*}
                    \bar{M}\left(\tilde{B}_{\text{supply}}\right) =
                    \argmax_{z \in \argmax_{z' \in [0,1)} \left(\tilde{B}_{\text{supply}}(z')\right)_2} \left(\tilde{B}_{\text{supply}}(z)\right)_1
                \text{;}
                \end{align*}

                \item [4.)] The acquired supply volume function $Q: \mathcal{B}_{\text{supply}} \to \mathbb{R}^+_0$ such that, for any $\tilde{B}_{\text{supply}} \in \mathcal{B}_{\text{supply}}$,
                \begin{align*}
                    &Q\left(\tilde{B}_{\text{supply}}\right) \\
                    &=
                    \begin{cases}
                        \left(1- \bar{q}\left(\tilde{B}_{\text{supply}}\right) \right)\min\left(\tilde{S}\left(\bar{\rho}\left(\tilde{B}_{\text{supply}}\right)\right)\right) + \bar{q}\left(\tilde{B}_{\text{supply}}\right) \max\left(\tilde{S}\left(\bar{\rho}\left(\tilde{B}_{\text{supply}}\right)\right)\right)
                        &\text{ if } \bar{\rho}\left(\tilde{B}_{\text{supply}}\right) \ne -\infty
                        \\
                        0 &\text{ if } \bar{\rho}\left(\tilde{B}_{\text{supply}}\right) = -\infty
                    \end{cases}
                    \text{;}
                \end{align*}

                \item [5.)] The unit cost function $c: \mathcal{B}_{\text{supply}} \to \mathbb{R}^+_0$ such that, for any $\tilde{B}_{\text{supply}} \in \mathcal{B}_{\text{supply}}$,
                \begin{align*}
                    c\left(\tilde{B}_{\text{supply}}\right)
                    =
                    \begin{cases}
                        \hat{c}\left(\bar{q}\left(\tilde{B}_{\text{supply}}\right), \bar{\rho}\left(\tilde{B}_{\text{supply}}\right)\right)
                        &\text{ if } \bar{\rho}\left(\tilde{B}_{\text{supply}}\right) \ne -\infty
                        \\
                        0 &\text{ if } \bar{\rho}\left(\tilde{B}_{\text{supply}}\right) = -\infty
                    \end{cases}
                    \text{;}
                \end{align*}

                \item [6.)] The divisor function $d: \mathcal{B}_{\text{supply}} \to \mathbb{R}^+$ such that, for any $\tilde{B}_{\text{supply}} \in \mathcal{B}_{\text{supply}}$,
                \begin{align*}
                    d\left(\tilde{B}_{\text{supply}}\right)
                    =
                    \begin{cases}
                        \left\vert \bar{M}\left(\tilde{B}_{\text{supply}}\right) \right\vert
                        &\text{ if } \left\vert \bar{M}\left(\tilde{B}_{\text{supply}}\right) \right\vert < \infty
                        \\
                        \lambda\left( \bar{M}\left(\tilde{B}_{\text{supply}}\right) \right)
                        &\text{ if } \left\vert \bar{M}\left(\tilde{B}_{\text{supply}}\right) \right\vert = \infty
                    \end{cases}
                    \text{;}
                \end{align*}

                \item [7.)] For any $\tilde{B}_{\text{supply}} \in \mathcal{B}_{\text{supply}}$, for any $\tilde{B}_{\text{demand}} \in \mathcal{B}_{\text{demand}}^{\left(\bar{M}\left(\tilde{B}_{\text{supply}}\right)\right)}$, 
                \begin{itemize}
                    \item 
                    The demand bid measure $\hat{\mu}_{\text{demand}}^{\left(\tilde{B}_{\text{supply}}, \tilde{B}_{\text{demand}}\right)}$ over $\mathcal{B}(\mathbb{R})$,
                    such that, for any $A \in \mathcal{B}(\mathbb{R})$,
                    \begin{align*}
                        \hat{\mu}_{\text{demand}}^{\left(\tilde{B}_{\text{supply}}, \tilde{B}_{\text{demand}}\right)}(A)
                        =
                        Q\left(\tilde{B}_{\text{supply}}\right)
                        \int_{z \in \bar{M}^{\left(\tilde{B}_{\text{supply}}\right)}}
                        \left(\tilde{B}_{\text{demand}}(z)\right)_1
                        \mathbf{1}_{ \left(\tilde{B}_{\text{demand}}(z)\right)_2 \in A}
                        I_{\bar{M}^{\left(\tilde{B}_{\text{supply}}\right)}}\left(dz\right)
                        \text{;}
                    \end{align*}

                    \item The number of enumerable sequences $N^{\left(\tilde{B}_{\text{supply}}, \tilde{B}_{\text{demand}}\right)}$ to be the smallest $n \in \mathbb{N}$ such that there exists a $n$-collection of non-decreasing infinite sequences $\left\{\left(r_i^{(j)}\right)_{i=1}^{\infty}\right\}_{j=1}^{n}$ such that $\text{support}\left(\hat{\mu}_{\text{demand}}^{\left(\tilde{B}_{\text{supply}}, \tilde{B}_{\text{demand}}\right)}\right) 
                        = \bigcup_{j=1}^n \left\{r^{(j)}_i\right\}_{i=1}^{\infty}$;

                    \item The supremum enumerable price function $\hat{r}^{\left(\tilde{B}_{\text{supply}}, \tilde{B}_{\text{demand}}\right)}: \mathbb{N}_{0} \cap \left(-\infty, N^{\left(\tilde{B}_{\text{supply}}, \tilde{B}_{\text{demand}}\right)}\right] \to \mathbb{R} \cup \{-\infty\}$ such that $\hat{r}^{\left(\tilde{B}_{\text{supply}}, \tilde{B}_{\text{demand}}\right)}(0) = -\infty$;
                    
                    \item For any $i \in \mathbb{N}_0\cap \left(-\infty, N^{\left(\tilde{B}_{\text{supply}}, \tilde{B}_{\text{demand}}\right)}\right)$, and, for any $i \in \mathbb{N}_0\cap \left(-\infty, N^{\left(\tilde{B}_{\text{supply}}, \tilde{B}_{\text{demand}}\right)}\right)$,
                    \begin{align*}
                        &\hat{r}^{\left(\tilde{B}_{\text{supply}}, \tilde{B}_{\text{demand}}\right)}(i+1)
                        =\\
                        &
                        \sup\left(\left\{x \in \left(\hat{r}^{\left(\tilde{B}_{\text{supply}}, \tilde{B}_{\text{demand}}\right)}(i),\infty\right) : \left\vert \text{support}\left(\hat{\mu}_{\text{demand}}^{\left(\tilde{B}_{\text{supply}}, \tilde{B}_{\text{demand}}\right)}\right) \cap \left(\hat{r}^{\left(\tilde{B}_{\text{supply}}, \tilde{B}_{\text{demand}}\right)}(i), x\right] \right\vert \in \mathbb{N}\right\}\right)
                        \text{;}
                    \end{align*}

                    \item For any $i \in \mathbb{N} \cap \left(-\infty, N^{\left(\tilde{B}_{\text{supply}}, \tilde{B}_{\text{demand}}\right)}\right]$, the demand price iteration function $R_i^{\left(\tilde{B}_{\text{supply}}, \tilde{B}_{\text{demand}}\right)}: \mathbb{N} \cap \left(-\infty,\left\vert \text{support}\left(\hat{\mu}_{\text{demand}}^{\left(\tilde{B}_{\text{supply}}, \tilde{B}_{\text{demand}}\right)}\right)
                        \cap \left[\hat{r}^{\left(\tilde{B}_{\text{supply}}, \tilde{B}_{\text{demand}}\right)}(i-1), \hat{r}^{\left(\tilde{B}_{\text{supply}}, \tilde{B}_{\text{demand}}\right)}(i)\right)
                    \right\vert\right]$ $\to \text{support}\left(\hat{\mu}_{\text{demand}}^{\left(\tilde{B}_{\text{supply}}, \tilde{B}_{\text{demand}}\right)}\right)
                        $ $\cap \left[\hat{r}^{\left(\tilde{B}_{\text{supply}}, \tilde{B}_{\text{demand}}\right)}(i-1), \hat{r}^{\left(\tilde{B}_{\text{supply}}, \tilde{B}_{\text{demand}}\right)}(i)\right)$ that is bijective and increasing;

                    \item For any $i \in \mathbb{N} \cap \left(-\infty, N^{\left(\tilde{B}_{\text{supply}}, \tilde{B}_{\text{demand}}\right)}\right]$, the clearing probability function $P_i^{\left(\tilde{B}_{\text{supply}}, \tilde{B}_{\text{demand}}\right)}:\\ \text{Domain}\left(R_i^{\left(\tilde{B}_{\text{supply}}, \tilde{B}_{\text{demand}}\right)}\right)$ $ \to [0,1]$ such that, for any $j \in \text{Domain}\left(R_i^{\left(\tilde{B}_{\text{supply}}, \tilde{B}_{\text{demand}}\right)}\right)$,
                    \begin{align*}
                        &P_i^{\left(\tilde{B}_{\text{supply}}, \tilde{B}_{\text{demand}}\right)}
                        (j)
                        =\\
                        &
                        \min\left(\left\{
                        \left(\bar{P}_i^{\left(\tilde{B}_{\text{supply}}, \tilde{B}_{\text{demand}}\right)} - \sum_{k=1}^{j-1} P_i^{\left(\tilde{B}_{\text{supply}}, \tilde{B}_{\text{demand}}\right)}(k)
                        \right),
                        \lim_{
                        x \to 
                        \left(
                        \frac{
                        \max\left(\tilde{D}\left(R_i^{\left(\tilde{B}_{\text{supply}}, \tilde{B}_{\text{demand}}\right)}(j)\right)\right)
                        }
                        {
                        \hat{\mu}_{\text{demand}}^{\left(\hat{\tilde{B}}_{\text{supply}}, \hat{\tilde{B}}_{\text{demand}}\right)}
                        \left(\left\{R_i^{\left(\tilde{B}_{\text{supply}}, \tilde{B}_{\text{demand}}\right)}(j)\right\}\right)
                        }
                        \right)^+
                        }
                        x^{-1}
                        \right\}\right)
                        \text{,}
                    \end{align*}
                    and the sequence $\left(\bar{P}_i^{\left(\tilde{B}_{\text{supply}}, \tilde{B}_{\text{demand}}\right)}\right)_{i=1}^{N^{\left(\tilde{B}_{\text{supply}}, \tilde{B}_{\text{demand}}\right)}}$ is defined such that, for any $i \in \mathbb{N} \cap \left(-\infty, N^{\left(\tilde{B}_{\text{supply}}, \tilde{B}_{\text{demand}}\right)}\right]$,
                    \begin{align*}
                        \bar{P}_i^{\left(\tilde{B}_{\text{supply}}, \tilde{B}_{\text{demand}}\right)} = 1 - \sum_{j=1}^{i-1} 
                        \left(
                        \sum_{k=1}^{\left\vert \text{Domain}\left(R_j^{\left(\tilde{B}_{\text{supply}}, \tilde{B}_{\text{demand}}\right)}\right)
                        \right\vert} P_j^{\left(\tilde{B}_{\text{supply}}, \tilde{B}_{\text{demand}}\right)}(k)
                        \right)
                        \text{;}
                    \end{align*}

                    \item The possible highest price
                    \begin{align*}
                        &\bar{r}^{\left(\tilde{B}_{\text{supply}}, \tilde{B}_{\text{demand}}\right)}
                        =\\
                        &\min\left(
                        \left(
                        \bigcup_{i=1}^{N^{\left(\tilde{B}_{\text{supply}}, \tilde{B}_{\text{demand}}\right)}} \left(A_i^{\left(\tilde{B}_{\text{supply}}, \tilde{B}_{\text{demand}}\right)}
                        \cup B_i^{\left(\tilde{B}_{\text{supply}}, \tilde{B}_{\text{demand}}\right)}\right)
                        \right)
                         \cup \left\{ \sup\left(\left\{r \in \mathbb{R}: \max\left(\tilde{D}(r)\right) > 0\right\}\right) \right\}\right)
                        \text{,}
                    \end{align*}
                    where, for any $i \in \mathbb{N} \cap \left(-\infty, N^{\left(\tilde{B}_{\text{supply}}, \tilde{B}_{\text{demand}}\right)}\right]$,
                    \begin{align*}
                        A_i^{\left(\tilde{B}_{\text{supply}}, \tilde{B}_{\text{demand}}\right)}
                        =
                        \left(
                        \left\{R_i^{\left(\tilde{B}_{\text{supply}}, \tilde{B}_{\text{demand}}\right)}(j)\right\}_{
                        \substack{j \in \text{Domain}\left(R_i^{\left(\tilde{B}_{\text{supply}}, \tilde{B}_{\text{demand}}\right)}\right) :\\ 
                        \sum_{k=1}^{j} P_i^{\left(\tilde{B}_{\text{supply}}, \tilde{B}_{\text{demand}}\right)}(k) = \bar{P}_i^{\left(\tilde{B}_{\text{supply}}, \tilde{B}_{\text{demand}}\right)}}
                         }
                         \right)
                        \text{,}
                    \end{align*}
                    and
                    \begin{align*}
                        &B_i^{\left(\tilde{B}_{\text{supply}}, \tilde{B}_{\text{demand}}\right)}
                        =\\
                        &
                        \begin{cases}
                            \emptyset
                            &\text{ if } \bar{P}_i^{\left(\tilde{B}_{\text{supply}}, \tilde{B}_{\text{demand}}\right)} \ne
                            \sum_{j \in \text{Domain}\left(R_i^{\left(\tilde{B}_{\text{supply}}, \tilde{B}_{\text{demand}}\right)}\right)} P_i^{\left(\tilde{B}_{\text{supply}}, \tilde{B}_{\text{demand}}\right)}(j)\\
                            \left\{\sup\left(\text{Range}\left(R_i^{\left(\tilde{B}_{\text{supply}}, \tilde{B}_{\text{demand}}\right)}\right)\right)\right\}
                            &\text{ if } \bar{P}_i^{\left(\tilde{B}_{\text{supply}}, \tilde{B}_{\text{demand}}\right)} =
                            \sum_{j \in \text{Domain}\left(R_i^{\left(\tilde{B}_{\text{supply}}, \tilde{B}_{\text{demand}}\right)}\right)} P_i^{\left(\tilde{B}_{\text{supply}}, \tilde{B}_{\text{demand}}\right)}(j)
                        \end{cases}
                        \text{;}
                    \end{align*}

                    \item The probability of the possible highest price getting matched
                    \begin{align*}
                        &p_{\text{high}}^{\left(\tilde{B}_{\text{supply}}, \tilde{B}_{\text{demand}}\right)}
                        =\\
                        &
                        \min \left(\left\{1,
                                \left(\lim_{x \to
                                \left(
                                \hat{\mu}_{\text{demand}}^{\left(\hat{\tilde{B}}_{\text{supply}}, \hat{\tilde{B}}_{\text{demand}}\right)}
                                \left(\left\{ \bar{r}^{\left(\tilde{B}_{\text{supply}}, \tilde{B}_{\text{demand}}\right)} \right\}\right)
                                \right)^+}
                                \frac{
                                    \max\left(\tilde{D}\left(\bar{r}^{\left(\tilde{B}_{\text{supply}}, \tilde{B}_{\text{demand}}\right)}\right)\right)
                                    \bar{p}^{\left(\tilde{B}_{\text{supply}}, \tilde{B}_{\text{demand}}\right)}
                                }
                                {
                                    x
                                }
                                \right)
                            \right\}\right)
                            \text{,}
                    \end{align*}
                    where
                    \begin{align*}
                        \bar{p}^{\left(\tilde{B}_{\text{supply}}, \tilde{B}_{\text{demand}}\right)}
                        =
                        1 -
                        \sum_{i=1}^{N^{\left(\tilde{B}_{\text{supply}}, \tilde{B}_{\text{demand}}\right)}}
                        \left(
                        \sum_{
                        \substack{
                        j \in \text{Domain}\left(R_i^{\left(\tilde{B}_{\text{supply}}, \tilde{B}_{\text{demand}}\right)}\right)\\
                        R_i^{\left(\tilde{B}_{\text{supply}}, \tilde{B}_{\text{demand}}\right)}(j) < \bar{r}^{\left(\tilde{B}_{\text{supply}}, \tilde{B}_{\text{demand}}\right)}
                        }}
                        P_i^{\left(\tilde{B}_{\text{supply}}, \tilde{B}_{\text{demand}}\right)}
                        (j)
                        \right)
                        \text{;}
                    \end{align*}

                    \item The resalable volume
                    \begin{align*}
                        &Q_{\text{resale}}^{\left(\tilde{B}_{\text{supply}}, \tilde{B}_{\text{demand}}\right)}
                        =
                        \\
                        &
                        \left(1 - p_{\text{high}}^{\left(\tilde{B}_{\text{supply}}, \tilde{B}_{\text{demand}}\right)}\right)
                        \hat{\mu}_{\text{demand}}^{\left(\tilde{B}_{\text{supply}}, \tilde{B}_{\text{demand}}\right)}\left(
                            \left\{\bar{r}^{\left(\tilde{B}_{\text{supply}}, \tilde{B}_{\text{demand}}\right)}\right\}
                        \right) +\hat{\mu}_{\text{demand}}^{\left(\tilde{B}_{\text{supply}}, \tilde{B}_{\text{demand}}\right)}\left(
                            \left( \bar{r}^{\left(\tilde{B}_{\text{supply}}, \tilde{B}_{\text{demand}}\right)}, \infty\right)
                        \right)
                        \text{;}
                    \end{align*}

                    \item The weighted revenue function $e_1^{\left(\tilde{B}_{\text{supply}}, \tilde{B}_{\text{demand}}\right)}: \mathbb{R} \to \mathbb{R}$ such that, for any $r \in \mathbb{R}$, 
                    \begin{align*}
                        e_1^{\left(\tilde{B}_{\text{supply}}, \tilde{B}_{\text{demand}}\right)}(r)
                        =
                        \begin{cases}
                            \hat{p}_{\text{demand}}(r)
                            &\text{ if } 
                            r < \bar{r}^{\left(\tilde{B}_{\text{supply}}, \tilde{B}_{\text{demand}}\right)}
                            \\
                            0
                            &\text{ if } 
                            r > \bar{r}^{\left(\tilde{B}_{\text{supply}}, \tilde{B}_{\text{demand}}\right)}
                            \\
                            p_{\text{high}}^{\left(\tilde{B}_{\text{supply}}, \tilde{B}_{\text{demand}}\right)} 
                            \hat{p}_{\text{demand}}(r)
                            &\text{ if } 
                            r = \bar{r}^{\left(\tilde{B}_{\text{supply}}, \tilde{B}_{\text{demand}}\right)}
                        \end{cases}
                        \text{;}
                    \end{align*}

                    \item The weighted revenue for resale function $e_2^{\left(\tilde{B}_{\text{supply}}, \tilde{B}_{\text{demand}}\right)}: \mathbb{R} \to \mathbb{R}$ such that, for any $r \in \mathbb{R}$, 
                    \begin{align*}
                        e_2^{\left(\tilde{B}_{\text{supply}}, \tilde{B}_{\text{demand}}\right)}(r)
                        =
                        \hat{p}_{\text{demand}}(r)\mathbf{1}_{r<\bar{r}^{\left(\tilde{B}_{\text{supply}}, \tilde{B}_{\text{demand}}\right)}}
                        \text{;}
                    \end{align*}

                    \item The profitable resale region
                    \begin{align*}
                        R^{\left(\tilde{B}_{\text{supply}}, \tilde{B}_{\text{demand}}\right)}_{\text{resale}}
                        =
                        \left\{r\in\mathbb{R}: 
                        \mathbf{1}_{\lambda\left(\bar{M}\left(\tilde{B}_{\text{supply}}\right)\right) = 0} 
                        \mathbf{1}_{Q_{\text{resale}}^{\left(\tilde{B}_{\text{supply}}, \tilde{B}_{\text{demand}}\right)} > 0}
                        e_2^{\left(\tilde{B}_{\text{supply}}, \tilde{B}_{\text{demand}}
                        \right)}(r) > c\left(\tilde{B}_{\text{supply}}\right)\right\}
                        \text{;}
                    \end{align*}

                    \item The conditional value function without resale $U^{\left(\tilde{B}_{\text{supply}}, \tilde{B}_{\text{demand}}\right)}: [0,1) \to \mathbb{R}^2$ such that, for any $m \in [0,1)$,
                    \begin{align*}
                        &U^{\left(\tilde{B}_{\text{supply}}, \tilde{B}_{\text{demand}}\right)}(m)
                        =\\
                        &
                        \begin{cases}
                            \frac{Q\left(\tilde{B}_{\text{supply}}\right)
                                \left[e_1^{\left(\tilde{B}_{\text{supply}}, \tilde{B}_{\text{demand}}\right)}
                                \left(\tilde{B}_{\text{demand}}(m)\right)
                                -c\left(\tilde{B}_{\text{supply}}\right)
                                \right]
                                }{d\left(\tilde{B}_{\text{supply}}\right)}
                                \left(
                                \mathbf{1}_{\left\vert \bar{M}\left(\tilde{B}_{\text{supply}}\right)\right\vert = \infty}, 1
                            \right)
                            &\text{ if } m \in \bar{M}\left(\tilde{B}_{\text{supply}}\right)
                            \\
                            (0,0) &\text{ if } m \not\in \bar{M}\left(\tilde{B}_{\text{supply}}\right)
                        \end{cases}
                        \text{;}
                    \end{align*}

                    \item For any $G \in \mathcal{G}$ and $G \subseteq \bar{M}\left(\tilde{B}_{\text{supply}}\right)$, the collection of deviated demand bid functions $\mathcal{B}_{\text{demand}}^{\left(\tilde{B}_{\text{supply}}, \tilde{B}_{\text{demand}}, G\right)}$ to contain every $f \in \mathcal{B}_{\text{demand}}^{\left(\bar{M}\left(\tilde{B}_{\text{supply}}\right)\right)}$ such that, for any $m \not\in G$,
                    \begin{align*}
                        f(m) = \tilde{B}_{\text{demand}}(m)
                        \text{,}
                    \end{align*}
                \end{itemize}      

                \item [8.)] For any $\tilde{B}_{\text{supply}} \in \mathcal{B}_{\text{supply}}$, 
                \begin{itemize}

                    \item A collection of non-resale demand bid function
                    \begin{align*}
                        \mathcal{C}_{\text{demand}}^{\left(\tilde{B}_{\text{supply}}\right)}
                        =
                            \left\{f\in \mathcal{B}_{\text{demand}}^{\left(\bar{M}\left(\tilde{B}_{\text{supply}}\right)\right)} :
                            R^{\left(\tilde{B}_{\text{supply}}, f\right)}_{\text{resale}} = \emptyset
                            \right\}
                        \text{;}
                    \end{align*}

                    \item The collection of possible subsequent equilibria with non-negative outcome $\mathcal{D}_{\text{demand}}^{\left(\tilde{B}_{\text{supply}}\right)}$ to contain every function $f \in \mathcal{C}_{\text{demand}}^{\left(\tilde{B}_{\text{supply}}\right)}$ such that
                    \begin{itemize}
                        \item (Non-Negative Outcome)
                        \begin{align*}
                            \left\{
                        U^{\left(\tilde{B}_{\text{supply}}, f\right)}(m)\right\}_{m \in [0,1)} \subseteq \left(\mathbb{R}^+_0 \right)^2
                        \text{;}
                        \end{align*}

                        \item (No Betrayal-Free Deviation) For any $G \in \mathcal{G}$ with $G \subseteq \bar{M}\left(\tilde{B}_{\text{supply}}\right)$, for any  $f_1 \in \mathcal{B}_{\text{demand}}^{\left(\tilde{B}_{\text{supply}}, f, G\right)} \cap \mathcal{C}_{\text{demand}}^{\left(\tilde{B}_{\text{supply}}\right)}$, if
                        \begin{align*}
                            I_G\left(\left\{
                                m \in G:
                                U^{\left(\tilde{B}_{\text{supply}}, f_1\right)}(m)
                                \succ_{2}
                                U^{\left(\tilde{B}_{\text{supply}},f\right)}(m)
                            \right\}\right)
                            = 1
                            \text{,}
                        \end{align*}
                        then there exists some $G' \subseteq G$ with $G', G-G' \in \mathcal{G}$ such that there exists some $f_2 \in \mathcal{B}_{\text{demand}}^{\left(\tilde{B}_{\text{supply}}, f_1, G'\right)} \cap \mathcal{C}_{\text{demand}}^{\left(\tilde{B}_{\text{supply}}\right)}$ such that
                        \begin{align*}
                            I_{G'}\left(\left\{
                                m \in G':
                                U^{\left(\tilde{B}_{\text{supply}}, f_2\right)}(m)
                                \succ_{2}
                                U^{\left(\tilde{B}_{\text{supply}}, f_1\right)}(m)
                            \right\}\right)
                            = 1
                            \text{,}
                        \end{align*}
                        and
                        \begin{align*}
                            I_{(G-G')}\left(\left\{
                                m \in G-G':
                                U^{\left(\tilde{B}_{\text{supply}}, f\right)}(m)
                                \succcurlyeq_{2}
                                U^{\left(\tilde{B}_{\text{supply}}, f_2\right)}(m)
                            \right\}\right)
                            > 0
                            \text{;}
                        \end{align*}

                    \end{itemize}
                
                    \item For any $G \in \mathcal{G}$, the collection of deviated supply bid functions $\mathcal{B}_{\text{supply}}^{\left(\tilde{B}_{\text{supply}}, G\right)}$ to contain every $f \in \mathcal{B}_{\text{supply}}$ such that, for any $m \not\in G$,
                    \begin{align*}
                        f(m) = \tilde{B}_{\text{supply}}(m)
                        \text{,}
                    \end{align*}
                \end{itemize}

                \item [9.)]  For any $\tilde{B}_{\text{supply}} \in \mathcal{B}_{\text{supply}}$ with $\bar{M}\left(\tilde{B}_{\text{supply}}\right) = \{0\}$, for any $\tilde{B}_{\text{demand}} \in \mathcal{B}_{\text{demand}}^{(\{0\})} - \mathcal{C}_{\text{demand}}^{(\{0\})}$,
                \begin{itemize}
                    \item The demand at resale set
                    \begin{align*}
                        D^{\left(\tilde{B}_{\text{supply}}, \tilde{B}_{\text{demand}}\right)} = &
                        \left(\left\{\lim_{z' \to z^+} \max\left(\tilde{D}(z')\right)\right\}_{
                        \substack{
                        z \in \mathbb{R} \cup \{-\infty\}:\\ \lim_{z' \to z^+} e_2^{\left(\tilde{B}_{\text{supply}}, \tilde{B}_{\text{demand}}\right)}(z') \ge c\left(\tilde{B}_{\text{supply}}\right);\\ \left\{e_2^{\left(\tilde{B}_{\text{supply}}, \tilde{B}_{\text{demand}}\right)}(z')\right\}_{z' \in (-\infty,z]} \subseteq \left(-\infty, c\left(\tilde{B}_{\text{supply}}\right)\right]}}
                        \right)
                        \\
                        &\cup
                        \left(
                        \left\{\max\left(\tilde{D}(z)\right)\right\}_{
                        \substack{
                        z \in \mathbb{R}:\\ e_2^{\left(\tilde{B}_{\text{supply}}, \tilde{B}_{\text{demand}}\right)}(z) \ge c\left(\tilde{B}_{\text{supply}}\right);\\ \left\{e_2^{\left(\tilde{B}_{\text{supply}}, \tilde{B}_{\text{demand}}\right)}(z')\right\}_{z' \in (-\infty,z)} \subseteq \left(-\infty, c\left(\tilde{B}_{\text{supply}}\right) \right]}}
                        \right)
                        \text{;}
                    \end{align*}

                    \item The monopoly utility after resale function $u^{\left(\tilde{B}_{\text{supply}}, \tilde{B}_{\text{demand}}\right)}:D^{\left(\tilde{B}_{\text{supply}}, \tilde{B}_{\text{demand}}\right)} \to \mathbb{R}$ such that, for any $d \in D^{\left(\tilde{B}_{\text{supply}}, \tilde{B}_{\text{demand}}\right)}$,
                    \begin{align*}
                        &u^{\left(\tilde{B}_{\text{supply}}, \tilde{B}_{\text{demand}}\right)}(d) 
                        =\\
                        &
                        \begin{cases}
                            -\left(Q\left(\tilde{B}_{\text{supply}}\right) - d\right)c\left(\tilde{B}_{\text{supply}}\right)
                            &\text{ if }
                            d \le Q\left(\tilde{B}_{\text{supply}}\right)
                            \\
                            \max\left(\tilde{D}\left(\tilde{B}_{\text{demand}}(0)\right))\right)
                            \frac{d - Q\left(\tilde{B}_{\text{supply}}\right)}{d - \max\left(\tilde{D}\left(\tilde{B}_{\text{demand}}(0)\right))\right)}
                            \left(\hat{p}_{\text{demand}}\left(\tilde{B}_{\text{demand}}(0)\right)-c\left(\tilde{B}_{\text{supply}}\right)\right)
                            &\text{ if }
                            d > Q\left(\tilde{B}_{\text{supply}}\right)
                        \end{cases}
                        \text{,}
                    \end{align*}
                    
                \end{itemize}
                
            \end{enumerate}
     
            we will have that
            \begin{enumerate}

                \item [0.)] (Equal Supply Bid)
                \begin{align*}
                    \left\vert \text{Range}\left(B_{\text{supply}}\right) \right\vert = 1
                \end{align*}

                \item [1.)] (Non-Negative Utility \& No Demand Bid Deviation) 
                \begin{align*}
                    B_{\text{demand}}
                    \in \mathcal{D}_{\text{demand}}^{\left(B_{\text{supply}}\right)}
                    \text{;}
                \end{align*}

                \item [2.)] (No Supply Bid Deviation) For any $G \in \mathcal{G}$, for any $f_1 \in \mathcal{B}_{\text{supply}}^{\left(B_{\text{supply}}, G\right)}$ such that there exists $g_1 \in \mathcal{D}_{\text{demand}}^{\left(f_1\right)}$ and there exists $g \in \mathcal{D}_{\text{demand}}^{\left(B_{\text{supply}}\right)}$ such that
                \begin{align*}
                    I_G\left(\left\{
                        m \in G:
                        U^{\left(f_1, g_1\right)}(m)
                        \succ_{2}
                        U^{\left(B_{\text{supply}},g\right)}(m)
                    \right\}\right)
                    = 1
                    \text{,}
                \end{align*}
                there exists some $G' \subseteq G$ with $G', G-G' \in \mathcal{G}$ such that there exists some $f_2 \in \mathcal{B}_{\text{supply}}^{\left(f_1, G'\right)}$ such that, for any $g_2 \in \mathcal{D}_{\text{demand}}^{\left(f_1\right)}$,
                \begin{align*}
                    I_{G'}\left(\left\{
                        m \in G':
                        U^{\left(f_2, g_2\right)}(m)
                        \succ_{2}
                        U^{\left(f_1, g_1\right)}(m)
                    \right\}\right)
                    = 1
                    \text{,}
                \end{align*}
                and, for any $g \in \mathcal{D}_{\text{demand}}^{\left(f\right)}$,
                \begin{align*}
                    I_{(G-G')}\left(\left\{
                        m \in G-G':
                        U^{\left(f, g\right)}(m)
                        \succcurlyeq_{2}
                        U^{\left(f_2, g_2\right)}(m)
                    \right\}\right)
                    > 0
                    \text{;}
                \end{align*}

                \item [3.)] (No Supply Bid Deviation: Monopoly with Optimistic Resale) For any $f \in \mathcal{B}_{\text{supply}}^{\left(B_{\text{supply}}, \{0\}\right)}$ such that 
                $\bar{M}\left(f\right) = \{0\}$, for any $g \in \mathcal{B}_{\text{demand}}^{(\{0\})} - \mathcal{C}_{\text{demand}}^{(\{0\})}$, 
                \begin{align*}
                    \sup_{d \in D^{(f,g)}} u^{(f,g)}(d) \le 0
                    \text{,}
                \end{align*}
            \end{enumerate}
            and
            \begin{enumerate}
                \item [A.)] $\mu_{\text{supply}}$ is a measure on $\mathcal{B}\left(\mathbb{R}\right)$ defined such that 
                \begin{align*}
                    \mu_{\text{supply}}\left(\mathbb{R}\right) = Q\left(B_{\text{supply}}\right)
                    \text{,}
                \end{align*}
                and 
                \begin{align*}
                    \text{support}\left(\mu_{\text{supply}}\right) \subseteq 
                    \left\{\bar{\rho}\left(B_{\text{supply}}\right)\right\}
                    \text{;}
                \end{align*}   
                \item [B.)] 
                $\mu_{\text{demand}}$ is a measure on $\mathcal{B}\left(\mathbb{R}\right)$ defined to be $\hat{\mu}^{\left(B_{\text{supply}}, B_{\text{demand}}\right)}$.
            \end{enumerate}

            We will further denote the tuple $\left(B_{\text{supply}}, B_{\text{demand}}\right)$ as a competitive strategy of the competitive market.
        \end{definition}

        Although this definition can be lengthy, it is important to note that this is not the general form of equilibrium since we have assumed the use of the strictly dominant strategy for the suppliers (both inside and outside of the market), the use of the weakly dominant strategy for the demanders (both inside and outside of the market), the joint decision of supply price and supply residual bidding, the specifically-treated instance of deviation of a single agent, and the converging employment of pseudo-equilibrium strategy in the resale bidding in the case of collusion deviation into the system with $0$ Lebesgue measure of winning mediators. However, a more general definition can be easily written but with the use of more complex notations. 

        Moreover, if more regularity is provided as in the informal description in the credit rationing model of the paper ``Credit Rationing in Markets with Imperfect Information", Stiglitz \& Weiss, 1981~\cite{stiglitz}, we can see that our reduced definition is analogous to the informal definition of competitive equilibrium in the paper.
        
    \subsection{Necessary \& Sufficient Conditions}

        From the definition~\ref{definition:competitiveequilibrium}, by considering the lack of a profitable deviation of a single mediator or the lack of a profitable deviation (collusion) of all mediators that if betrayal-free, we will have the following results.

        \subsubsection{Necessary Conditions}
        \label{subsubsection:necessary}

        \paragraph{Demand Side}

            First, we can consider a criteria for a possible demand-side strategy given that every mediator has secured the same differential volume of scarcity from the supply side.
            
            \begin{lemma}
            \label{lemma:nounmacthedsupply}
                (No Unmatched Supply)
                A well-behaved competitive market will have $\left(B_{\text{supply}}, B_{\text{demand}}\right)$ being a competitive strategy only if 
                \begin{align*}
                    Q_{\text{resale}}^{\left(B_{\text{supply}}, B_{\text{demand}}\right)}
                    =
                    0
                    \text{,}
                \end{align*}
                and, for any $r \in \text{Range}\left(B_{\text{demand}}\right)$, \begin{align*}
                    e_1^{\left(B_{\text{supply}}, B_{\text{demand}}\right)}(r) = \hat{p}_{\text{demand}}(r) \ge c\left(B_{\text{supply}}\right)
                    \text{.}
                \end{align*}
            \end{lemma}
            \begin{proof}
                Let $\left(B_{\text{supply}}, B_{\text{demand}}\right)$ be such that $Q_{\text{resale}}^{\left(B_{\text{supply}}, B_{\text{demand}}\right)} > 0$, so there exists some $m \in [0,1)$ such that $B_{\text{demand}}(m) = \bar{r}^{\left(B_{\text{supply}}, B_{\text{demand}}\right)}$, and $p_{\text{high}}^{\left(B_{\text{supply}}, B_{\text{demand}}\right)} < 1$. From the condition~\ref{condition:15} (demand left-continuity), $\lim_{r \to \left(B_{\text{demand}}(m)\right)^-} e_1^{\left(\tilde{B}_{\text{supply}}, \tilde{B}_{\text{demand}}\right)}\left(B_{\text{demand}}(m)\right)
                =
                \lim_{r \to \left(B_{\text{demand}}(m)\right)^-} \hat{p}_{\text{demand}}\left(B_{\text{demand}}(m)\right)
                =
                \hat{p}_{\text{demand}}\left(B_{\text{demand}}(m)\right)
                >
                e_1^{\left(\tilde{B}_{\text{supply}}, \tilde{B}_{\text{demand}}\right)}\left(B_{\text{demand}}(m)\right)$, so there exists some $r' < B_{\text{demand}}(m)$ such that $m$ will profit from changing the quote price to $r'$. Thus, $\left(B_{\text{supply}}, B_{\text{demand}}\right)$ is not a competitive strategy.
    
                If there exists some $m$ such that $e_1^{\left(\tilde{B}_{\text{supply}}, \tilde{B}_{\text{demand}}\right)}\left(B_{\text{demand}}(m)\right) \ne \hat{p}_{\text{demand}}\left(B_{\text{demand}}(m)\right)$, then $B_{\text{demand}}(m) > \bar{r}^{\left(B_{\text{supply}}, B_{\text{demand}}\right)}$, making $e_1^{\left(\tilde{B}_{\text{supply}}, \tilde{B}_{\text{demand}}\right)}\left(B_{\text{demand}}(m)\right) = 0 < c\left(B_{\text{supply}}\right)$, contradicting the non-negative utility criteria.
            \end{proof}  

            \begin{lemma} (Sandwiched Demand Prices)
            \label{lemma:sandwich}
                A well-behaved competitive market will have $\left(B_{\text{supply}}, B_{\text{demand}}\right)$ being a competitive strategy with $Q\left(B_{\text{supply}}\right) > 0$ only if
                \begin{align*}
                    \lambda\left(\left\{m \in [0,1): \max\left(\tilde{D}\left(B_{\text{demand}}(m)\right)\right) \ge Q\left(B_{\text{supply}}\right)\right\}\right) > 0
                    \text{,}
                \end{align*}
                and
                \begin{align*}
                    \lambda\left(\left\{m \in [0,1): \max\left(\tilde{D}\left(B_{\text{demand}}(m)\right)\right) \le Q\left(B_{\text{supply}}\right)\right\}\right) > 0
                    \textbf{, or }
                    \text{Range}\left(B_{\text{demand}}\right) \subseteq
                    \argmax_{r \in \mathbb{R}} \hat{p}_{\text{demand}}(r)
                    \text{.}
                \end{align*}
            \end{lemma}
            \begin{proof}
                Assume that $\left(B_{\text{supply}}, B_{\text{demand}}\right)$ is a competitive strategy. The first claim follows directly from the lemma~\ref{lemma:nounmacthedsupply}.
    
                If $\lambda\left(\left\{m \in [0,1): \max\left(\tilde{D}\left(B_{\text{demand}}(m)\right)\right) \le Q\left(B_{\text{supply}}\right)\right\}\right) = 0$, then possible highest price $\bar{r}^{\left(B_{\text{supply}}, B_{\text{demand}}\right)} = $ $\sup\left(\left\{r \in \mathbb{R}: \max\left(\tilde{D}(r)\right) > 0\right\}\right)$, so the weighted revenue function $e_1^{\left(B_{\text{supply}}, \tilde{B}_{\text{demand}}\right)} = \hat{p}_{\text{demand}}$. Since there is no profitable deviation of a single agent, for any $r \in \text{Range}\left(B_{\text{demand}}\right)$, for any $r' \in \mathbb{R}$, we have that $\hat{p}_{\text{demand}}(r) \ge  \hat{p}_{\text{demand}}(r')$. Thus, $\text{Range}\left(B_{\text{demand}}\right) \subseteq
                    \argmax_{r \in \mathbb{R}} \hat{p}_{\text{demand}}(r)$.
            \end{proof}

            Although these lemmas are from the consideration of the demand-side strategy, they also restrict that the market supply price (as well as supply quantity) cannot be too high.
    
            \begin{lemma}
            \label{lemma:pricemaxmizer}
                (Demand Prices as Conditional Maximizers)
                A well-behaved competitive market will have $\left(B_{\text{supply}}, B_{\text{demand}}\right)$ being a competitive strategy with $Q\left(B_{\text{supply}}\right) > 0$ only if
                \begin{align*}
                    \left\{\hat{p}_{\text{demand}}(r)\right\}_{r \in \text{Range}\left(B_{\text{demand}}\right)}
                    \subseteq
                    \left\{
                        \sup\left(\left\{\hat{p}_{\text{demand}}(r):  r\in \mathbb{R} \land  \max\left(\tilde{D}(r)\right)  
                        \ge Q\left(B_{\text{supply}}\right) \right\}\right)
                    \right\}
                    \cap \left[c\left(B_{\text{supply}}\right),\infty\right)
                    \text{.}
                \end{align*}
            \end{lemma}
            \begin{proof}
                Assume that $\left(B_{\text{supply}}, B_{\text{demand}}\right)$ is a competitive strategy.
                Since a deviation of a demand price quote of a single mediator will not change the weighted revenue function $e_1^{\left(B_{\text{supply}}, B_{\text{demand}}\right)}$. From the the lemma~\ref{lemma:nounmacthedsupply}, $\left\{\hat{p}_{\text{demand}}(r)\right\}_{r \in \text{Range}\left(B_{\text{demand}}\right)} = \left\{e_1^{\left(B_{\text{supply}}, B_{\text{demand}}\right)}(r)\right\}_{r \in \text{Range}\left(B_{\text{demand}}\right)}$ is a singleton whose element, denoted as $p$, is higher or equal to the supply cost $c\left(B_{\text{supply}}\right)$. 
    
                If the equality in the theorem statement is not held, there exists some $r'$ such that $\hat{p}_{\text{demand}}(r') > p$ and $\max\left(\tilde{D}(r')\right)\ge Q\left(B_{\text{supply}}\right)$, we can consider a possible collusion of using $\tilde{B}_{\text{demand}}$ such that $\text{Range}\left(\tilde{B}_{\text{demand}}\right) = \{r'\}$, and get that, for any $m \in [0,1)$, $U^{\left(B_{\text{supply}}, \tilde{B}_{\text{demand}}\right)}(m) \succ_2 U^{\left(B_{\text{supply}}, B_{\text{demand}}\right)}(m)$. This collusion is betrayal-free, because for any $G \in \mathcal{G}$, the joint deviation will not change the weighted revenue function $e_1^{\left(B_{\text{supply}}, \tilde{B}_{\text{demand}}\right)}$, because of the maximal monotonicity of the negative of the real demand correspondence $\tilde{D}$ and that $\max\left(\tilde{D}(r')\right)\ge Q\left(B_{\text{supply}}\right)$, making $\bar{r}^{\left(B_{\text{supply}}, \tilde{B}_{\text{demand}}\right)} \ge r'$. This contradicts with the criteria for an equilibrium strategy.
            \end{proof}

            These $2$ lemmas characterize what a demand side strategy can be if a supply side strategy supports an equilibrium. Although these are necessary conditions from considering a lack of deviation for a single mediator or for all mediators, we will later see that this is also a sufficient condition.

            \begin{definition}
                For a well-behaved competitive market, for any $\left(\rho,q\right) \in \mathbb{R} \times [0,1]$ with $q\max\left(\tilde{S}(\rho)\right) > 0$, we define $\tilde{\mathcal{B}}_{\text{demand}}^{(\rho, q)}$ to be a collection containing every $g \in \mathcal{B}_{\text{demand}}^{[0,1)}$ such that, by denoting $f \in \mathcal{B}_{\text{supply}}$ such that $\text{Range}\left(f\right) = (q,\rho)$, 
                \begin{enumerate}
                    \item $\bar{r}^{\left(f, g\right)} \ge \sup\left(\text{Range}\left(f\right)\right)$
                    \item If $\bar{r}^{\left(f,g\right)} \in \text{Range}\left(g\right)$, then $p_{\text{high}}^{\left(f,g\right)} = 1$;
                    \item 
                    $\left\{\hat{p}_{\text{demand}}(r)\right\}_{r \in \text{Range}\left(f\right)}
                    \subseteq
                    \left\{
                        \sup\left(\left\{\hat{p}_{\text{demand}}(r):  r\in \mathbb{R} \land  \max\left(\tilde{D}(r)\right)  
                        \ge Q\left(f\right) \right\}\right)
                    \right\}
                    \cap \left[c\left(f\right),\infty\right)$;
                    \item At least one of the following holds:
                    \begin{itemize}
                        \item $\text{Range}\left(g\right) \subseteq \argmax_{r \in \mathbb{R}} \hat{p}_{\text{demand}}(r)$, \textbf{or}
                        \item 
                        $\bar{r}^{\left(f,g\right)} \le \sup\left(\text{Range}\left(g\right)\right)$ \textbf{and}, if $\bar{r}^{\left(f,g\right)} \not\in \text{Range}\left(g\right)$, then $p_{\text{high}}^{\left(f,g\right)} = 0$.
                    \end{itemize}
                \end{enumerate}
            \end{definition}

            \begin{corollary}
            \label{corollary:11}
                (Demand Side Necessary Condition)
                
                A well-behaved competitive market will have $\left(B_{\text{supply}}, B_{\text{demand}}\right)$ with $Q\left(B_{\text{supply}}\right) > 0$ being a competitive strategy only if 
                \begin{align*}
                    B_{\text{demand}} 
                    \in
                    \tilde{\mathcal{B}}_{\text{demand}}^{(\bar{\rho}\left(B_{\text{supply}}\right), \bar{q}\left(B_{\text{supply}}\right))}
                    \text{.}
                \end{align*}
            \end{corollary}
            \begin{proof}
                This follows from the lemmas~\ref{lemma:nounmacthedsupply},\ref{lemma:sandwich},\ref{lemma:pricemaxmizer}.
            \end{proof}

        \paragraph{Supply Side}

            We will consider a criteria that allows a supply side strategy to support a competitive equilibrium.

            It suffices to only consider the supply-side strategy that change the market supply price.

            \begin{lemma} (Not Higher Supply Price? I)
            \label{lemma:nothighersupply I}
                A well-behaved competitive market will have $\left(B_{\text{supply}}, B_{\text{demand}}\right)$ being a competitive strategy only if, for any $\left(\rho',q'\right) 
                \in \mathbb{R} \times [0,1]$ with $(\rho',q') \succ_2 \left(\bar{\rho}\left(B_{\text{supply}}\right),\bar{q}\left(B_{\text{supply}}\right)_1\right)$, there does not exist any $r\in \mathbb{R}$ such that
                \begin{align*}
                    \hat{p}_{\text{demand}}(r) > 
                    \hat{c}(q',\rho')
                    \text{,}
                \end{align*}
                and
                \begin{align*}
                    \max\left(\tilde{D}(r)\right) \ge (1-q')\min\left(\tilde{S}(\rho')\right) + q'\max\left(\tilde{S}(\rho')\right) > 0 
                    \text{.}
                \end{align*}
            \end{lemma}
            \begin{proof}
                Assume the setting in the lemma statement but that there exists some $r \in \mathbb{R}$ such that $\hat{p}_{\text{demand}}(r) > \hat{c}(q',\rho')$ and $\max\left(\tilde{D}(r)\right) \ge (1-q')\min\left(\tilde{S}(\rho')\right) + q'\max\left(\tilde{S}(\rho')\right) > 0$. 
                We consider a deviation bid $\tilde{B}_{\text{supply}} \in \mathcal{B}_{\text{supply}}^{\left(B_{\text{supply}}, \{0\}\right)}$ such that $\tilde{B}_{\text{supply}}(0) = (q',\rho')$, making $\bar{M}\left(\tilde{B}_{\text{supply}}\right) = \{0\}$, and $Q\left(\tilde{B}_{\text{supply}}\right) = 0$. By selecting $\tilde{B}_{\text{demand}}: \{0\} \to \{r\}$, the resale quantity $Q_{\text{resale}}^{\left(\tilde{B}_{\text{supply}}, \tilde{B}_{\text{demand}}\right)} = 0$, making $\left(U^{\left(\tilde{B}_{\text{supply}}, \tilde{B}_{\text{demand}}\right)}(0)\right)_1 = Q\left(\tilde{B}_{\text{supply}}\right) \left(\hat{p}_{\text{demand}}(r) - \hat{c}(q',\rho')\right) > 0$. This suggests that the mediator $0$ can deviate to such strategy, contradicting with the assumption. 
            \end{proof}

            This lemma suggests that the supply price (as well as the acquired supply volume) cannot be too low, since a mediator can deviate and become a ``second-stage monopoly".

            \begin{lemma} (Not Higher Supply Price? II)
            \label{lemma:nothighersupply II}
                A well-behaved competitive market will have $\left(B_{\text{supply}}, B_{\text{demand}}\right)$ being a competitive strategy only if, for any $\left(\rho',q'\right) 
                \in \mathbb{R} \times [0,1]$ with $(\rho',q') \succ_2 \left(\bar{\rho}\left(B_{\text{supply}}\right),\bar{q}\left(B_{\text{supply}}\right)_1\right)$, for all $(q,\rho) \in \text{}
                \left(\left(B_{\text{supply}}\right)_2,\left(B_{\text{supply}}\right)_1\right)$, there does not exist any $(r, r') \in \mathbb{R}^2$ such that
                \begin{align*}
                    \hat{p}_{\text{demand}}(r) &\ge \hat{c}(q',\rho')
                    \text{,}\\
                    \hat{p}_{\text{demand}}(r') &> \hat{c}(q',\rho')
                    \text{,}
                \end{align*}
                and
                \begin{align*}
                    \max\left(\tilde{D}(r)\right) > (1-q')\min\left(\tilde{S}(\rho')\right) + q'\max\left(\tilde{S}(\rho')\right) > 0 
                    \text{.}
                \end{align*}
            \end{lemma}
            \begin{proof}
                Assume the setting in the lemma statement but that there exists some $r \in \mathbb{R}$ such that $\hat{p}_{\text{demand}}(r) \ge 
                \hat{c}(q',\rho')$ and $\max\left(\tilde{D}(r)\right) > (1-q')\min\left(\tilde{S}(\rho')\right) + q'\max\left(\tilde{S}(\rho')\right) > 0$, and there exists some $r' \in \mathbb{R}$ such that $\hat{p}_{\text{demand}}(r') > \hat{c}(q',\rho')$. 
                We consider a deviation bid $\tilde{B}_{\text{supply}} \in \mathcal{B}_{\text{supply}}^{\left(B_{\text{supply}}, \{0\}\right)}$ such that $\tilde{B}_{\text{supply}}(0) = (q',\rho')$, making $\bar{M}\left(\tilde{B}_{\text{supply}}\right) = \{0\}$, and select $\tilde{B}_{\text{demand}}: \{0\} \to \{r'\}$. If $Q_{\text{resale}}^{\left(\tilde{B}_{\text{supply}}, \tilde{B}_{\text{demand}}\right)} = 0$, the lemma~\ref{lemma:nothighersupply I} will be contradicted, so $Q_{\text{resale}}^{\left(\tilde{B}_{\text{supply}}, \tilde{B}_{\text{demand}}\right)} > 0$, and $r< r' = \bar{r}^{\left(\tilde{B}_{\text{supply}}, \tilde{B}_{\text{demand}}\right)}$. From the condition~\ref{condition:15} (demand left-continuity), there exists some $x < r'$ such that $e_2^{\left(\tilde{B}_{\text{supply}}, \tilde{B}_{\text{demand}}\right)}(x) > 
                c\left(\tilde{B}_{\text{supply}}\right)$, so $\tilde{B}_{\text{demand}} \in \tilde{B}^{\left(\tilde{B}_{\text{supply}}\right)} - \tilde{C}^{\left(\tilde{B}_{\text{supply}}\right)}$. The condition~\ref{condition:15} (demand left-continuity) also suggests that there exists some $d \ge \max\left(\tilde{D}(r)\right)$ such that $d \in D^{\left(\tilde{B}_{\text{supply}}, \tilde{B}_{\text{demand}}\right)}$. Since $u^{\left(\tilde{B}_{\text{supply}}, \tilde{B}_{\text{demand}}\right)}$ is non-decreasing, $\sup_{d \in D^{\left(\tilde{B}_{\text{supply}}, \tilde{B}_{\text{demand}}\right)}} u^{\left(\tilde{B}_{\text{supply}}, \tilde{B}_{\text{demand}}\right)}(d) \ge u^{\left(\tilde{B}_{\text{supply}}, \tilde{B}_{\text{demand}}\right)}\left(\max\left(\tilde{D}(r)\right)\right) > 0$. Thus, the mediator $0$ will profit from this deviation.
            \end{proof}
            
            This lemma further suggests that the supply price cannot be too low. Next, we can consider an additional lemma restricting the market supply price to be not too high in addition to the result we have already got form the lemma~\ref{lemma:sandwich}. 

            \begin{lemma} (Not Lower Supply Price?)
            \label{lemma:notlowersupply}
                A well-behaved competitive market will have $\left(B_{\text{supply}}, B_{\text{demand}}\right)$ being a competitive strategy only if, for any $\left(\rho',q'\right) 
                \in \mathbb{R} \times [0,1]$ with $\left(\bar{\rho}\left(B_{\text{supply}}\right),\bar{q}\left(B_{\text{supply}}\right)_1\right) \succ_2 (\rho',q')$, by defining $\tilde{B}_{\text{supply}}: [0,1) \to \left\{(q',\rho')\right\}$, \textbf{if} there exists a demand-side strategy $\tilde{B}_{\text{demand}} \in \tilde{\mathcal{B}}_{\text{demand}}^{(\rho', q')}$ such that
                \begin{align*}
                    \left(c\left(\tilde{B}_{\text{supply}}\right) - \hat{p}_{\text{demand}}\left(\tilde{B}_{\text{demand}}(0)\right)\right)Q\left(\tilde{B}_{\text{supply}}\right)
                    >
                    0
                    \text{,}
                \end{align*} \textbf{then} there exists some $(\rho'',q'') \succ_2 (\rho',q')$ such that at least one of the following holds:
                \begin{itemize}
                    \item There exists some $r \in \mathbb{R}$ such that $\hat{p}_{\text{demand}}(r) > \hat{c}(q'',\rho'')$ and $\max\left(\tilde{D}(r)\right) \ge (1-q'')\min\left(\tilde{S}(\rho'')\right) + q''\max\left(\tilde{S}(\rho'')\right)$, \textbf{or}
                    \item There exists some $(r, r') \in \mathbb{R}^2$ such that $\hat{p}_{\text{demand}}(r) \ge \hat{c}(q'',\rho'')$, $\hat{p}_{\text{demand}}(r') > \hat{c}(q'',\rho'')$, and $\max\left(\tilde{D}(r)\right) > (1-q'')\min\left(\tilde{S}(\rho'')\right) + q''\max\left(\tilde{S}(\rho'')\right)$.
                \end{itemize}
            \end{lemma}
            \begin{proof}
                Assume the setting to be true except for the statement after the word ``then" to be false. We will get that everyone will benefit from the deviation of the strategy $\left(B_{\text{supply}}, B_{\text{demand}}\right)$ to the strategy $\left(\tilde{B}_{\text{supply}}, \tilde{B}_{\text{demand}}\right)$. 
                
                First, we will show that there does not exists a betrayal-free collusion demand-side strategy. 
                
                If there exists some $G \in \mathcal{G}$, and there exists some $g \in \mathcal{B}_{\text{demand}}^{\left(\tilde{B}_{\text{supply}}, \tilde{B}_{\text{demand}}, G\right)}$ such that, by denoting $\bar{G} = \left\{m \in G: \left(U^{\left(\tilde{B}_{\text{supply}}, g\right)}(m)\right)_{2} > \left(U^{\left(\tilde{B}_{\text{supply}}, \tilde{B}_{\text{demand}}\right)}(m)\right)_{2}\right\}$, $I_G\left(\bar{G}\right) =1$. Therefore, $\lambda(G) = \lambda\left(\bar{G}\right) > 0$, $\bar{r}^{\left(\tilde{B}_{\text{supply}}, g\right)} \in \text{support}\left(\mu^{\left(\tilde{B}_{\text{supply}}, g\right)}\right)$ and $p_{\text{high}}^{\left(\tilde{B}_{\text{supply}}, g\right)} \in \text{support}\left(\mu^{\left(\tilde{B}_{\text{supply}}, g\right)}\right) \in (0,1)$. From the condition~\ref{condition:15} (demand left-continuity), there exists some $\bar{r}' < \bar{r}^{\left(\tilde{B}_{\text{supply}}, g\right)}$ such that $e_1^{\left(\tilde{B}_{\text{supply}}, g\right)}\left(\bar{r}'\right) > e_1^{\left(\tilde{B}_{\text{supply}}, g\right)}\left(\bar{r}^{\left(\tilde{B}_{\text{supply}}, g\right)}\right)$. Denote $M = \left\{m \in G: g(m) = \bar{r}^{\left(\tilde{B}_{\text{supply}}, g\right)} \right\}$. Thus, there exists a betrayal of $G' \in \mathcal{G}$ and $G' \subseteq M$ with $Q\left(\tilde{B}_{\text{supply}}\right)\lambda(G') = \max\left(\tilde{D}\left(\bar{r}'\right)\right)$ by deviating to quoting a demand price at $\bar{r}'$. Thus, $\lambda\left(M-G'\right)>0$, and the payoff for every $m \in M-G'$ is negative.
                
                Next, we will show that there will not be a profitable deviation of supply-side bid. 
                
                If there exists $G \in \mathcal{G} - \{[0,1)\}$, there exists $f_1 \in \mathcal{B}_{\text{supply}}^{\left(\tilde{B}_{\text{supply}}, G\right)}$, and there exists $g_1 \in \mathcal{C}^{(f_1)}_{\text{demand}}$ such that, by denoting $\bar{G} = \left\{m \in G:  U^{(f_1, g_1)}(m) \succ_2 U^{\left(\tilde{B}_{\text{supply}}, \tilde{B}_{\text{demand}}\right)}(m) \right\}$,  $I_G\left(\bar{G}\right) = 1$. Thus, there exists some $\ubar{r} \in \text{support}\left(\hat{\mu}_{\text{demand}}^{(f_1,g_1)}\right)$ such that $\hat{p}_{\text{demand}}(r) > 0$. Therefore, there will also be a profitable deviation in supply-side strategy by a monopoly that gives a non-zero utility, by choosing the same supply price, the same supply residual ratio, and the demand price being $\ubar{r}$. This suggests that we can restrict the consideration upon the supply-side deviation of a single agent. However, from the proofs of the lemmas~\ref{lemma:nothighersupply I},\ref{lemma:nothighersupply II}, there does not exists such deviation.

                Thus, we will get that $\left(B_{\text{supply}}, B_{\text{demand}}\right)$ is not a competitive strategy.
            \end{proof}

            These lemmas will restrict the quantity that an acquired supply could be. Moreover, this will be a sufficient condition.

        \subsubsection{Sufficient Conditions}

            We consider the main theorem of this paper that provides a matching sufficient and necessary conditions for a competitive equilibrium.

            \begin{theorem}
            \label{theorem:mainTheorem}
            (Necessary and Sufficient Condition)
                For any well-behaved competitive market, by defining $\tilde{P}$ to be a set containing every $(\rho,q) \in \mathbb{R} \times [0,1]$ such that [there exists some $r\in \mathbb{R}$ such that $\hat{p}_{\text{demand}}(r) > \hat{c}(q,\rho)$ and $\max\left(\tilde{D}(r)\right) \ge (1-q)\min\left(\tilde{S}(\rho)\right) + q\max\left(\tilde{S}(\rho)\right) > 0$] or [there exists some $(r, r') \in \mathbb{R}^2$ such that $\hat{p}_{\text{demand}}(r) \ge \hat{c}(q,\rho)$, $\hat{p}_{\text{demand}}(r') > \hat{c}(q,\rho)$, and $\max\left(\tilde{D}(r)\right) > (1-q)\min\left(\tilde{S}(\rho)\right) + q\max\left(\tilde{S}(\rho)\right) > 0$], 
                a tuple $\left(B_{\text{supply}}, B_{\text{demand}}\right)$ is a competitive strategy if and only if
                \begin{align*}
                    \max\left(\tilde{D}(r)\right) > (1-q')\min\left(\tilde{S}(\rho')\right) + q'\max\left(\tilde{S}(\rho')\right) > 0 
                    \text{.}
                \end{align*}
                
                \begin{enumerate}
                    \item $\left\vert \text{Range}\left(B_{\text{supply}}\right) \right\vert = 1$ and $B_{\text{demand}} \in \mathcal{B}^{[0,1)}$;
                    
                    \item If $Q\left(B_{\text{supply}}\right) > 0$, then $B_{\text{demand}}
                    \in
                    \tilde{\mathcal{B}}_{\text{demand}}^{(\bar{\rho}\left(B_{\text{supply}}\right), \bar{q}\left(B_{\text{supply}}\right))}$;
                    
                    \item There does not exist $\left(\rho,q\right) \in \tilde{P}$ such that $(\rho,q) \succ_2 \left(\bar{\rho}\left(B_{\text{supply}}\right),\bar{q}\left(B_{\text{supply}}\right)_1\right)$;

                    \item If there exists some $\left(\rho,q\right) \in \left\{
                    \left(\rho',q'\right) \in \mathbb{R} \times [0,1]: \tilde{\mathcal{B}}^{(\rho',q')}_{\text{demand}} \ne \emptyset
                    \right\}$ such that $(\rho,q) \succ_2 \left(\bar{\rho}\left(B_{\text{supply}}\right),\bar{q}\left(B_{\text{supply}}\right)_1\right)$, $q\max\left(\tilde{S}(\rho)\right) > 0$, and $\hat{c}(q,\rho) > \sup_{r: \max\left(\tilde{D}(r)\right) \ge (1-q)\min\left(\tilde{S}(\rho)\right)+q\max\left(\tilde{S}(\rho)\right)} \hat{p}_{\text{demand}}(r)$
                    , then there exists some $\left(\rho'',q''\right) \in \tilde{P}$ such that $\left(\rho'',q''\right) \succ_2 \left(\rho',q'\right) $.
                \end{enumerate}
            \end{theorem}
            \begin{proof}
                This is a necessary condition because of the corollary~\ref{corollary:11} and the lemmas~\ref{lemma:nothighersupply I},\ref{lemma:nothighersupply II},\ref{lemma:notlowersupply}. The proof for the sufficiency is equivalent to the proof for the lemma\ref{lemma:notlowersupply}.
            \end{proof}

    \subsection{Graphical Approach}
    \label{subsection:graph}

        From the definition~\ref{definition:competitiveequilibrium}, the collection of all competitive market equilibria of a well-behaved competitive market $\circled{C}^{ \left(\mathcal{A}_{\text{supply}}, \mathcal{A}_{\text{demand}},
            \bar{\Omega}_{\text{supply}}, \bar{\Omega}_{\text{demand}}, w,\omega,\tilde{w},\tilde{\omega}\right)}
            \left(
                \hat{\Pi}_{\text{supply}}, \hat{\Pi}_{\text{demand}}
            \right)$ can be fully determined by
        \begin{itemize}
            \item [1.)] The real supply correspondence $\tilde{S}: \mathbb{R} \rightrightarrows \mathbb{R}^+_0$;
            \item [2.)] The real demand correspondence $\tilde{D}: \mathbb{R} \rightrightarrows \mathbb{R}^+_0$;
            \item [3.)] The supply cost functions $\bar{p}_{\text{supply}}: \mathbb{R} \to \mathbb{R}^+_0$ and $\ubar{p}_{\text{supply}}: \mathbb{R} \to \mathbb{R}^+_0$;
            \item [4.)] The demand revenue function $\hat{p}_{\text{demand}}: \mathbb{R} \to \mathbb{R}$.
        \end{itemize}

        We will introduce a graphical method to provide qualitative analysis on some qualitative characteristics of interest, such as the existence of a competitive market equilibrium, the uniqueness of the competitive market equilibrium, the multiplicity of demand prices\footnote{From the market mechanism, the supply price is unique or $-\infty$.}, the existence of demand rationing equilibrium, and the existence of hard demand rationing. Some quantitative result, such as the total transfer of scarcity from the suppliers through the mediators to the demander.

        \subsubsection{Demand Graph}

            \begin{definition}
                (Highest Revenue)
                For a well-behaved competitive market, the supremum revenue $p^*_{\text{demand}}$ is defined such that
                \begin{align*}
                    p^*_{\text{demand}} = \sup_{r \in \mathbb{R}}\left(\hat{p}_{\text{demand}}(r)\right)
                    \text{.}
                \end{align*}
            \end{definition}

            \begin{definition}
                (Demand Graph)
                For a well-behaved competitive market, the demand plotting function $\hat{d}: \mathbb{R} \to \mathbb{R} \times \left[0, D_{\text{max}}\right]$ is defined such that, for any $r\in\mathbb{R}$,
                \begin{align*}
                    \hat{d}(r) = \left(\hat{p}_{\text{demand}}(r), \max\left(\tilde{D}\left(r\right)\right)\right)
                    \text{,}
                \end{align*}
                and the demand graph $D$ is defined such that
                \begin{align*}
                    D = \text{Range}\left(\hat{d}\right) - \{(0,0)\}
                    \text{.}
                \end{align*}
            \end{definition}

            \begin{corollary}
                $D \cap \left(\mathbb{R} \times \{0\}\right) = \emptyset$.
            \end{corollary}
            \begin{proof}
                For any $r \in \mathbb{R}$, if $\max\left(\tilde{D}(r)\right) = 0$, then $\hat{p}_{\text{demand}}(r) = 0$ by the definition~\ref{definition:realdemand}.
            \end{proof}
            
            \begin{corollary}
            \label{corollary:demandpricefinding}
                (Demand Price Finding)
                There exists a unique injective function, denoted as the demand price finding function $\hat{r}: D \to \mathbb{R}$ such that, for any $(p,d) \in D$,
                \begin{align*}
                    (p,d) = \hat{d}\left(\hat{r}(p,d)\right)
                    \text{.}
                \end{align*}
                We will define the function $\hat{r}$ as such.
            \end{corollary}
            \begin{proof}
                It suffices to show that, for any $r_1, r_2 \in \mathbb{R}$, if $r_1 \ne r_2$, and $\max\left(\tilde{D}\left(r_1\right)\right) = \max\left(\tilde{D}\left(r_2\right)\right) > 0$, then $\hat{p}_{\text{demand}}\left(r_1\right) \ne \hat{p}_{\text{demand}}\left(r_2\right)$. This follows directly from the condition~\ref{condition:14} (demand monotone trend).
            \end{proof}

            This suggests that, even when the graph is represented in the revenue space, which is the same as the cost space, and the volume space, we are still equipped with the demand price finding function $\hat{r}$ to trace which demand price is needed to achieve such revenue and demand volume (as long as the demand volume is non-zero).

            \begin{corollary}
                For any $(p,d), (p',d') \in D$, if $(-d',p') \succ_2 (-d,p)$, then $\hat{r}(-d',p') > \hat{r}(-d,p)$.
            \end{corollary}
            \begin{proof}
                The proof follows directly from the proof for the corollary~\ref{corollary:demandpricefinding}.
            \end{proof}

            \begin{definition}
                (Augmented Demand Graph)
                For a well-behaved competitive market, the augmented demand graph is $A_D \subseteq \mathbb{R}_0^+ \times \left(0, D_{\text{max}}\right]$ uniquely defined such that
                \begin{align*}
                    A_{D} = 
                    \bigcup_{(p,d) \in D}
                    \left(
                        \left((-\infty, p] \cap \mathbb{R}^+_0 \right) \times (0,d]\right)
                    \text{.}
                \end{align*}
            \end{definition}

            \begin{definition}
                (Vertical Border Demand Graph)
                For a well-behaved competitive market, the vertical border demand graph $V_D^{(0)} \subseteq A_{D}$ is defined such that
                \begin{align*}
                    V_D^{(0)} = 
                    \left(\bigcup_{(p,d) \in A_D^{(0)}} 
                        \left\{
                        \left(
                            \sup\left(\left\{p':(p',d) \in A_D\right\}\right)
                            , d
                        \right)
                        \right\}
                    \right) \cap A_D
                    \text{.}
                \end{align*}
            \end{definition}

            \begin{definition}
                (Farthest Vertical Border Demand Graph)
                For a well-behaved competitive market, the farthest vertical border demand graph $V_D^{(1)} \subseteq V_{D}^{(0)}$ is defined such that
                \begin{align*}
                    V_D^{(1)} = 
                    \left(
                        \left(-\infty, p^*_{\text{demand}}\right]
                        \times \bigcup_{(p,d) \in D: p = p^*_{\text{demand}}} [0,d]
                    \right)
                    \cap
                    V_D^{(0)}
                    \text{.}
                \end{align*}
            \end{definition}

            \begin{definition}
                (Sharp Vertical Border Demand Graph)
                For a well-behaved competitive market, the sharp vertical border demand graph $V_D^{(2)} \subseteq V_{D}^{(0)}$ is defined such that
                \begin{align*}
                    V_D^{(2)} = 
                    V_D^{(0)} - V_D^{(1)}
                    -
                    \bigcup_{(p,d) \in V_D^{(0)}} 
                        \left\{\left(p,
                        \sup\left(\left\{d':(p,d') \in V_D^{(0)} \right\}\right)\right)\right\}
                    \text{.}
                \end{align*}
            \end{definition}

            \begin{definition}
                (Admissible Vertical Border Demand Graph)
                For a well-behaved competitive market, the admissible vertical border demand graph $V_D^{(3)} \subseteq V_{D}^{(0)}$ is defined such that
                \begin{align*}
                    V_D^{(3)} = 
                    \left(
                    \left(
                        \bigcup_{(p,d) \in D} 
                        \left(
                            \{p\}  \times
                            [d,\infty)
                        \right)
                    \right)
                    \cap 
                    V_D^{(0)}
                    \right) \cup V_D^{(1)} \cup \{(0,0)\}
                    \text{.}
                \end{align*}
            \end{definition}

            We get that, for any point $(p,s) \in V_D^{(3)} - V_D^{(1)} - \{(0,0)\}$, there exists a demand price measure that satisfies the demand revenue is the same across the support of the measure and the resale volume for a demand-side bid strategy $B_{\text{demand}}$ that supports the measure will be $0$.

            Moreover, for any point $(p,s) \in V_D^{(1)}$, there exists a demand price measure that satisfies the demand revenue is the same across the support of the measure and the highest possible demand price volume for a demand-side bid strategy $B_{\text{demand}}$ that supports the measure will be higher or equal to the supremum of the measure support.
            
            \begin{definition}
                (Demand Price Measures Finding)
                For a well-behaved competitive market, the demand price measures finding correspondence $M_{\text{demand}}: V_D^{(3)} \rightrightarrows \mathcal{M}\left(\mathbb{R}\right)$ is defined such that, $M_{\text{demand}}(0,0)$ only contains a measure $\mu$ with $\mu\left(\mathbb{R}\right) = 0$, and, for any $(p,s) \in V_D^{(3)} - \{(0,0)\}$, $M_{\text{demand}}(p,s)$ is a collection of every measure $\mu$ upon $\mathcal{B}\left(\mathbb{R}\right)$ such that
                \begin{align*}
                    \text{support}\left(\mu\right) \in W^{\left(\left\{\hat{r}(p,d) : (p,d) \in D\right\}\right)}
                    \text{,}
                \end{align*}
                \begin{align*}
                    \mu\left(\mathbb{R}\right) = s\text{,}
                \end{align*}
                and
                \begin{align*}
                    \sum_{r \in \text{support}\left(\mu\right)} \frac{\mu\left(\{r\}\right)}{\max\left(\tilde{D}(r)\right)}
                    \in \left[\mathbf{1}_{(p,s) \not\in V^{(1)}_d} , 1\right]
                    \text{.}
                \end{align*}
            \end{definition}

            \begin{corollary}
                For any $(p,s) \in V_D^{(3)}$, $M_{\text{demand}}(p,s) \ne \emptyset$.
            \end{corollary}
            \begin{proof}
                If $(p,s) \in D$, we have that $M_{\text{demand}}(p,s)$ contains a measure $\mu$ such that $\mu\left(\left\{\hat{r}\left(p,s\right)\right\}\right) = \mu\left(\mathbb{R}\right) = s$.

                If $(p,s) \in V_D^{(1)}$, there exists some $d \ge s$ such that $(p,d) \in D$, and we have that $M_{\text{demand}}(p,s)$ contains a measure $\mu$ such that $\mu\left(\left\{\hat{r}\left(p,d\right)\right\}\right) = \mu\left(\mathbb{R}\right) = s$.

                If $(p,s) \in V_D^{(3)} - V_D^{(1)} - D$, there exist some $d_1 > s$ and some $d_2 < s$ such that $(p,d_1), (p,d_2) \in D$. We have that $M_{\text{demand}}(p,s)$ contains a measure $\mu$ such that $\mu\left(\left\{\hat{r}(p,d_1)\right\}\right) = \frac{d_1\left(s-d_2\right)}{d_1-d_2}$, $\mu\left(\left\{\hat{r}(p,d_2)\right\}\right) = \frac{d_2\left(d_1-s\right)}{d_1-d_2}$, and $\mu\left(\mathbb{R}\right) =s$.
            \end{proof}

        \subsubsection{Supply Graph}

            \begin{definition}
                (Supply Graph)
                For a well-behaved competitive market, the supply plotting function $\hat{g}: \mathbb{R} \times [0,1] \to \mathbb{R}^+_0 \times \left[0, S_{\text{max}}\right]$ is defined such that, for any $\rho \in \mathbb{R}$, $q \in [0,1]$,
                \begin{align*}
                    \hat{g}(\rho,q) =
                    \lim_{x \to \left((1-q)\min\left(\tilde{S}(\rho)\right)
                        + q\max\left(\tilde{S}(\rho)\right)\right)^+}
                        \left(
                         \frac{(1-q)\ubar{p}_{\text{supply}}(\rho)\min\left(\tilde{S}(\rho)\right)
                        + q\bar{p}_{\text{supply}}(\rho)\max\left(\tilde{S}(\rho)\right)}{x}
                        ,
                        x
                        \right)
                    \text{,}
                \end{align*}
                and the supply graph $G$ is defined such that
                \begin{align*}
                    G = 
                    \text{Range}\left(\hat{g}\right)
                    - \{(0,0)\}
                    \text{.}
                \end{align*}
            \end{definition}
        
            \begin{corollary}
            \label{corollary:no0supply}
                $G \cap \left(\mathbb{R} \times \{0\}\right) = \emptyset$.
            \end{corollary}
            \begin{proof}
                This follows from the definition.
            \end{proof}

            Since the supply is determined by both supply price $\rho$ and the residual quantity $q$, there can be multiple points $(p,s) \in G$ that can be supported by multiple tuples of supply price and residual quantity $(\rho, q)$. In fact, there exists uncountably many $(p,s) \in G$ that can be supported by uncountably many tuples of supply price and residual quantity $(\rho, q)$, because there exist uncountably many $\rho \in \mathbb{R}$ such that $\left\vert \tilde{S}(\rho) \right\vert = 1$ from that the correspondence $\tilde{S}(\rho)$ is bounded and maximal monotone (as shown in the corollary~\ref{corollary:supplymaximalmonotone}).

            Still, by restricting the value of residual quantity $q$ to be more specific, we get that we can trace back point $(p,s) \in G - \{(0,0)\}$ to a unique tuple of supply price and residual quantity $(\rho, q)$.

            \begin{corollary}
            \label{corollary:supplypricefinding}
                (Supply Price Finding)
                There exists a unique injective function, denoted as the supply price finding function $\hat{\rho}: G \to \mathbb{R}$ such that, for any $(p,s) \in G$,
                \begin{align*}
                    (p,s) \in
                    \left\{\hat{g}\left(\hat{\rho}(p,s), q\right)\right\}_{q \in [0,1]}\text{.}
                \end{align*}
            \end{corollary}
            \begin{proof}
                Consider any $(p,s) \in G$. From the corollaries~\ref{corollary:supplymaximalmonotone},\ref{corollary:no0supply}, the set $X = \left\{\rho \in \mathbb{R}:s \in \tilde{S}(\rho)\right\}$ is a closed interval with $\inf(X) \in \mathbb{R}$. By denoting $A = \left\{x \in X: (p,s) \in \left\{\hat{g}\left(\rho, q\right)\right\}_{q \in [0,1]}\right\}$, it suffices to show that $\vert A \vert \le 1$. Since $A \subseteq X$, we can focus on the case when $\vert X \vert > 1$. We have that $\max\left(\tilde{S}\left(\inf\left(X\right)\right)\right) = s$, and $\left\{\tilde{S}(x)\right\}_{x \in \text{int}(X)} = \{\{s\}\}$. Thus, $\text{int}(X) \subseteq \left\{x \in \mathbb{R}: \ubar{p}_{\text{supply}}(x) = \bar{p}_{\text{supply}}(x)\right\}$, the function $\ubar{p}_{\text{supply}}$ is increasing on $X -\left\{ \inf\left(X\right)\right\}$, and the function $\bar{p}_{\text{supply}}$ is increasing on $X -\left\{ \sup\left(X\right)\right\}$, because if the condition~\ref{condition:11} (supply monotone trend). Therefore, if $p = \bar{p}_{\text{supply}}\left(\inf\left(X\right)\right)$, then $A = \{\inf\left(X\right)\}$, and, if $p \ne \bar{p}_{\text{supply}}\left(\inf\left(X\right)\right)$, then $A = \left\{x \in X - \left\{ \inf\left(X\right)\right\}: p = \ubar{p}_{\text{supply}}(x)\right\}$, which can have at most $1$ element.
            \end{proof}

            \begin{definition}
                (Supply Residual Ratio Finding)
                For a well-behaved competitive market,
                the supply residual ratio finding function $\hat{q}: G - \{0\} \to [0,1]$ such that, for any $(p,s) \in G$ with $s > 0$,
                \begin{align*}
                    \hat{q}(p,s)
                    =
                    1 - \lim_{x \to 
                    \left(\max\left(\tilde{S}\left(\hat{\rho}(p,s)\right)\right) - \min\left(\tilde{S}\left(\hat{\rho}(p,s)\right)\right)\right)^+} \frac{\max\left(\tilde{S}\left(\hat{\rho}(p,s)\right)\right) - s}{x}
                    \text{.}
                \end{align*}
            \end{definition}

            This restricts that, any time when the total residual volume is $0$, the residual ratio selected will be $1$.
            
            \begin{corollary}
                For any $(p,s) \in G$ with $s > 0$,
                \begin{align*}
                    (p,s) =
                    \hat{g}\left(\hat{\rho}(p,s), \hat{q}(p,s)\right)
                    \text{.}
                \end{align*}
            \end{corollary}
            \begin{proof}
                Consider any $(p,s) \in G$ with $s > 0$. Denote $A = \left\{q \in [0,1]: (p,s) = \hat{g}\left(\hat{\rho}(p,s), \hat{q}(p,s)\right)\right\}$.
                
                If $\left\vert \tilde{S}\left(\hat{\rho}(p,s)\right) \right\vert = 1$, then, $A = [0,1] \ni \hat{q}(p,s)$. 

                If $\left\vert \tilde{S}\left(\hat{\rho}(p,s)\right) \right\vert > 1$, then, 
                $A = \left\{q \in [0,1]: (1-q) \min\left(\tilde{S}\left(\hat{\rho}(p,s)\right)\right) + q \max\left(\tilde{S}\left(\hat{\rho}(p,s)\right) \right) = s\right\} = \{\hat{q}(p,s)\}$.
            \end{proof}

            \begin{corollary}
            \label{corollary:19}
                For any $(p,s), (p',s') \in S$, if $(s',p') \succ_2 (s,p)$, then $\left(\hat{\rho}(s',p'), \hat{q}(s',p')\right) > \left(\hat{\rho}(s,p), \hat{q}(s,p)\right)$.
            \end{corollary}
            \begin{proof}
                The proof follows directly from the corollary~\ref{corollary:supplymaximalmonotone} and the proof for the corollary~\ref{corollary:supplypricefinding}.
            \end{proof}

            \begin{definition}
                (Augmented Supply Graph)
                For a well-behaved competitive market, the augmented supply graph $A_S \subseteq \mathbb{R} \times \left[0, S_{\text{max}}\right]$ is defined such that
                \begin{align*}
                    A_S =
                        \bigcup_{(p,s) \in S} 
                        \left(
                            [p,\infty)  \times
                            \{s\}
                        \right)
                    \text{.}
                \end{align*}
            \end{definition}

            We can see that, for any point $(p,s) \in S$, we can find a supply price measure such that there exists some supply-side strategy $B_{\text{supply}}$ that supports such supply price measure and satisfies $\left(c\left(B_{\text{supply}}\right),Q\left(B_{\text{supply}}\right)\right) = (p,s)$.

            Moreover, there exists a unique supply price measure, following from the corollary~\ref{corollary:supplypricefinding}.
            
            \begin{definition}
                (Supply Price Measure Finding)
                For a well-behaved competitive market, the supply price measure finding function supply graph $m_{\text{supply}}: S \cup \{(0,0)\} \to \mathcal{M}\left(\mathbb{R}\right)$ is defined such that, for any $(p,s) \in S \cup \{(0,0)\}$,
                \begin{align*}
                    \left(m_{\text{supply}}(p,s)\right)\left(\mathbb{R}\right) = s\text{,}
                \end{align*}
                and
                \begin{align*}
                    \text{support} \left(m_{\text{supply}}(p,s)\right) =
                    \begin{cases}
                        \left\{\hat{\rho}\left(p,s\right)\right\} &\text{ if } (p,s) \ne (0,0)\\
                        \emptyset &\text{ if } (p,s) = (0,0)
                    \end{cases}
                    \text{.}
                \end{align*}
            \end{definition}

            \paragraph{Supply Side Necessary Conditions}

            We get that the necessary conditions on the supply side lemmas in the subsubsection~\ref{subsubsection:necessary} can be represented more concisely in a graphical wording.

            For a competitive strategy $\left(B_{\text{supply}}, B_{\text{demand}}\right)$ with $Q\left(B_{\text{supply}}\right) > 0$, we will have that by denoting $(p_0, s_0) = \left(c\left(B_{\text{supply}}\right), Q\left(B_{\text{supply}}\right)\right) \in S$,
            \begin{enumerate}
                \item The lemma~\ref{lemma:pricemaxmizer} suggests that $\left\{p \in [p_0, \infty): (p,s_0) \in A_D\right\}$ is a closed compact interval;
                
                \item The lemma~\ref{lemma:sandwich} further suggests that
                $\left(\sup\left(\left\{p : (p,s_0) \in A_D\right\}\right), s_0\right) \in V^{(3)}_D$;

                \item The lemma~\ref{lemma:nothighersupply I} suggests that, for any $(p,s) \in S$, if $(s,p) \succ_2 (s_0,p_0)$, then $\lambda\left(\left\{p' \in [p, \infty): (p',s) \in A_D\right\}\right) = 0$;

                \item The lemma~\ref{lemma:nothighersupply II} further suggests that, for any $(p,s) \in S$, if $(s,p) \succ_2 (s_0,p_0)$, then $(p,s) \not\in V_D^{(2)}$;

                \item The lemma~\ref{lemma:notlowersupply} suggests that, for any $(p,s) \in S$ with $s < s_0$, if $\lambda\left(\left\{p' \in [p, \infty): (p',s) \in A_D\right\}\right) > 0$, then there exists some $(p',s') \in S$ such that 
                \begin{itemize}
                    \item $(s',p') \succ_2 (s,p)$;
                    \item $\lambda\left(\left\{p'' \in [p', \infty): (p'',s') \in A_D\right\}\right) > 0$, \textbf{or} $(p,s) \in V_D^{(2)}$.
                \end{itemize}
            \end{enumerate}

        \subsubsection{Graphical Algorithm}

            We consider the following algorithm (the algorithm~\ref{alg:algorithm}).\\
            \SetEndCharOfAlgoLine{}
            \begin{algorithm}[H]
                \caption{Competitive Equilibria Finding}\label{alg:algorithm}
                \KwData{$A_D, V_D^{(2)}, V_D^{(3)}, S, A_S, M_{\text{demand}}, m_{\text{supply}}$}
                \KwResult{a collection $E^{(1)} \subseteq S$}
                $B \gets \left(A_D \cap A_S\right) \cup \{(0,0)\}$\\
                $\ubar{Q} \gets \sup\left(
                \{0\} \cup \left\{d\right\}_{(p,d) \in V_D^{(2)} \cap S} \cup \left\{d : \lambda\left(\left\{p: (p,d) \in B\right\}\right) > 0\right\}
                \right)$ \Comment*[r]{$\ubar{Q}$ can be $0$}
                $\bar{v} \gets \max\left(\{0\} \cup \left\{p: \left(p,\ubar{Q}\right) \in B \cap V_D^{(3)}\right\}\right)$\\
                $\ubar{S}^{(0)} \gets \left\{p: \left(p,\ubar{Q}\right) \in S \cup \{(0,0)\}\right\}
                \cap \left[0,\bar{v}\right]$\\
                $\ubar{v} \gets \min\left(
                    \{\infty\} \cup 
                    \left(
                    \left\{
                        \sup\left(\{0\} \cup \left(\ubar{S}^{(0)} - \left\{\bar{v}\right\}\right) \right), \bar{v}
                    \right\} \cap \ubar{S}^{(0)}
                    \right)
                    \right)$\Comment*[r]{$\ubar{v}$ can be $\infty$}
                \eIf{$\ubar{v} < \bar{v}$}{
                    \Return $\left\{m_{\text{supply}}\left(\ubar{v}, \ubar{Q}\right)\right\} \times M_{\text{demand}}\left(\bar{v}, \ubar{Q}\right)$ 
                }{
                    $E \gets 
                    V_D^{(3)} \cap S \cap \left(\mathbb{R}^+ \times \left[\ubar{Q},\infty\right)\right)$\Comment*[r]{$\E$ can be $\emptyset$}
                    \Return $\bigcup_{e \in E} \left(
                        \left\{m_{\text{supply}}\left(e\right)\right\} \times M_{\text{demand}}\left(e\right)
                    \right)$
                }
            \end{algorithm}

            This algorithm will return the collection of every competitive equilibrium.

            \begin{theorem} (Graphical Solutions for Equilibrium)
            \label{theorem:2}
                For any competitive market, the algorithm~\ref{alg:algorithm} will return the collection of every competitive equilibrium.
            \end{theorem}
            \begin{proof}
                W define the set $\tilde{P}$ as in the theorem~\ref{theorem:mainTheorem}.
                
                The set $\left(V_D^{(2)} \cap S\right) \cup \left(\left\{(p',d) \in S : \lambda\left(\left\{p: (p,d) \in B\right\}\right) > 0\right\}\right) = \tilde{P}$. From the corollary~\ref{corollary:19} and the the theorem~\ref{theorem:mainTheorem}, we will have that, for any equilibrium supply bid strategy $B_{\text{supply}} \in \mathcal{B}_{\text{supply}}$, $\bar{Q}\left(B_{\text{supply}}\right) \ge \ubar{Q}$.

                Note that $\ubar{v} < \bar{v}$ if and only if there does not exists any $\left(\rho', q'\right) \succ_2 \left(\hat{\rho}\left(\ubar{v}, \ubar{Q}\right), \hat{q}\left(\ubar{v}, \ubar{Q}\right)\right)$ and $\left(\rho', q'\right) \in \tilde{P}$.

                Thus, if $\ubar{v} < \bar{v}$, by the theorem~\ref{theorem:mainTheorem}, the equilibrium market supply price is unique and yields non-zero profit.

                Otherwise, every equilibrium (if exists some) gives zero profit. Thus, the market supply price and residual supply ratio will be such that the cost and supply volume is $(p,s) \in S \cap V^{(3)}_d$ with $s \ge \bar{Q}$ by the theorem~\ref{theorem:mainTheorem}.
            \end{proof}

            We can also see that if there exists an equilibrium where every mediator gets a positive differential utility then the competitive equilibrium is unique.This result is not easy to see in the theorem~\ref{theorem:mainTheorem}, but is obvious from the graphical algorithm~\ref{alg:algorithm}.
            
\part{Applications}
\label{part:application}

\section{Trading Market}
\label{section:trading}

    In a trading market, the scarcity in the system is the physical good. The suppliers are the producers, the mediators are the retails, and the demanders are the consumers.

    \paragraph{Supply Side}

        The supplier contract $\left(v,\rho,s,m\right) \in \mathbb{R}^+ \times \mathbb{R} \times \mathcal{S} \times \mathcal{M}$
        dictates that 
        \begin{itemize}
            \item the producer $s$ has to provide $v$ amount of physical good to the retail $m$;
            \item the retail $m$ has to pay $v \rho$ to the producer $s$.
        \end{itemize}
        
        The characteristics space for the producers is $\bar{\Omega}_{\text{supply}} = \mathbb{R}^+ \times \{(x,x)\}_{x \in \mathbb{R}^+}$. If a producer has a characteristics of $\left(h_0, h_1, v\right) \in \bar{\Omega}_{\text{supply}}$, such producer can produce $h_1$ amount of scarcity at the cost of $h_0$ per unit, and can bid the supply volume in the market upto $v$, which is the same as $h_1$ in this case. 
        
        The set of admissible actions for the producers is $\mathcal{A}_{\text{supply}} = \{\text{``produce"}, \text{``not produce"}\}$.

        The utility per unit for the producer function $w: \mathbb{R} \times \mathcal{A}_{\text{supply}} \times \mathbb{R}^+_0 \times \mathbb{R} \times [0,1] \to \mathbb{R} \cup \{-\infty\}$ and the utility per unit for the retail function $\tilde{w}: \mathbb{R} \times \mathcal{A}_{\text{supply}} \times \mathbb{R}^+_0 \times \mathbb{R} \times [0,1] \to \mathbb{R}$ are such that,  for any $h_0 \in \mathbb{R}$, $\alpha \in \mathcal{A}_{\text{supply}}$, $x \in \mathbb{R}^+_0$, $\rho \in \mathbb{R}$, $z \in [0,1]$,
        \begin{align*}
            w\left(h_0, \alpha, x, \rho, z\right)
            &=
            \begin{cases}
                \rho x - h_0 \mathbf{1}_{\alpha = \text{``produce"}}
                & \text{ if } x \le \mathbf{1}_{\alpha = \text{``produce"}}\\
                -\infty
                & \text{ if } x > \mathbf{1}_{\alpha = \text{``produce"}}
            \end{cases}
            \text{,}
        \end{align*}
        and
        \begin{align*}
            \tilde{w}\left(h_0, \alpha, x, \rho, z\right)
            &=
            -\rho x \mathbf{1}_{\alpha = \text{``produce"}}
            \text{.}
        \end{align*}

    \paragraph{Demand Side}
        The demander contract $\left(v,r,d,m\right) \in \mathbb{R}^+ \times \mathbb{R} \times \mathcal{D} \times \mathcal{M}$
        dictates that 
        \begin{itemize}
            \item the retail $m$ has to provide $v$ amount of physical good to the consumer $d$;
            \item the consumer $d$ has to pay $r v$ to the retail $m$.
        \end{itemize}
        
        The characteristics space for the consumers is $\bar{\Omega}_{\text{demand}} = \left(\mathbb{R}^+\right)^2$. If a consumer has a characteristics of $\left(\eta_0, \eta_1\right) \in \bar{\Omega}_{\text{demand}}$, such producer consume only a bulk of volume $\eta_1$ and gain the utility (money equivalence) of $\eta_0$ per unit. The set of admissible actions for the consumers is $\mathcal{A}_{\text{demand}} = \{\text{``consume"}, \text{``not consume"}\}$.

        The utility per unit for the consumer function $\omega: \mathbb{R} \times \mathcal{A}_{\text{demand}} \times \mathbb{R}^+_0 \times \mathbb{R} \times [0,1] \to \mathbb{R}  \cup \{-\infty\}$ and the utility per unit for the retail function $\tilde{w}: \mathbb{R} \times \mathcal{A}_{\text{producer}} \times \mathbb{R}^+_0 \times \mathbb{R} \times [0,1] \to \mathbb{R}$ are such that, for any $\eta_0 \in \mathbb{R}$, $\alpha \in \mathcal{A}_{\text{demand}}$, $x \in \mathbb{R}^+_0$, $r \in \mathbb{R}$, $z \in [0,1]$,
        \begin{align*}
            \omega\left(\eta_0, \alpha, x, r, z\right)
            &=
            \begin{cases}
                \eta_0 \mathbf{1}_{\alpha = \text{``consume"}}
                -
                r x
                & \text{ if } x \ge \mathbf{1}_{\alpha = \text{``consume"}}\\
                -\infty
                & \text{ if } x < \mathbf{1}_{\alpha = \text{``consume"}}
            \end{cases}
            \text{,}
        \end{align*}
        and
        \begin{align*}
            \tilde{\omega}\left(\eta_0, \alpha, x, r, z\right)
            &=
            r x
            \text{.}
        \end{align*}

    \subsection{Competitive Market Condition}

        \paragraph{Individual Level Conditions}

        We can see that, under dominant strategy, each producer $\left(h_0, h_1, v\right) \in \bar{\Omega}_{\text{supply}}$ will place a bid of volume $h_1$ and the lowest willing to take supply price of $h_0$, which is the production cost. The producer will choose to produce if and only if the contract is created.

        Similarly, each consumer $(\eta_0, \eta_1) \in \bar{\Omega}_{\text{demand}}$ will place a bid of volume $\eta_1$ and the highest willing to take demand price of $\eta_0$. The consumer will choose to consume, which is allowed since the contracted volume will be $\eta_1$, if and only if the contract is created.

        Thus, the conditions~\ref{condition:1},\ref{condition:2},\ref{condition:4},\ref{condition:5},\ref{condition:7},\ref{condition:8} are satisfied.

        \paragraph{Aggregated Level Conditions}

        By enforcing that the distribution of the producers characteristics $\hat{\Pi}_{\text{supply}}$ be such that $\mathbb{E}_{x \sim \hat{\Pi}_{\text{supply}}}\left[x_2\right] \in \mathbb{R}^+$, meaning that the expected amount of scarcity that can be produced is finite, then we will get that the condition~\ref{condition:3} is satisfied.

        The cost at each given supply price $\rho \in \mathbb{R}$ will then be $\rho$ if the supply volume is positive, and $0$ otherwise. Thus, the conditions~\ref{condition:9},\ref{condition:11} are satisfied.

        Moreover, since there is no producer that can produce at zero or negative cost, to acquire some amount of the scarcity, the retail has to offer a positive supply price. This satisfies the condition~\ref{condition:10} (no free supply).

        Similarly, by enforcing that the distribution of the consumers characteristics $\hat{\Pi}_{\text{demand}}$ be such that $\mathbb{E}_{x \sim \hat{\Pi}_{\text{demand}}}\left[x_2\right] \in \mathbb{R}^+$, meaning that the expected amount of possible consumption is finite, then we will get that the condition~\ref{condition:6} is satisfied.

        The revenue at each given demand price $r \in \mathbb{R}$ will then be $r$ if the demand volume is positive, and $0$ otherwise. Thus, the conditions~\ref{condition:13},\ref{condition:14},\ref{condition:15} are satisfied.

        Thus, a competitive trading market with $\left(\hat{\Pi}_{\text{supply}}, \hat{\Pi}_{\text{demand}}\right)$ with $\mathbb{E}_{x \sim \hat{\Pi}_{\text{supply}}}[x_2], \mathbb{E}_{x \sim \hat{\Pi}_{\text{demand}}}[x_2] \in \mathbb{R}^+$ will be a well-behaved competitive market.

    \subsection{Equilibrium Analysis}
    \label{subsection:tradeanalysis}

        It can be seen that
        \begin{align*}
            S
            =
            \bigcup_{\rho \in \mathbb{R}}
            \left(\rho \times \left(\tilde{S} - \{0\}\right) 
            \right)
            \text{.}
        \end{align*}

        This suggests that, for any supply price $\rho$ that can acquire some supply, the cost per unit of supply equals to the supply price itself.

        Similarly, we have that
        \begin{align*}
            D
            =
            \left\{\left(r, \max\left(\tilde{D}(r)\right) 
            \right)\right\}_{r \in \mathbb{R}}
            \cap \left(\mathbb{R} \times \mathbb{R}^+\right)
            \text{,}
        \end{align*}
        meaning that, for any demand price, where the sale could occur with non-zero volume, the revenue per unit is the same as the demand price.

        Thus, if an equilibrium exists, then it is unique, and the welfare is the same as the welfare obtained from any Walrasian equilibrium with the supply price equals the demand price.

        \subsubsection{Example}
        \label{subsubsection:example}

            Consider the case when $\bar{\Omega}_{\text{demand}} = \{1,2\} \times \{1\}$, and $\hat{\Pi}_{\text{demand}}$ be such that $\hat{\Pi}_{\text{demand}}\left(\{(1,1)\}\right) = 0.6$.

            \begin{figure}
                \begin{subfigure}{.5\textwidth}
                  \centering
                  \includegraphics[width=.9\linewidth]{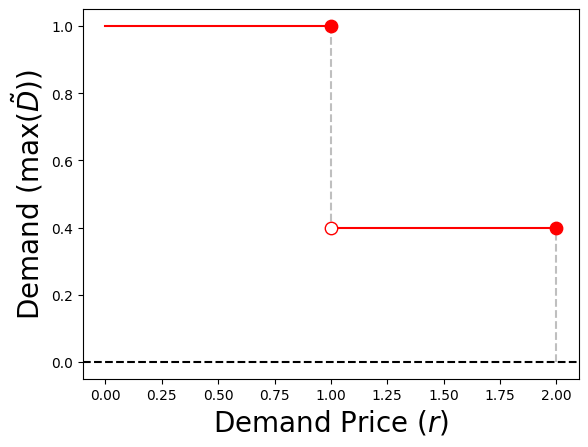}
                  \caption{Maximum Real Demand}
                \end{subfigure}%
                \begin{subfigure}{.5\textwidth}
                  \centering
                  \includegraphics[width=.9\linewidth]{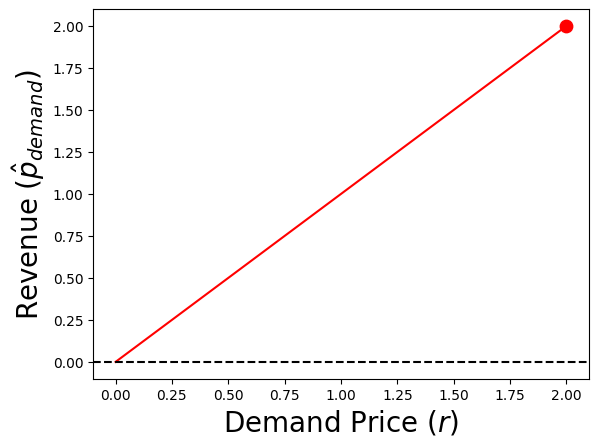}
                  \caption{Demand Revenue per unit}
                \end{subfigure}
                \caption{Trading Market: Example maximum real demand $\max\left(\tilde{D}(r)\right)$ and demand revenue per unit $\hat{p}_{\text{demand}}(r)$ as a function of demand price $r$ over the region where demand is non-zero used in the equilibrium analysis example (subsubsection~\ref{subsubsection:example}).}
            \end{figure}

            \begin{figure}
                \begin{subfigure}{.5\textwidth}
                  \centering
                  \includegraphics[width=.9\linewidth]{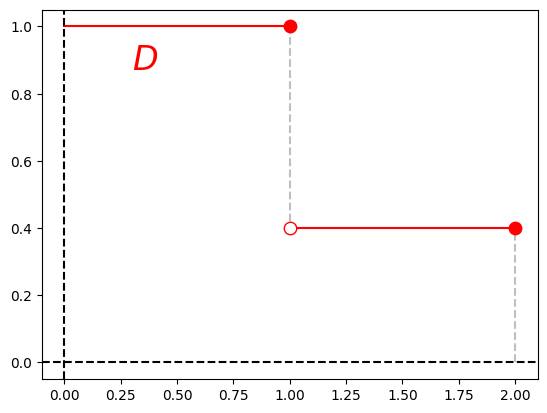}
                  \caption{$D$}
                \end{subfigure}%
                \begin{subfigure}{.5\textwidth}
                  \centering
                  \includegraphics[width=.9\linewidth]{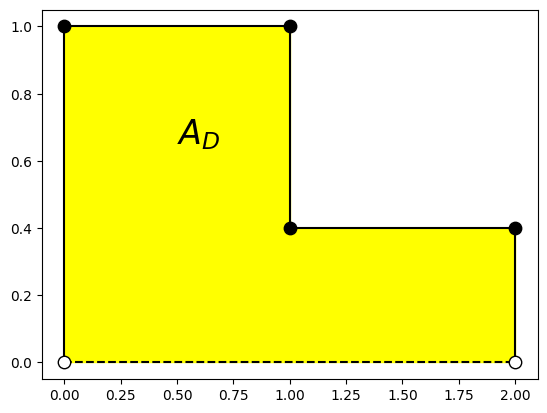}
                  \caption{$A_D$}
                \end{subfigure}
                \begin{subfigure}{.5\textwidth}
                  \centering
                  \includegraphics[width=.9\linewidth]{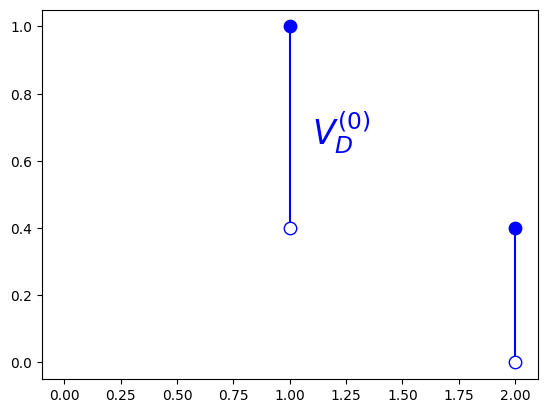}
                  \caption{$V_D^{(0)}$}
                \end{subfigure}%
                \begin{subfigure}{.5\textwidth}
                  \centering
                  \includegraphics[width=.9\linewidth]{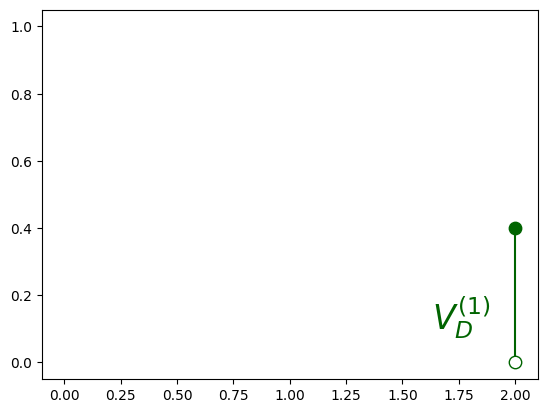}
                  \caption{$V_D^{(1)}$}
                \end{subfigure}
                \begin{subfigure}{.5\textwidth}
                  \centering
                  \includegraphics[width=.9\linewidth]{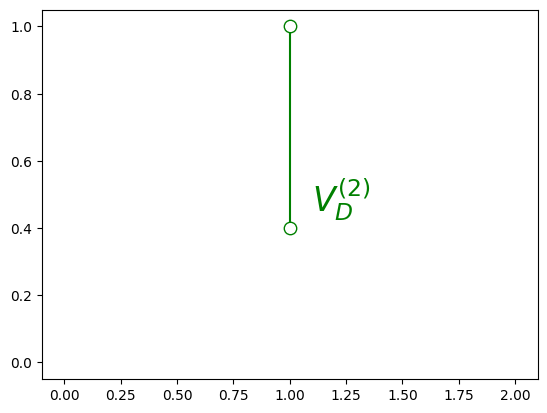}
                  \caption{$V_D^{(2)}$}
                \end{subfigure}%
                \begin{subfigure}{.5\textwidth}
                  \centering
                  \includegraphics[width=.9\linewidth]{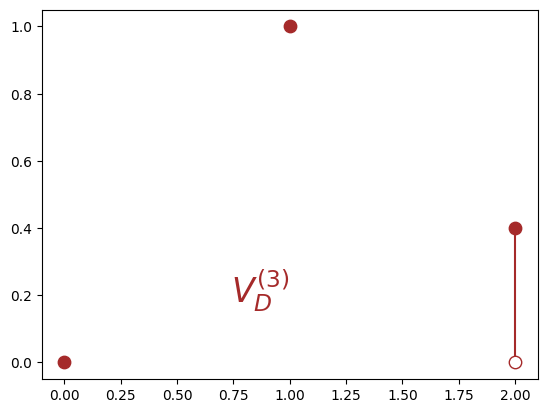}
                  \caption{$V_D^{(3)}$}
                \end{subfigure}
                \caption{Trading Market: The example construction of the demand graph $D$, the augmented demand graph $A_D$, the vertical border demand graph $V_D^{(0)}$, the farthest vertical border demand graph $V_D^{(1)}$, the sharp vertical border demand graph $V_D^{(2)}$, and the admissible vertical border demand graph $V_D^{(3)}$ used in the equilibrium analysis example (subsubsection~\ref{subsubsection:example}). Note that $D$ looks similar to the maximum real demand.}
            \end{figure}

            Consider the case when, for some $v \in \mathbb{R}^+$, $\bar{\Omega}_{\text{supply}} = \left(0,\frac{5}{4}\right] \times \{v\} \times \{v\}$, and $\hat{\Pi}_{\text{supply}}$ be such that, for any $A \in \mathcal{B}\left(\left(0,\frac{5}{4}\right]\right)$,
            \begin{align*}
                \hat{\Pi}_{\text{supply}}
                \left(
                    A \times \{v\} \times \{v\}
                \right)
                = 
                \frac{4 \lambda(A)}{5}
                \text{.}
            \end{align*}

            It can be seen that, if $v=0.8$, then there does not exist any equilibrium. If $v=0.43$, then there exist a unique equilibrium, and such equilibrium does not entail rationing. If $v=0.24$, there exist a unique equilibrium, and such equilibrium entails rationing.

            \begin{figure}
                \centering
                \includegraphics[width=.8\linewidth]{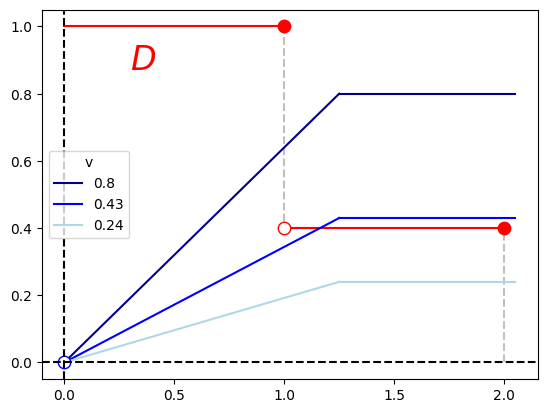}
                \caption{Trading Market: The example used in the equilibrium analysis example (subsubsection~\ref{subsubsection:example}) with $v$ chosen to be $0.8$, $0.43$, or $0.24$.}
            \end{figure}

\section{Credit Market}
\label{section:credit}

    In a credit market, the scarcity in the system is the loanable fund. The suppliers are the depositors, the mediators are the banks, and the demanders are the entrepreneurs. There exists two time periods, denoted as $t_1$ and $t_2$ with $t_1 < t_2$. The supply price is denoted as the deposit interest rate, and the demand price is denoted as the loan interest rate.
    
    \paragraph{Supply Side}

        The supplier (deposit) contract $\left(v,\rho,s,m\right) \in \mathcal{R}^+ \times \mathcal{R} \times \mathcal{S} \times \mathcal{M}$
        dictates that 
        \begin{itemize}
            \item the depositor $s$ has to provide $v$ amount of loanable fund to the bank $m$ at time $t_1$;
            \item the bank $m$ has to pay $(1+\rho)v$ to the depositor $s$ at time $t_2$.
        \end{itemize}
        
        The characteristics space for the producers is $\bar{\Omega}_{\text{supply}} = \mathbb{R}^+ \times \{(x,x)\}_{x \in \mathbb{R}^+}$. If a producer has a characteristics of $\left(h_0, h_1, v\right) \in \bar{\Omega}_{\text{supply}}$, such producer can deposit $h_1$ amount of loanable fund at the average cost\footnote{The cost can be a gas price or an opportunity cost} of $h_0$, and can bid the supply volume in the market upto $v$, which is the same as $h_1$ in this case. 
        
        The set of admissible actions for the producers is $\mathcal{A}_{\text{supply}} = \{\text{``deposit"}, \text{``not deposit"}\}$.

        The utility per unit for the depositor function $w: \mathbb{R} \times \mathcal{A}_{\text{supply}} \times \mathbb{R}^+_0 \times \mathbb{R} \times [0,1] \to \mathbb{R} \cup \{-\infty\}$ and the utility per unit for the bank function $\tilde{w}: \mathbb{R} \times \mathcal{A}_{\text{supply}} \times \mathbb{R}^+_0 \times \mathbb{R} \times [0,1] \to \mathbb{R}$ are such that,  for any $h_0 \in \mathbb{R}$, $\alpha \in \mathcal{A}_{\text{supply}}$, $x \in \mathbb{R}^+_0$, $\rho \in \mathbb{R}$, $z \in [0,1]$,
        \begin{align*}
            w\left(h_0, \alpha, x, \rho, z\right)
            &=
            \begin{cases}
                (1+\rho) x - (1+h_0) \mathbf{1}_{\alpha = \text{``deposit"}}
                & \text{ if } x \le \mathbf{1}_{\alpha = \text{``deposit"}}\\
                -\infty
                & \text{ if } x > \mathbf{1}_{\alpha = \text{``deposit"}}
            \end{cases}
            \text{,}
        \end{align*}
        and
        \begin{align*}
            \tilde{w}\left(h_0, \alpha, x, \rho, z\right)
            &=
            - (1+\rho) x \mathbf{1}_{\alpha = \text{``deposit"}}
            \text{.}
        \end{align*}

    \paragraph{Demand Side}
        The demander contract $\left(v,r,d,m\right) \in \mathbb{R}^+ \times \mathbb{R} \times \mathcal{D} \times \mathcal{M}$
        dictates that 
        \begin{itemize}
            \item the bank $m$ has to provide a principal value $v$ of loanable fund to the entrepreneur $d$ at time $t_1$;
            \item the entrepreneur $d$ has to pay back the principal and the interest $(1+r) v$ or as much as one has to the bank $m$.
        \end{itemize}
        
        The characteristics space for the entrepreneurs is $\bar{\Omega}_{\text{demand}} = \left((0,1) \times \Omega\right) \times \mathbb{R}^+$ for some $n \in \mathbb{N}$. If an entrepreneur has a characteristics of $\left(\eta_0, \eta_1\right) \in \bar{\Omega}_{\text{demand}}$, by denoting $\left(\eta_{0,1}, \eta_{0,2}\right) = \eta_0$ with $\eta_{0,1} \in (0,1)$ and $\eta_{0,2} \in \Omega$, such entrepreneur will have a single project requiring a budget of $\eta_1$ to invest. Currently, the entrepreneur has an initial wealth of $\eta_{0,1}\eta_1 < \eta_1$, so an additional loanable fund of at least $(1-\eta_{0,1})\eta_1 > 0$ is required for the investment to be an option.
        
        Let consider a collection of random variables $\left(X^{\left(\omega\right)}\right)_{\omega \in \Omega}$ such that, for each $\omega \in \Omega$, $X^{\left(\omega\right)} \in \Delta\left(\mathbb{R}^+_0\right)$, and $\mathbb{E}_{x \sim X^{\left(\omega\right)}}[x] > 1$, meaning that each possible project has positive net profit.

        The entrepreneur has a knowledge that the project payoff would be distributed according to $\eta_1 X^{\left(\eta_{0,2}\right)}$, making it a random non-negative variable. Thus, if the entrepreneur acquires a loanable fund of size $x \eta_1 \ge (1-\eta_{0,1})\eta_1$ and decides to invest, the final wealth if investment occurs is a random non-negative variable
        \begin{align*}
            F_{\text{invest}} = \left[x \eta_1 - (1-\eta_{0,1})\eta_1\right] 
            + \eta_1 X^{\left(\eta_{0,2}\right)}
            \text{,}
        \end{align*}
        while the final wealth if investment does not occur is a positive number $f_{\text{not invest}} = x\eta_1 + \eta_{0,1} \eta_1$. After the final wealth is realized to be a non-negative number $f$, the payment of $\min\left(\left\{f, (1+r)x\eta_1\right\}\right)$ has to be given to the bank.

        The set of admissible actions for the consumers is $\mathcal{A}_{\text{demand}} = \{\text{``invest"}, \text{``not invest"}\}$.

        The utility per unit for the entrepreneur function $\omega: \left((0,1) \times \Omega\right) \times \mathcal{A}_{\text{demand}} \times \mathbb{R}^+_0 \times \mathbb{R} \times [0,1]  \to \mathbb{R} \cup \{-\infty\}$ and the utility per unit for the bank function $\tilde{w}: \left((0,1) \times \Omega\right) \times \mathcal{A}_{\text{producer}} \times \mathbb{R}^+_0 \times \mathbb{R} \times [0,1] \to \mathbb{R}$ are such that, for any $\eta_0 \in (0,1) \times \Omega$, $\alpha \in \mathcal{A}_{\text{demand}}$, $x \in \mathbb{R}^+_0$, $r \in \mathbb{R}$, $z \in [0,1]$,
        \begin{align*}
            \omega\left(\eta_0, \alpha, x, r, z\right)
            &=
            \begin{cases}
                \max\left(\left\{
                    0
                    ,\eta_{0,1} + \mathbf{1}_{\alpha = \text{``invest"}}\left(\text{Quantile}_{X^{\left(\eta_{0,2}\right)}}(z) - 1\right)-rx
                \right\}\right)
                & \text{ if } x + \eta_{0,1} \ge \mathbf{1}_{\alpha = \text{``invest"}}\\
                -\infty
                & \text{ if } x + \eta_{0,1}  < \mathbf{1}_{\alpha = \text{``invest"}}
            \end{cases}
            \text{,}
        \end{align*}
        and
        \begin{align*}
            \tilde{\omega}\left(\eta_0, \alpha, x, r, z\right)
            &=
            \min\left(\left\{
                x + \eta_{0,1} + \mathbf{1}_{\alpha = \text{``invest"}}\left(\text{Quantile}_{X^{\left(\eta_{0,2}\right)}}(z) - 1\right)
                ,(1+r)x
            \right\}\right)
            \text{.}
        \end{align*}

    \subsection{Competitive Market Condition}

        \paragraph{Individual Level Conditions}

        We can see that, under dominant strategy, each depositor $\left(h_0, h_1, v\right) \in \bar{\Omega}_{\text{supply}}$ will place a bid of volume $h_1$ and the lowest willing to take supply price of $h_0$. The depositor will choose to deposit if and only if the contract is created.

        Similarly, each entrepreneur $(\eta_0, \eta_1) \in \bar{\Omega}_{\text{demand}}$ will place a bid of volume $\left(1 - \eta_{0,1}\right)\eta_1$ and the highest willing to take demand price of $\bar{r}(\eta_0)$, where the function $\bar{r}: \Omega_{\text{demand}} \to \mathbb{R}$ is such that, for any $\eta_0 \in \Omega$,
        \begin{align*}
            0 = \mathbb{E}_{y \sim X^{\left(\eta_{1,1}\right)}}\left[
            \max\left(\left\{-\eta_{0,1}, y -r \left(1-\eta_{0,1}\right) \eta_{0,1}\right\}\right) \right]
            \text{.}
        \end{align*}
        
        The entrepreneur will choose to invest if and only if the contract is created.

        Thus, the conditions~\ref{condition:1},\ref{condition:2},\ref{condition:4},\ref{condition:5},\ref{condition:7},\ref{condition:8} are satisfied.

        \paragraph{Aggregated Level Conditions}
            
            Similar to the behavior of the supply in a  trading market, by assuming that the distribution of the producers characteristics $\hat{\Pi}_{\text{supply}}$ be such that $\mathbb{E}_{x \sim \hat{\Pi}_{\text{supply}}}\left[x_2\right] \in \mathbb{R}^+$, we will get that the condition~\ref{condition:3}.\ref{condition:9},\ref{condition:10},\ref{condition:11} are satisfied. 

            Similarly, by assuming that the distribution of the consumers characteristics $\hat{\Pi}_{\text{demand}}$ be such that $\mathbb{E}_{x \sim \hat{\Pi}_{\text{demand}}}\left[x_2\right] \in \mathbb{R}^+$, we will get that the condition~\ref{condition:6} is satisfied.

            The finiteness on the loss and gain for the loan contract of finite size and a loan interest rate $r\in \mathbb{R}$ ensures that the condition~\ref{condition:13} is satisfied.

            The proposition~\ref{proposition:randomfunction2} ensures that the condition~\ref{condition:14} is satisfied. We can also further show that the condition~\ref{condition:15} is satisfied.

            Thus, a competitive credit market with $\left(\hat{\Pi}_{\text{supply}}, \hat{\Pi}_{\text{demand}}\right)$ with $\mathbb{E}_{x \sim \hat{\Pi}_{\text{supply}}}[x_2], \mathbb{E}_{x \sim \hat{\Pi}_{\text{demand}}}[x_2] \in \mathbb{R}^+$ will be a well-behaved competitive market.

    \subsection{Equilibrium Analysis}

        In this subsection, we will restrict the attention when, for any $x,y \in \Omega_{\text{demand}}$ then the random variable $X^{(y)}$ is first-order stochastically dominated by, or first-order stochastically dominates $X^{(x)}$, which is introduced in the paper ``Credit Rationing in Markets with Imperfect Information", Stiglitz \& Wiess, 1981~\cite{stiglitz}. 

        Although the structure between $\tilde{D}$ and $D$ are not apparent as in a competitive trading market case, we can see that, for this type of competitive credit markets, the equilibrium always exists, but it may not be unique. In the graphical approach, it is trivial to show that the set $\left\{d \in \mathbb{R}_0^+ : V_D^{(3)} \cap \left(\mathbb{R} \times \{d\}\right) \ne \emptyset \right\}$ is a connected set. This condition will not be held in a competitive trading as can be seen in the example used in the subsubsection~\ref{subsubsection:example}.

        The proof used in the paper ``Credit Rationing in Markets with Imperfect Information", Stiglitz \& Wiess, 1981~\cite{stiglitz} requires the existence of a Walrasian equilibrium (as defined in the paper). However, our graphical approach can show that, without the existence of a Walrasian equilibrium, the competitive equilibrium still exist.
        
        Other results in the literature can also be reformulated into the graphical result. The paper ``On the Possibility of Credit Rationing in the Stiglitz-Weiss Model", Arnold \& Riley, 2009~\cite{arnold} shows that, if there exists some $K \in \mathbb{R}^+$ such that $X^{(x)} \le K$ surely for every $x \in \Omega_{\text{demand}}$, then $V_D^{(1)} \subseteq \bigcup_{r \in \mathbb{R}}\left(\left\{\hat{p}_{\text{demand}}(r)\right\} \times \tilde{D}(r)\right)$. The paper ``Unbounded returns and the possibility of credit rationing: A note on the Stiglitz–Weiss and Arnold–Riley models", Lu \& Rong, 2018~\ref{lu} shows an example where, if there does not exits such $K \in \mathbb{R}^+$, then it is possible that $V_D^{(1)} \not\subseteq \bigcup_{r \in \mathbb{R}}\left(\left\{\hat{p}_{\text{demand}}(r)\right\} \times \tilde{D}(r)\right)$.

        Another important result is that we can easily construct a case when there could be an equilibrium with (countably) infinitely many demand price $r$'s observed in the market (with non-zero probability), which would be shown in the subsubsection~\ref{subsubsection:infinite}. Moreover, we can make the support of the demand price not compact. This, therefore, suggests that the uniformly bounded distribution in the paper~\cite{arnold} has to be violated.

        \subsubsection{Example}
        \label{subsubsection:examplecredit}

            We consider the case when $\bar{\Omega}_{\text{demand}} = \left\{\frac{1}{2}\right\} \times \left\{1,2\right\} \times \left\{2\right\}$, the distribution $\hat{\Pi}_{\text{demand}}$ be such that $\hat{\Pi}_{\text{demand}}\left(\{\left(\frac{1}{2},1,2\right)\}\right) = 0.95$, and $\left(X_{\omega}\right)_{\omega \in \{1,2\}}$ be such that
            \begin{itemize}
                \item $X_1$ is a non-negative random variable with its support being $\{0,10\}$ and $\mathbb{P}\left(X_1 = 10\right) = 0.4$;
                \item $X_2$ is a non-negative random variable with its support being $\{0,20\}$ and $\mathbb{P}\left(X_1 = 20\right) = 0.2$.
            \end{itemize}
            
            \begin{figure}
                \begin{subfigure}{.5\textwidth}
                  \centering
                  \includegraphics[width=.9\linewidth]{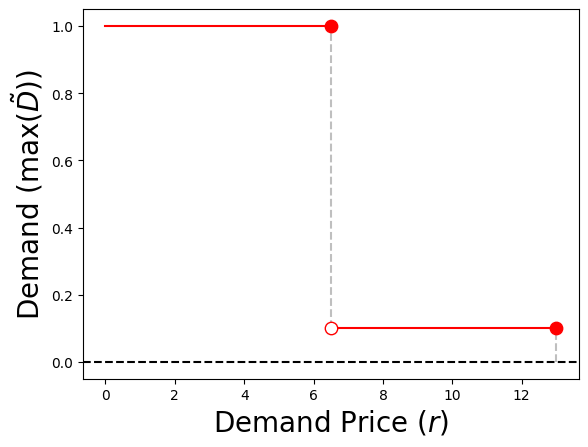}
                  \caption{Maximum Real Demand}
                \end{subfigure}%
                \begin{subfigure}{.5\textwidth}
                  \centering
                  \includegraphics[width=.9\linewidth]{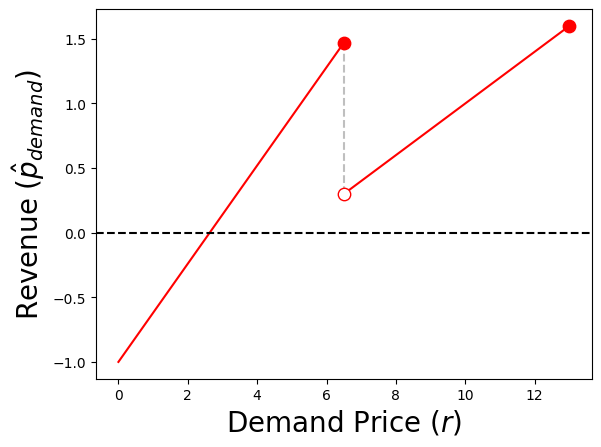}
                  \caption{Demand Revenue per unit}
                \end{subfigure}
                \caption{Credit Market: Example maximum real demand $\max\left(\tilde{D}(r)\right)$ and demand revenue per unit $\hat{p}_{\text{demand}}(r)$ as a function of demand price $r$ over the region where demand is non-zero used in the equilibrium analysis example (subsubsection~\ref{subsubsection:examplecredit}).}
            \end{figure}

            It can be seen that the set $\left\{d \in \mathbb{R}_0^+ : V_D^{(3)} \cap \left(\mathbb{R} \times \{d\}\right) \ne \emptyset \right\}$ is $\left[0,1\right]$. Together with the maximal monotonicity of the real supply $\tilde{S}$, this suggests that an equilibrium always exists.
            
            \begin{figure}
                \begin{subfigure}{.5\textwidth}
                  \centering
                  \includegraphics[width=.9\linewidth]{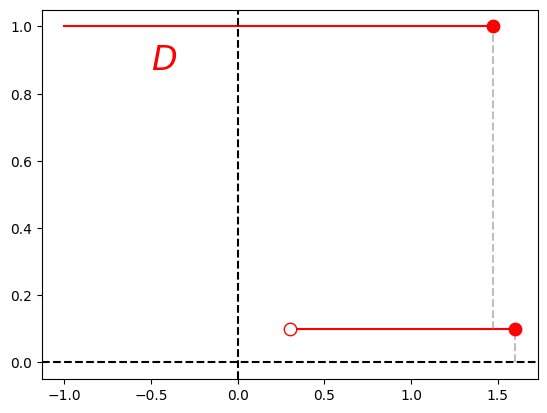}
                  \caption{$D$}
                \end{subfigure}%
                \begin{subfigure}{.5\textwidth}
                  \centering
                  \includegraphics[width=.9\linewidth]{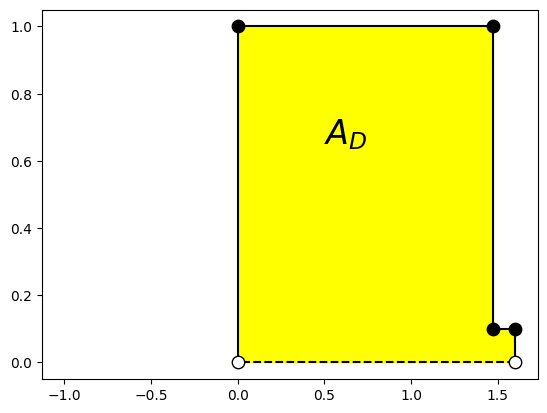}
                  \caption{$A_D$}
                \end{subfigure}
                \begin{subfigure}{.5\textwidth}
                  \centering
                  \includegraphics[width=.9\linewidth]{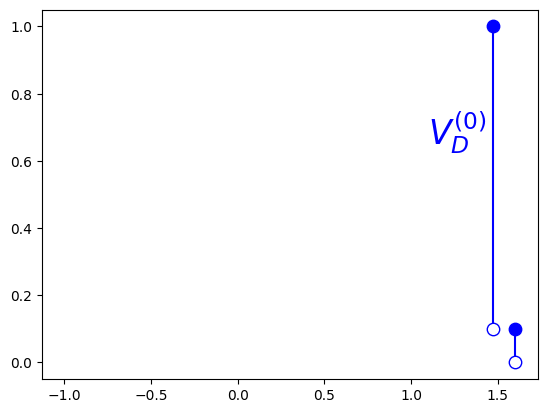}
                  \caption{$V_D^{(0)}$}
                \end{subfigure}%
                \begin{subfigure}{.5\textwidth}
                  \centering
                  \includegraphics[width=.9\linewidth]{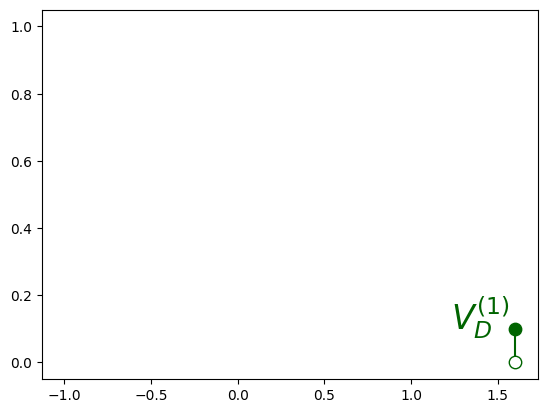}
                  \caption{$V_D^{(1)}$}
                \end{subfigure}
                \begin{subfigure}{.5\textwidth}
                  \centering
                  \includegraphics[width=.9\linewidth]{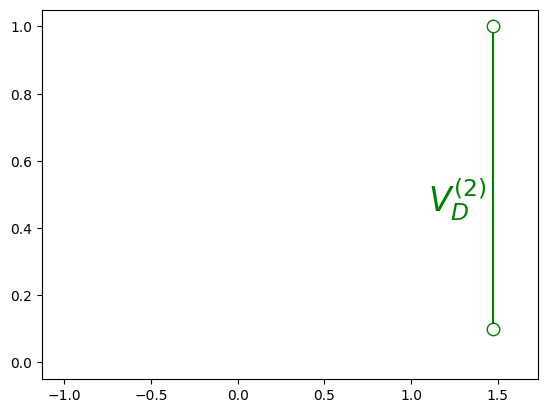}
                  \caption{$V_D^{(2)}$}
                \end{subfigure}%
                \begin{subfigure}{.5\textwidth}
                  \centering
                  \includegraphics[width=.9\linewidth]{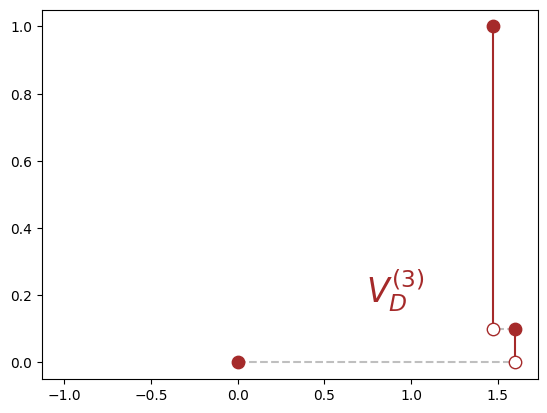}
                  \caption{$V_D^{(3)}$}
                \end{subfigure}
                \caption{Credit Market: The example construction of the demand graph $D$, the augmented demand graph $A_D$, the vertical border demand graph $V_D^{(0)}$, the farthest vertical border demand graph $V_D^{(1)}$, the sharp vertical border demand graph $V_D^{(2)}$, and the admissible vertical border demand graph $V_D^{(3)}$ used in the equilibrium analysis example (subsubsection~\ref{subsubsection:examplecredit}).}
            \end{figure}

            Since there are only $2$ types of entrepreneurs in this example market, the paper~\cite{arnold} 
            suggests that $V_D^{(1)} \subseteq \bigcup_{r \in \mathbb{R}}\left(\left\{\hat{p}_{\text{demand}}(r)\right\} \times \tilde{D}(r)\right)$.

        \subsubsection{Infinite Demand Price Support Example}
        \label{subsubsection:infinite}

            We consider the case when $\bar{\Omega}_{\text{demand}} = \left\{\frac{1}{2}\right\} \times \mathbb{N} \times \left\{2\right\}$, the distribution $\hat{\Pi}_{\text{demand}}$ such that, for any $i \in \mathbb{N}$,
            \begin{align*}
                \hat{\Pi}_{\text{demand}}\left(\left\{\left(\frac{1}{2},1,2\right)\right\}\right)
                =
                2^{-i}
                \text{,}
            \end{align*}
            and the collection $\left(X_{i}\right)_{i \in \mathbb{N}}$ be such that, for any $i \in \mathbb{N}$, the random variable $X_i$ has a support of $\left\{0, 3 \cdot 2^{i}\right\}$, and
            \begin{align*}
                \mathbb{P}\left(X_i = 3 \cdot 2^{i}\right)
                =
                2^{-i}
                \text{.}
            \end{align*}

            We will then have that $D$ will consist of countably infinitely many horizontal straight lines. However, the graphical analysis will still hold since
            \begin{itemize}
                \item The augmented demand graph
                \begin{align*}
                    A_D = \left[0, \frac{2}{3}\right] \times \left(0,1\right]\text{;}
                \end{align*}
                \item The vertical border demand graph
                \begin{align*}
                    V_D^{(0)} = \left\{ \frac{2}{3}\right\} \times \left(0,1\right]\text{;}
                \end{align*}
                \item The farthest vertical border demand graph
                \begin{align*}
                    V_D^{(1)} = \left\{ \frac{2}{3}\right\} \times \left(0,1\right]\text{;}
                \end{align*}
                \item The sharp vertical border demand graph
                \begin{align*}
                    V_D^{(2)} = \emptyset\text{;}
                \end{align*}
                \item The admissible vertical border demand graph
                \begin{align*}
                    V_D^{(3)} = \left(\left\{ \frac{2}{3}\right\} \times \left(0,1\right]\right) \cup \{(0,0)\} \text{.}
                \end{align*}
            \end{itemize}

            Similarly, consider the case when, for some $v \in \mathbb{R}^+$, $\bar{\Omega}_{\text{supply}} = \left(0,\frac{2}{3}\right] \times \{v\} \times \{v\}$, and $\hat{\Pi}_{\text{supply}}$ be such that, for any $A \in \mathcal{B}\left(\left(0,\frac{2}{3}\right]\right)$,
            \begin{align*}
                \hat{\Pi}_{\text{supply}}
                \left(
                    A \times \{v\} \times \{v\}
                \right)
                = 
                \frac{3 \lambda(A)}{2}
                \text{.}
            \end{align*}

            We will then have that
            \begin{itemize}
                \item If $v \ge 1$, then the equilibrium uniquely exists, with the same $\hat{\mu}_{\text{demand}}$ such that the support is $\{4\}$ and $\hat{\mu}_{\text{demand}}(\{4\}) = 1$. This suggests that there will only be a single observable demand price (interest rate) being $4$ in the market;
                \item If $v \in (0,1) $, then the equilibrium exists and there exist uncountably infinitely many equilibria. Moreover, there always exist uncountably infinitely many equilibria such that the supports of $\hat{\mu}_{\text{demand}}$'s are the same and are
                \begin{align*}
                    \left\{2^{i+1}\right\}_{i=1}^{\infty}
                    \text{.}
                \end{align*}
                This suggests that it is possible to see countably infinitely many demand price (interest rate) in a single equilibrium.
            \end{itemize}

            With a construction similar to that being used in the paper~\cite{lu}, it can be easily shown that the occurrence of having (countably) infinitely many demand prices in the support (with non-zero probability of observing) can also be supported even when the distribution of distribution the projects payoff is continuous (with respect to some distance metric of distribution).

{\small
\bibliographystyle{unsrt}
\bibliography{egbib}}

\end{document}